\documentclass[12pt]{article}
\usepackage{amsmath}
\usepackage{amsfonts}
\usepackage{mathrsfs}
\newcommand{\Section}[1]%
{\section{#1}\setcounter{equation}{0}%
\setcounter{theorem}{0}}
\newtheorem{theorem}{Theorem}
\newtheorem{lemma}[theorem]{Lemma}

\newtheorem{pro}[theorem]{Proposition}

\newtheorem{assumption}[theorem]{Assumption}

\def\co{\mathbb{C}}

\def\ze{\mathbb{Z}}

\newenvironment{proof}[1]%
{\par\noindent{\em #1:\ }}%
{~\rule{2mm}{2mm}\par\bigskip}

\begin{document}
\newpage\thispagestyle{empty}
{\topskip 2cm
\begin{center}
{\Large\bf Stability of the Spectral Gap for Lattice Fermions\\} 
\bigskip\bigskip
{\Large Tohru Koma\footnote{\small \it Department of Physics, Gakushuin University, Mejiro, Toshima-ku, Tokyo 171-8588, JAPAN,
{\small\tt e-mail: tohru.koma@gakushuin.ac.jp}}
\\}
\end{center}
\vfil
\noindent
{\bf Abstract:} It is widely believed that the spectral gap 
above the Fermi sea ground state of lattice free fermions is stable against generic weak interactions.  
Quite recently, Hastings presented a new interesting idea to prove the stability for Majorana free fermions. 
He ingeniously used Majorana fermions so that the unperturbed Hamiltonian of the free fermion part 
is mapped into a frustration-free form. 
In this paper, we refine Hastings's argument, and give a complete proof of the stability of 
the spectral gap above the Fermi sea ground state against generic weak interactions 
for generic lattice free fermions.
\par
\noindent
\vfil}


\Section{Introduction} 

Quite recently, Hastings \cite{Hastings} presented a new interesting idea to prove the stability of the spectral gap 
above the Fermi sea ground state of lattice free fermions against generic weak interactions. 
It ingeniously uses Majorana fermions so that the unperturbed Hamiltonian of the free fermion part 
is mapped into a frustration-free form. 
Although the stability is widely believed to hold, there has been no proof so far. 
In this paper, we refine Hastings's argument, and give a complete proof of the statement. 
Although Hastings started from a Hamiltonian written in terms of Majorana operators, 
we will begin with a lattice fermion system which conserves the charge. 
In particular, we do not require the translational invariance of the system.

We organize the present paper as follows. In Sec.~\ref{Sec:ModelResult}, we describe the main result 
for the present fermion system. In Sec.~\ref{Sec:TransFree}, the unperturbed Hamiltonian 
of the free fermion part is mapped into a frustration-free form. 
Under the transformation of the fermion operators, the interactions between the fermions 
is also transformed into a certain form which satisfies local boundedness. 
In Sec.~\ref{Sec:conneGS}, following the idea of Hastings \cite{Hastings,{BravHast}}, 
we will further transform the total Hamiltonian with the interactions so that 
all the local interactions annihilate the ground state of the unperturbed Hamiltonian. 
This property enables us to obtain a relative bound for the perturbation in Sec.~\ref{Sec:relativeBound}.
The proof of the main Theorem~\ref{mainTheorem} below is given in Sec.~\ref{Sec:ProofMainTheorem}. 
Appendices~\ref{append:boundedtildeV}--\ref{Sec:BoundtildecalWZ3s} are devoted to technical estimates. 

\Section{Model and result} 
\label{Sec:ModelResult}

Let $\Lambda$ be a finite sublattice of the $d$-dimensional hypercubic lattice $\ze^d$. 
Firstly, we introduce a metric on the lattice $\Lambda$. 
Let $x_0,x_1,x_2,\ldots,x_n$ be a sequence of sites in the lattice $\Lambda$. 
If the two sites, $x_{j-1}$ and $x_j$, are nearest neighbours for all $j=1,2,\ldots,n$, 
then we call the sequence the path. We also say that the path, $\{x_0,x_1,\ldots,x_n\}$, has the length $n$ 
and connects $x_0$ to $x_n$.  For $x,y\in\Lambda$, we denote by ${\rm dist}(x,y)$ 
the distance which is defined to be the shortest path length 
that is needed to connect $x$ to $y$. When the periodic boundary condition is imposed for the finite lattice $\Lambda$, 
the pairs of nearest neighbours must be defined to be compatible with the periodic boundary condition. 

Consider a discrete Schr\"odinger operator $h$ on the lattice $\Lambda$. 
The Hamiltonian operator $h$ acts on a complex-valued wavefunction $\varphi$ by  
\begin{equation}
(h\varphi)(x)=\sum_{y\in\Lambda}t_{x,y}\varphi(y)\quad \mbox{for \ } x\in \Lambda,
\end{equation}
where the wavefunction $\varphi$ is defined on the finite lattice $\Lambda$, 
and the hopping amplitudes, $t_{x,y}\in\co$, 
satisfy $\overline{t_{y,x}}=t_{x,y}$ for $x,y\in\Lambda$ so that $h$ is self-adjoint.  
Here, $\overline{\cdots}$ denotes the complex conjugate. 

\begin{assumption}
\label{assump:texpdecay}
We assume that there exist positive constants, $C_t$ and $m_t$, which are independent of the size $|\Lambda|$ of 
the lattice such that 
\begin{equation}
\label{texpdecay}
|t_{x,y}|\le C_te^{-m_t{\rm dist}(x,y)}.
\end{equation}
For the spectrum $\sigma(h)$ of the Hamiltonian $h$, we assume that there exists a lower bound $\Delta E>0$ for the spectral gap 
at a Fermi energy $\mathcal{E}_{\rm F}$,  
i.e., the distance between them satisfies ${\rm dist}(\sigma(h),\mathcal{E}_{\rm F})\ge \Delta E$ 
with a positive constant $\Delta E$ which is independent of the size $|\Lambda|$ of the lattice. 
\end{assumption}

In order to describe the corresponding system of many fermions, we introduce the creation and annihilation operators 
for the fermions, which obey the canonical anticommutation relations, 
\begin{equation}
\label{CACR}
\{a_x,a_y^\dagger\}=\delta_{x,y},\quad \{a_x,a_y\}=0, \quad \mbox{and}\quad \{a_x^\dagger,a_y^\dagger\}=0,
\end{equation}
where $\delta_{x,y}$ is Kronecker's delta. The Hamiltonian on the fermion Fock space is given by 
\begin{equation}
\label{calHam}
{H}_0=\sum_{x,y\in\Lambda}a_x^\dagger (t_{x,y}-\mathcal{E}_{\rm F}\delta_{x,y})a_y.
\end{equation}
The ground state is given by filling the states below the Fermi energy $\mathcal{E}_{\rm F}$ with fermions, 
and choosing the empty states above the Fermi energy. 

The total Hamiltonian is given by 
\begin{equation}
H:=H_0+V,
\end{equation}
where the interaction Hamiltonian $V$ is written 
\begin{equation}
V=\sum_{X\subset \Lambda} V_X.
\end{equation}
Each term $V_X$ has the support $X$ and the even parity for the fermion operators. 
Namely, the local interaction $V_X$ is a polynomial which consists of monomials with an even number of 
the fermion operators.  

\begin{assumption}
\label{AssumptionVXfinite}
We assume that $V_X=0$ if $|X|> n^{\rm max}$ for a positive integer $n^{\rm max}$. 
Here, $|X|$ is the cardinality of the set $X$. 
We also assume that there exist positive constants, $C_{V}$ and $m_{V}$, 
which are independent of the size $|\Lambda|$ of the lattice such that 
the local interactions $V_X$ satisfy 
\begin{equation}
\label{assumpVfinite}
\sum_{X\subset\Lambda :\; x,y\in X}\Vert V_X\Vert<C_{V}e^{-m_{V}{\rm dist}(x,y)}.
\end{equation}
\end{assumption}

When the condition that $V_X=0$ if $|X|> n^{\rm max}$ for a positive integer $n^{\rm max}$ does not holds,  
we need to introduce a positive constant $K_h$ as in Assumption~\ref{AssumptionVX} below. 
The role of the constant $K_h$ is to suppress 
the strength of the coupling constants of the many-fermion interactions whose number of the fermions infinitely increases. 

\begin{assumption}
\label{AssumptionVX}
We assume that there exist positive constants, $C_{h,V}$ and $m_{h,V}$, 
which are independent of the size $|\Lambda|$ of the lattice such that 
the local interactions $V_X$ satisfy 
\begin{equation}
\label{assumpV}
\sum_{X\subset\Lambda :\; x,y\in X}(K_h)^{|X|}\Vert V_X\Vert<C_{h,V}e^{-m_{h,V}{\rm dist}(x,y)}
\end{equation}
with a positive constant $K_h>1$.  
\end{assumption}

The constant $K_h>1$ is determined by the single fermion Hamiltonian $h$ so that the statement of Theorem~\ref{mainTheorem} 
below holds. 
Under the assumption about the spectral gap of the single fermion Hamiltonian $h$, 
the constant $K_h$ can be chosen to be finite and independent of the size $|\Lambda|$ of the lattice. More precisely, 
the constant $K_h$ is determined by the decay of the matrix elements of the Fermi-sea projection with large distance. 
(See Appendix~\ref{append:boundedtildeV} for the details.) 
Roughly speaking, the upper bound of the right-hand side of (\ref{assumpV}) implies that 
the interactions $V_X$ decay exponentially 
with the large size of the diameter of the support of $V_X$. The factor $(K_h)^{|X|}$ implies that 
the interactions $V_X$ decay exponentially with the large number $|X|$ of the fermion operators. 

In order to control the strength of the interaction $V$, we introduce a real parameter $s$, and write 
\begin{equation}
H_s:=H_0+sV
\end{equation}
for the Hamiltonian. 

\begin{theorem}
\label{mainTheorem}
Under Assumption~\ref{assump:texpdecay} and Assumption~\ref{AssumptionVX} with a certain positive constant $K_h$,  
there exists a positive constant $s_0$ such that for $|s|\le s_0$, the following holds: 
(i) The ground state of the Hamiltonian $H_s$ is non-degenerate. (ii) 
The spectral gap between the ground-state energy and the first-excited-state energy is greater than 
a positive constant which is independent of the size $|\Lambda|$ of the lattice. 
Namely, the spectral gap above the ground state does not close if the strength of the interaction is weak. 
\end{theorem}

\noindent
{\it Remark:} The stability of the spectral gap was proved by using series expansion methods 
for translationally invariant fermion \cite{GMP} and spin systems \cite{KennedyTasaki,Yarotsky}. 
In this paper, we do not require the translational invariance, and hence the statement of the above theorem holds for 
disordered or aperiodic systems. Hastings's approach also suggests that fermionic perturbation theory is applicable to 
the same situation. Actually, De Roeck and Salmhofer \cite{RS} reached a very similar result by using a Feynman graph method.   

\Section{Transformations for free fermions}
\label{Sec:TransFree}

In this section, we write the Hamiltonian $H_0$ in a positive semi-definite quadratic form, following  
the idea of Hastings \cite{Hastings}. More precisely, we first decompose the fermion operator $a_x$ into 
two Majorana fermions, and then embed the two Hamiltonians, $H_0$ and $-H_0$, into a doubled fermion Fock space 
as a direct sum of the two Hamiltonians. Here, we stress that the two Hamiltonians have opposite signs.  
This procedure enables us to realize a positive semi-definite quadratic form of the doubled Hamiltonian 
in terms of certain local fermion operators. This is nothing but a frustration-free form of the Hamiltonian. 
Therefore, the techniques which were developed for frustration-free systems can be applied to the resulting system. 

\subsection{Bogoliubov-de Gennes form of the Hamiltonian $H_0$}

In order to write the Hamiltonian $H_0$ in terms of Majorana fermions, we first express $H_0$ 
in a Bogoliubov-de Gennes form, which has a particle-hole symmetry. 

For the matrix elements in the Hamiltonian $H_0$ of (\ref{calHam}), 
we write $T_{x,y}=t_{x,y}-\mathcal{E}_{\rm F}\delta_{x,y}$. 
By using the first one in the canonical anticommutation relations (\ref{CACR}), the Hamiltonian ${H}_0$ of 
(\ref{calHam}) can be written  
\begin{eqnarray}
\label{calHammat}
{H}_0&=& \frac{1}{2}\sum_{x,y\in\Lambda}\left(a_x^\dagger T_{x,y}a_y-a_y T_{x,y}a_x^\dagger\right)
+\frac{1}{2}\sum_x T_{x,x} \nonumber\\
&=& \frac{1}{2}\left(a^\dagger T a -a\; ^tT a^\dagger \right)+\frac{1}{2}\sum_x T_{x,x}\nonumber\\
&=&\frac{1}{2}\left(\begin{matrix}{a^\dagger , a}\end{matrix}\right)
\left(\begin{matrix}T & 0 \\ 0 & -\; ^tT \\\end{matrix}\right)
\left(\begin{matrix}a \\ a^\dagger \\\end{matrix}\right)+\frac{1}{2}\sum_x T_{x,x},
\end{eqnarray}
where we have written $a=(a_x)_{x\in\Lambda}$ for the vector, 
and $^t T$ denotes the transpose of the matrix $T=(T_{x,y})$. 

Consider the unitary matrix $U=(U_{x,p})$ which diagonalizes the matrix $T$, i.e., $U^\ast TU=D$ with 
the diagonal matrix $D$ with the diagonal elements, $\lambda_1,\lambda_2,\ldots.$
The corresponding canonical transformation for the fermions is given by 
$$
a_x=\sum_p U_{x,p}b_p 
$$
with the transformed operators $b_p$. The adjoint is transformed as 
$$
a_x^\dagger=\sum_p \overline{U_{x,p}}b_p^\dagger. 
$$
One can show 
\begin{eqnarray*}
a^\dagger T a&=&\sum_{x,y}\sum_p \overline{U_{x,p}}b_p^\dagger T_{x,y}\sum_q U_{y,q}b_q\\
&=&\sum_{p,q}\sum_{x,y}\overline{U_{x,p}}T_{x,y}U_{y,q}b_p^\dagger b_q\\
&=&\sum_{p,q}\lambda_p\delta_{p,q}b_p^\dagger b_q=\sum_p \lambda_p b_p^\dagger b_p,
\end{eqnarray*}
and similarly 
\begin{eqnarray*}
a\; ^tT a^\dagger&=&\sum_{x,y}\sum_{p,q} U_{x,p}b_pT_{y,x}\overline{U_{y,q}}b_q^\dagger\\
&=&\sum_{p,q}\sum_{x,y}\overline{U_{y,q}}T_{y,x}U_{x,p}b_pb_q^\dagger\\
&=&\sum_{p.q}\lambda_p\delta_{p,q}b_pb_q^\dagger=\sum_p \lambda_p b_p b_p^\dagger.
\end{eqnarray*}
Clearly, the matrix in (\ref{calHammat}) is diagonalized as 
$$
\left(\begin{matrix}U^\ast & 0 \\ 0 & ^t U\\ \end{matrix}\right)\left(\begin{matrix}T & 0 \\ 0 & -\; 
^tT \\ \end{matrix}
\right)
\left(\begin{matrix}U & 0 \\ 0 & \overline{U}\\ \end{matrix}\right)
=\left(\begin{matrix}\lambda_1 & & & & & \\ & \lambda_2 & & & & \\ & & \ddots & & & \\ 
& & & -\lambda_1 & & \cr & & & & -\lambda_2 & \cr & & & & & \ddots\\ \end{matrix}\right). 
$$
Here, the important thing is that the two eigenvalues, $\lambda_i$ and $-\lambda_i$, come in pairs in the spectrum.  

\subsection{Majorana representation of the Hamiltonian $H_0$}

Now, we introduce Majorana fermion operators, 
$$
c_x=\frac{1}{\sqrt{2}}(a_x^\dagger + a_x),\quad \mbox{and} \quad d_x=\frac{i}{\sqrt{2}}(a_x^\dagger - a_x), 
$$
or equivalently, 
\begin{equation}
\label{abycd}
a_x^\dagger =\frac{1}{\sqrt{2}}(c_x-id_x) \quad \mbox{and} \quad 
a_x=\frac{1}{\sqrt{2}}(c_x+id_x).
\end{equation}
Note that 
\begin{eqnarray*}
a_x^\dagger T_{x,y}a_y&=&\frac{1}{2}(c_x-id_x)T_{x,y}(c_y+id_y)\\
&=&\frac{1}{2}\left(c_xT_{x,y}c_y+d_xT_{x,y}d_y+ic_xT_{x,y}d_y-i d_xT_{x,y}c_y\right)
\end{eqnarray*}
and 
\begin{eqnarray*}
-a_xT_{y,x} a_y^\dagger&=& -\frac{1}{2}(c_x+id_x)T_{y,x}(c_y-id_y)\\
&=&-\frac{1}{2}\left(c_xT_{y,x}c_y+d_xT_{y,x}d_y-ic_xT_{y,x}d_y+id_xT_{y,x}c_y\right).
\end{eqnarray*}
Therefore, one has 
\begin{eqnarray*}
a_x^\dagger T_{x,y}a_y-a_xT_{y,x} a_y^\dagger
&=&\frac{1}{2}c_x(T_{x,y}-T_{y,x})c_y+\frac{1}{2}d_x(T_{x,y}-T_{y,x})d_y\\
&+&\frac{i}{2}c_x(T_{x,y}+T_{y,x})d_y-\frac{i}{2}d_x(T_{x,y}+T_{y,x})c_y\\
&=&i(c_x{\rm Im}\;T_{x,y}c_y+d_x{\rm Im}\;T_{x,y}d_y
+c_x{\rm Re}\;T_{x,y}d_y-d_x{\rm Re}\;T_{x,y}c_y),
\end{eqnarray*}
where we have used  $T_{y,x}=\overline{T_{x,y}}$ which follows from the self-adjointness of the Hamiltonian $h$. 
In consequence, we obtain 
$$
{H}_0=\frac{1}{2}\left(\begin{matrix}c , d\end{matrix}\right)A
\left(\begin{matrix}c \cr d \end{matrix}\right)+\frac{1}{2}\sum_x T_{x,x},
$$
where we have written $c=(c_x)$ and $d=(d_x)$, and 
$$
A:=i\left(\begin{matrix}{\rm Im}\;T & {\rm Re}\;T \\ -{\rm Re}\;T 
& {\rm Im}\;T \\ \end{matrix}\right). 
$$
Here, one notices that ${\rm Im}\;T$ is antisymmetric, i.e., ${\rm Im}\;T_{y,x}=-{\rm Im}\;T_{x,y}$, 
and ${\rm Re}\;T$ is symmetric, i.e., ${\rm Re}\;T_{y,x}={\rm Re}\;T_{x,y}$. 
{From} these properties, the matrix $A$ is pure imaginary and antisymmetric, and hence $A$ is self-adjoint. 

Note that 
$$
\left(\begin{matrix}a_x \\ a_x^\dagger \\\end{matrix}\right)
=\left(\begin{matrix}\frac{1}{\sqrt{2}} & \frac{i}{\sqrt{2}} \\ \frac{1}{\sqrt{2}} & -\frac{i}{\sqrt{2}}\\\end{matrix}\right)
\left(\begin{matrix}c_x \\ d_x \\\end{matrix}\right)
$$
and 
$$
\left(\begin{matrix}a_x^\dagger \ , \ a_x \end{matrix}\right)
=\left(\begin{matrix}c_x \ , \ d_x \end{matrix}\right)
\left(\begin{matrix}\frac{1}{\sqrt{2}} & \frac{1}{\sqrt{2}} \\ -\frac{i}{\sqrt{2}} & \frac{i}{\sqrt{2}}\\\end{matrix}\right).
$$
Therefore, from (\ref{calHammat}), the matrix $A$ is written 
$$
A=\left(\begin{matrix}\frac{1}{\sqrt{2}} & \frac{1}{\sqrt{2}} \\ -\frac{i}{\sqrt{2}} & \frac{i}{\sqrt{2}}\\\end{matrix}\right)
\left(\begin{matrix}T & 0 \\ 0 & - ^tT \\\end{matrix}\right)
\left(\begin{matrix}\frac{1}{\sqrt{2}} & \frac{i}{\sqrt{2}} \\ \frac{1}{\sqrt{2}} & -\frac{i}{\sqrt{2}}\\\end{matrix}\right).
$$
This implies that, since the matrix $T$ has the spectral gap at the zero energy from the assumption, 
the matrix $A$ has the same spectral gap at the zero energy. 

We write 
\begin{equation}
\label{gammacd}
\gamma_x^\mu=\begin{cases}c_x, & \mbox{for \ } \mu=c; \\ 
                d_x, & \mbox{for \ } \mu=d. \\ \end{cases}
\end{equation} 
Then, the Hamiltonian ${H}_0$ is written 
\begin{equation}
{H}_0=\frac{1}{2}\gamma A\gamma+\frac{1}{2}\sum_x T_{x,x}
=\frac{1}{2}\sum_{x,y}\sum_{\mu,\nu}\gamma_x^\mu A_{x,y}^{\mu,\nu}\gamma_y^\nu
+\frac{1}{2}\sum_x T_{x,x}
\end{equation}

\subsection{Doubled Hamiltonian $\tilde{H}_0$}

Following Hastings \cite{Hastings}, we introduce two copies, $\gamma_1$ and $\gamma_2$, of 
the Majorana fermion $\gamma$, and two copies of the Hamiltonian ${H}_0$ with opposite signs 
so that the total Hamiltonian is given by  
\begin{eqnarray}
\label{gammaAgamma}
\tilde{{H}}_0&=&\sum_{x,y}\sum_{\mu,\nu}\left(\gamma_{1,x}^\mu A_{x,y}^{\mu,\nu}
\gamma_{1,y}^\nu-\gamma_{2,x}^\mu A_{x,y}^{\mu,\nu}\gamma_{2,y}^\nu\right)\nonumber\\
&=&\left(\begin{matrix}\gamma_1,\gamma_2\end{matrix}\right)
\left(\begin{matrix}A & 0\\ 0 & -A \\ \end{matrix}\right)
\left(\begin{matrix}\gamma_1\\ \gamma_2\\ \end{matrix}\right). 
\end{eqnarray}
Since $A$ is self-adjoint, we write $P_\pm$ for the spectral projection onto 
the positive and negative energies for $A$, respectively. We also write $s(A)=P_+-P_-$. 
Moreover, we introduce a unitary transformation \cite{Hastings},
\begin{equation}
\label{tildeU}
\hat{U}=\frac{1}{\sqrt{2}}\left(\begin{matrix}1 & is(A)\\ is(A) & 1\\ \end{matrix}\right).
\end{equation}
Then, one has 
\begin{eqnarray}
\label{VAV}
\hat{U}^\ast\left(\begin{matrix}A & 0 \\ 0 & -A \\ \end{matrix}\right)\hat{U}&=&\frac{1}{2}
\left(\begin{matrix}1 & -is(A) \\ -is(A) & 1\\ \end{matrix}\right)
\left(\begin{matrix}A & 0 \\ 0 & -A \\ \end{matrix}\right)
\left(\begin{matrix}1 & is(A)\\  is(A) & 1\\ \end{matrix}\right)\nonumber\\
&=&\frac{1}{2}
\left(\begin{matrix}1 & -is(A) \\ -is(A) & 1\\ \end{matrix}\right)
\left(\begin{matrix}A & i|A|\\ -i|A| & -A\\ \end{matrix}\right)\nonumber\\
&=&\left(\begin{matrix}0 & i|A|\\  -i|A| & 0\\ \end{matrix}\right). 
\end{eqnarray}
The corresponding new fermion operators are given by 
\begin{equation}
\label{tildegamma}
\left(\begin{matrix}\tilde{\gamma}_1 \\ \tilde{\gamma}_2\\ \end{matrix}\right):=
\hat{U}^\ast \left(\begin{matrix}\gamma_1\\ \gamma_2\\ \end{matrix}\right)
=\frac{1}{\sqrt{2}}\left(\begin{matrix}\gamma_1-is(A)\gamma_2\\  -is(A)\gamma_1+\gamma_2\\ \end{matrix}\right).
\end{equation}
\medskip

\begin{lemma}
The operator $|A|$ is symmetric, i.e., $|A|_{x,y}^{\mu,\nu}=|A|_{y,x}^{\nu,\mu}$, and $s(A)$ is antisymmetric. 
\end{lemma}

\begin{proof}{Proof} As is well known, the positive square root $R^{1/2}$ of a positive operator $R$ 
can be expressed in terms of power series in $R$. (See, e.g., Sec.~2.2.2 in \cite{BraRobi}.) 
Therefore, it is sufficient to show that $A^2$ is symmetric. Since $A$ is antisymmetric, one has 
$$
(A^2)_{x,y}^{\mu,\nu}=\sum_{z,\kappa} A_{x,z}^{\mu,\kappa}A_{z,y}^{\kappa,\nu}
=\sum_{z,\kappa}(-A_{z,x}^{\kappa,\mu})(-A_{y,z}^{\nu,\kappa})
=\sum_{z,\kappa} A_{y,z}^{\nu,\kappa}A_{z,x}^{\kappa,\mu}=(A^2)_{y,x}^{\nu,\mu}.
$$
Therefore, $|A|=\sqrt{A^\ast A}=\sqrt{A^2}$ is symmetric. 

Next, let us prove that $s(A)$ is antisymmetric. Since $|A|$ is symmetric as shown in above, one has   
$$
A_{x,y}^{\mu,\nu}=\sum_{z,\kappa} s(A)_{x,z}^{\mu,\kappa}|A|_{z,y}^{\kappa,\nu}
=-A_{y,x}^{\nu.\mu}=-\sum_{z,\kappa} |A|_{y,z}^{\nu,\kappa}s(A)_{z,x}^{\kappa,\mu}
=-\sum_{z,\kappa} s(A)_{z,x}^{\kappa,\mu}|A|_{z,y}^{\kappa,\nu}.
$$
This implies $s(A)|A|=-\;^ts(A)|A|$. Therefore, it is enough to show that $|A|$ is invertible.
In order to show this, we recall the assumption of the nonvanishing spectral gap 
of the Hamiltonian $h-\mathcal{E}_{\rm F}$. 
This assumption implies that the spectral gap of $A$ at zero energy 
is nonvanishing from the definition of $A$. Therefore, the operator $|A|$ is invertible. 
\end{proof}

By relying on this lemma, one can show that $\tilde{\gamma}_j^\ast=\tilde{\gamma}_j$ for $j=1,2$, 
i.e., both of the transformed fermions are a Majorana fermion. 
In order to show this, we note that 
$$
\tilde{\gamma}_{1,x}^\mu=\frac{1}{\sqrt{2}}\left(\gamma_{1,x}^\mu-\sum_{y,\nu}is(A)_{x,y}^{\mu,\nu}\gamma_{2,y}^\nu\right).
$$
The adjoint is 
\begin{eqnarray*}
(\tilde{\gamma}_{1,x}^\mu)^\ast
&=&\frac{1}{\sqrt{2}}\left(\gamma_{1,x}^\mu+\sum_{y,\nu}i\overline{s(A)_{x,y}^{\mu,\nu}}\gamma_{2,y}^\nu\right)\\
&=&\frac{1}{\sqrt{2}}\left(\gamma_{1,x}^\mu+\sum_{y,\nu}is(A)_{y,x}^{\nu,\mu}\gamma_{2,y}^\nu\right)\\
&=&\frac{1}{\sqrt{2}}\left(\gamma_{1,x}^\mu-\sum_{y,\nu}is(A)_{x,y}^{\mu,\nu}\gamma_{2,y}^\nu\right)
=\tilde{\gamma}_{1,x}^\mu,
\end{eqnarray*}
where we have used that $s(A)$ is self-adjoint and antisymmetric. 
In the same way, one can show $(\tilde{\gamma}_{2,x}^\mu)^\ast=\tilde{\gamma}_{2,x}^\mu$. 
Further, by using the antisymmetry of $s(A)$, one has 
\begin{eqnarray}
\left(\gamma_1,\gamma_2\right)\hat{U}
&=&\frac{1}{\sqrt{2}}\left(\gamma_1,\gamma_2\right)\left(\begin{matrix}1 & is(A)\\ is(A) & 1\\ \end{matrix}\right)\nonumber\\
&=&\frac{1}{\sqrt{2}}\left(\gamma_1+i\gamma_2s(A),i\gamma_1s(A)+\gamma_2\right)\nonumber\\
&=&\frac{1}{\sqrt{2}}\left(\gamma_1-is(A)\gamma_2,-is(A)\gamma_1+\gamma_2\right)
=\left(\tilde{\gamma}_1,\tilde{\gamma}_2\right).
\end{eqnarray}
Combining this, (\ref{gammaAgamma}), (\ref{VAV}) and (\ref{tildegamma}), one obtains 
\begin{eqnarray*}
\tilde{{H}}_0&=&\left(\tilde{\gamma}_1,\tilde{\gamma}_2\right)
\left(\begin{matrix}0 & i|A|\\ -i|A| & 0 \\ \end{matrix}\right)
\left(\begin{matrix}\tilde{\gamma}_1\\ \tilde{\gamma}_2\\ \end{matrix}\right)\\
&=&i\tilde{\gamma}_1|A|\tilde{\gamma}_2-i\tilde{\gamma}_2|A|\tilde{\gamma}_1\\
&=&2i\tilde{\gamma}_1|A|\tilde{\gamma}_2.
\end{eqnarray*}
Finally, we introduce a usual complex fermion \cite{Hastings},
$$
\eta_x^\mu=\frac{1}{\sqrt{2}}(\tilde{\gamma}_{1,x}^\mu+i\tilde{\gamma}_{2,x}^\mu),
$$
or equivalently
$$
\tilde{\gamma}_{1,x}^\mu=\frac{1}{\sqrt{2}}\left[(\eta_x^\mu)^\dagger+\eta_x^\mu\right],
\quad \mbox{and}\quad 
\tilde{\gamma}_{2,x}^\mu=\frac{i}{\sqrt{2}}\left[(\eta_x^\mu)^\dagger-\eta_x^\mu\right].
$$
Then, the Hamiltonian $\tilde{{H}}_0$ can be written 
\begin{eqnarray*}
\tilde{{H}}_0&=&-\sum_{x,y}\sum_{\mu,\nu} \left[(\eta_x^\mu)^\dagger+\eta_x^\mu\right]
|A|_{x,y}^{\mu,\nu}\left[(\eta_y^\nu)^\dagger-\eta_y^\nu\right]\\
&=&\sum_{x,y}\sum_{\mu,\nu}\left[(\eta_x^\mu)^\dagger|A|_{x,y}^{\mu,\nu}\eta_y^\nu
-\eta_x^\mu|A|_{x,y}^{\mu,\nu}(\eta_y^\nu)^\dagger\right]\\
&=&2\sum_{x,y}\sum_{\mu,\nu}(\eta_x^\mu)^\dagger|A|_{x,y}^{\mu,\nu}\eta_y^\nu+\sum_{x,\mu}|A|_{x,x}^{\mu,\mu},
\end{eqnarray*}
where we have used the fact that $|A|$ is symmetric, and the anticommutation relations, 
$\{\eta_x^\mu,(\eta_y^\nu)^\dagger\}=\delta_{\mu,\nu}\delta_{x,y}$, $\{\eta_x^\mu,\eta_y^\nu\}=0$ and 
$\{(\eta_x^\mu)^\dagger,(\eta_y^\nu)^\dagger\}=0$. 
Since $|A|\ge \Delta E$ from the assumption of the spectral gap of the Hamiltonian $H_0$,  
one has 
$$
\tilde{{H}}_0-\sum_{x,\mu}|A|_{x,x}^{\mu,\mu}=2\eta^\dagger|A|\eta\ge 2\Delta E \eta^\dagger \eta.
$$
We write 
$$
\tilde{\mathcal{H}}_0:=\eta^\dagger|A|\eta
$$
and 
$$
\tilde{\mathcal{N}}:=\eta^\dagger \eta.
$$
The local Hamiltonian $\tilde{\mathcal{H}}_{0,Z}$ is given by 
\begin{equation}
\tilde{\mathcal{H}}_{0,Z}:=\sum_{\mu,\nu}(\eta_x^\mu)^\dagger|A|_{x,y}^{\mu,\nu}\eta_y^\nu
\end{equation} 
with $Z=\{x,y\}$ if $x\ne y$, and $Z=\{x\}$ if $x=y$. 

Let $\hat{\mathcal{U}}$ be the unitary transformation which diagonalizes the matrix $|A|$. Namely, 
$\hat{\mathcal{U}}^\ast |A|\hat{\mathcal{U}}=\hat{\mathcal{D}}$ with the diagonal matrix $\hat{\mathcal{D}}$ 
and the diagonal elements $\tilde{\lambda}_j$.   
Then, the fermion operator $\eta$ is transformed as $\alpha:=\hat{\mathcal{U}}^\ast\eta$. 
The Hamiltonian $\tilde{\mathcal{H}}_0$ is written 
$$
\tilde{\mathcal{H}}_0=\sum_i \tilde{\lambda}_i\alpha_i^\dagger\alpha_i,
$$
and the number operator $\tilde{\mathcal{N}}$ is written 
$$
\tilde{\mathcal{N}}=\sum_i \alpha_i^\dagger\alpha_i. 
$$
Since the eigenvalues $\tilde{\lambda}_i$ are all positive and satisfy $\tilde{\lambda}_i\ge\Delta E$, 
one has 
\begin{equation}
\label{H02N2}
(\tilde{\mathcal{H}}_0)^2\ge (\Delta E)^2(\tilde{\mathcal{N}})^2.
\end{equation}

Since the fermion operators $a_x$ is expressed in terms of the fermion operators $\eta$, 
the interaction $V$ can be written in terms of the operator $\eta$. We write $\tilde{V}$ 
for the transformed interaction. The local interaction $V_X$ is also transformed to $\tilde{V}_{\tilde{X}}$ 
with the different supports $\tilde{X}$. Namely, the transformed interaction $\tilde{V}$ is written 
$$ 
\tilde{V}=\sum_{\tilde{X}\subset \Lambda}\tilde{V}_{\tilde{X}}.
$$
{From} the assumption (\ref{assumpV}) on the interaction $V$, we have 
\begin{equation}
\label{boundednesstildeV0}
\sum_{\tilde{X}\ni x,y}|\tilde{X}|^3\Vert \tilde{V}_{\tilde{X}}\Vert \le \tilde{C}_{h,V}e^{-\tilde{m}_{h,V}{\rm dist}(x,y)}
\end{equation}
with positive constants, $\tilde{C}_{h,V}$ and $\tilde{m}_{h,V}$, which are independent of the size $|\Lambda|$ of the lattice. 
The proof of this local boundedness of the interaction $\tilde{V}$ is given in Appendix~\ref{append:boundedtildeV}.

\Section{Transformation connecting two ground states}
\label{Sec:conneGS}

In the preceding section, we have obtained the desired form of the unperturbed Hamiltonian $\tilde{\mathcal{H}}_0$. 
In this section, following the idea of Hastings \cite{Hastings,{BravHast}}, 
we will further transform the total Hamiltonian $\tilde{\mathcal{H}}_0+\tilde{V}$ so that 
all the local interactions annihilate the ground state of the unperturbed Hamiltonian $\tilde{\mathcal{H}}_0$.

Consider a Hamiltonian, 
\begin{equation}
\label{tildecalHs}
\tilde{\mathcal{H}}_s:= \tilde{\mathcal{H}}_0+s\tilde{{V}},
\end{equation}
with a real parameter $s\in[0,1]$. We assume that the ground state $\tilde{\Phi}_0(s)$ of 
$\tilde{\mathcal{H}}_s$ is unique, and 
that the spectrum has a nonvanishing gap $\gamma$ above the ground-state energy uniformly in 
the parameter $s\in[0,1]$. We write $\tilde{P}_0(s)$ for the projection onto the ground state. 
Then, there exists a unitary transformation $\tilde{U}(s)$ such that 
\begin{equation}
\label{U(s)}
\tilde{U}^\ast(s)\tilde{P}_0(s)\tilde{U}(s)=\tilde{P}_0(0).
\end{equation}
See Chap.~II, Sec.~4.2 in Kato's book \cite{Kato}. 

Note that 
$$
\tilde{\mathcal{H}}_s=[1-\tilde{P}_0(s)]\tilde{\mathcal{H}}_s[1-\tilde{P}_0(s)]+\tilde{E}_0(s)\tilde{P}_0(s),
$$
where $\tilde{E}_0(s)$ is the ground state energy of $\tilde{\mathcal{H}}_s$. 
Using the unitary transformation $\tilde{U}(s)$ of (\ref{U(s)}), 
one has 
$$
\tilde{\mathcal{H}}(s):=\tilde{U}^\ast(s)\tilde{\mathcal{H}}_s\tilde{U}(s)
=[1-\tilde{P}_0(0)]\tilde{U}^\ast(s)\tilde{\mathcal{H}}_s\tilde{U}(s)[1-\tilde{P}_0(0)]
+\tilde{E}_0(s)\tilde{P}_0(0). 
$$
Therefore, 
\begin{equation}
\label{commuHs'P00}
[\tilde{\mathcal{H}}(s),\tilde{P}_0(0)]=0.
\end{equation}
This implies that the ground state of $\tilde{\mathcal{H}}(s)$ is given by 
the same ground state $\tilde{\Phi}_0(0)$ as that of $\tilde{\mathcal{H}}_0=\tilde{\mathcal{H}}(0)$. 

We write
\begin{equation}
\label{tildecalH0s}
\tilde{\mathcal{H}}_0(s):=\tilde{U}^\ast(s)\tilde{\mathcal{H}}_0\tilde{U}(s)
=\sum_{Z\subset\Lambda}\tilde{U}^\ast(s)\tilde{\mathcal{H}}_{0,Z}\tilde{U}(s)
\end{equation}
and 
\begin{equation}
\label{tildeVs}
\tilde{V}(s):=\tilde{U}^\ast(s)\tilde{V}\tilde{U}(s)
=\sum_{Z\subset\Lambda}\tilde{U}^\ast(s)\tilde{V}_Z\tilde{U}(s).
\end{equation}
Clearly, 
\begin{equation}
\tilde{\mathcal{H}}(s)=\tilde{\mathcal{H}}_0(s)+s\tilde{V}(s).
\end{equation}

Further, following Hastings \cite{Hastings}, we introduce 
\begin{equation}
\label{calRs}
\mathcal{R}_s(\cdots):=\int_{-\infty}^{+\infty}\tilde{\tau}_{s,t}(\cdots)w_\gamma(t)dt,
\end{equation}
where 
$$
\tilde{\tau}_{s,t}(\cdots):=\exp[it\tilde{\mathcal{H}}(s)](\cdots)\exp[-it\tilde{\mathcal{H}}(s)]
$$
and $w_\gamma$ is a function such that its Fourier transform $\hat{w}_\gamma$ satisfies 
$\hat{w}_\gamma(0)=1$ and $\hat{w}_\gamma(\omega)=0$ for $|\omega|\ge \gamma$. 
As a concrete function, we use the function $w_\gamma(t)$ of the equation~(2.1) in Lemma~2.3 in \cite{BMNS}. 
Since the Hamiltonian $\tilde{\mathcal{H}}(s)$ 
has the spectral gap $\gamma>0$ above the ground state from the assumption about the Hamiltonian $\tilde{\mathcal{H}}_s$, 
the following holds: 
\begin{equation}
\label{localGScondition}
[1-\tilde{P}_0(0)]\mathcal{R}_s(\mathcal{A})\tilde{P}_0(0)=0
\end{equation}
for any operator $\mathcal{A}$, from (\ref{commuHs'P00}). We also note that 
$$
\mathcal{R}_s(1)=1
$$
and 
$$
\mathcal{R}_s(\tilde{\mathcal{H}}(s))=\tilde{\mathcal{H}}(s)
$$
{from} $\hat{w}_\gamma(0)=1$. By combining this second one with (\ref{tildecalH0s}) and (\ref{tildeVs}), 
one has 
\begin{equation}
\label{tildecalHsp}
\tilde{\mathcal{H}}(s)=\sum_{Z\subset\Lambda}\mathcal{R}_s(\tilde{U}^\ast(s)\tilde{\mathcal{H}}_{0,Z}\tilde{U}(s))
+s\sum_{Z\subset\Lambda}\mathcal{R}_s(\tilde{U}^\ast(s)\tilde{V}_Z\tilde{U}(s)).
\end{equation}

Consider first the summand in the second sum in the right-hand side of (\ref{tildecalHsp}). 
{From} the general relation (\ref{localGScondition}), 
one notices that the ground state $\tilde{\Phi}_0(0)$ is the eigenvector of 
the summand $\mathcal{R}_s(\tilde{U}^\ast(s)\tilde{V}_Z\tilde{U}(s))$. On the other hand, one has 
$$
\langle \tilde{\Phi}_0(0),\mathcal{R}_s(\tilde{U}^\ast(s)\tilde{V}_Z\tilde{U}(s))\tilde{\Phi}_0(0)\rangle
=\langle \tilde{\Phi}_0(0),\tilde{U}^\ast(s)\tilde{V}_Z\tilde{U}(s)\tilde{\Phi}_0(0)\rangle
$$
because $\tilde{\Phi}_0(0)$ is the eigenvector of $\tilde{\mathcal{H}}(s)$ and $\hat{w}_\gamma(0)=1$. From these observations, 
we have 
\begin{equation}
\label{RUVZUP0}
\mathcal{R}_s(:\tilde{U}^\ast(s)\tilde{V}_Z\tilde{U}(s):)\tilde{P}_0(0)=0,
\end{equation}
where $:(\cdots):$ is the subtraction of the expectation value with respect to the ground state $\tilde{\Phi}_0(0)$ 
which is defined by  
$$
:\mathcal{A}:=\mathcal{A}-\langle\tilde{\Phi}_0(0),\mathcal{A}\tilde{\Phi}_0(0)\rangle
$$
for an operator $\mathcal{A}$. The important point is that the corresponding term satisfies the condition (\ref{RUVZUP0}), 
and is vanishing in the weak coupling limit $s\searrow 0$ for the interaction. 
We write 
\begin{equation}
\label{tildecalWZ1s}
\tilde{\mathcal{W}}_Z^{(1)}(s):=s\mathcal{R}_s(:\tilde{U}^\ast(s)\tilde{V}_Z\tilde{U}(s):).
\end{equation}

Next consider the first sum in the right-hand side of (\ref{tildecalHsp}). Note that 
\begin{eqnarray}
\label{UtildecalHZint}
\tilde{U}^\ast(s)\tilde{\mathcal{H}}_{0,Z}\tilde{U}(s)-\tilde{\mathcal{H}}_{0,Z}
&=&\int_0^s ds'\; \frac{d}{ds'}\tilde{U}^\ast(s')\tilde{\mathcal{H}}_{0,Z}\tilde{U}(s')\nonumber\\
&=&\int_0^s ds'\;\tilde{U}^\ast(s')i[\tilde{\mathcal{H}}_{0,Z},\tilde{D}(s')]\tilde{U}(s'),
\end{eqnarray}
where we have used \cite{BMNS}
$$
\frac{d}{ds}\tilde{U}(s)=i\tilde{D}(s)\tilde{U}(s)
$$
with the self-adjoint operator $\tilde{D}(s)$. Therefore, 
$$
\mathcal{R}_s(\tilde{U}^\ast(s)\tilde{\mathcal{H}}_{0,Z}\tilde{U}(s))
=\mathcal{R}_s(\tilde{\mathcal{H}}_{0,Z})+
\int_0^s ds'\;\mathcal{R}_s(\tilde{U}^\ast(s')i[\tilde{\mathcal{H}}_{0,Z},\tilde{D}(s')]\tilde{U}(s')).
$$
The second term in the right-hand side is vanishing in the weak coupling limit $s\searrow 0$ for the interaction, 
and can be treated in the same way as in above. Namely, we have 
\begin{equation}
\int_0^s ds'\;\mathcal{R}_s(:\tilde{U}^\ast(s')i[\tilde{\mathcal{H}}_{0,Z},\tilde{D}(s')]\tilde{U}(s'):)\tilde{P}_0(0)=0.
\end{equation}
We write 
\begin{equation}
\label{tildecalWZ2s}
\tilde{\mathcal{W}}_Z^{(2)}(s)
:=\int_0^s ds'\;\mathcal{R}_s(:\tilde{U}^\ast(s')i[\tilde{\mathcal{H}}_{0,Z},\tilde{D}(s')]\tilde{U}(s'):).
\end{equation}

In order to treat the first term $\mathcal{R}_s(\tilde{\mathcal{H}}_{0,Z})$, we note that  
\begin{eqnarray*}
e^{it\tilde{\mathcal{H}}(s)}\tilde{\mathcal{H}}_{0,Z}e^{-it\tilde{\mathcal{H}}(s)}-\tilde{\mathcal{H}}_{0,Z}
&=&\int_0^t dt'\;\frac{d}{dt'}e^{it'\tilde{\mathcal{H}}(s)}\tilde{\mathcal{H}}_{0,Z}e^{-it'\tilde{\mathcal{H}}(s)}\\
&=&\int_0^t dt'\;e^{it'\tilde{\mathcal{H}}(s)}i[\tilde{\mathcal{H}}(s),\tilde{\mathcal{H}}_{0,Z}]
e^{-it'\tilde{\mathcal{H}}(s)}. 
\end{eqnarray*}
On the other hand, by using (\ref{UtildecalHZint}), we have 
\begin{eqnarray*}
\tilde{\mathcal{H}}(s)&=&\tilde{U}^\ast(s)\tilde{\mathcal{H}}_0\tilde{U}(s)+s\tilde{U}^\ast(s)\tilde{V}\tilde{U}(s)\\
&=&\tilde{\mathcal{H}}_0+\int_0^s ds'\;\tilde{U}^\ast(s')i[\tilde{\mathcal{H}}_0,\tilde{D}(s')]\tilde{U}(s')
+s\tilde{U}^\ast(s)\tilde{V}\tilde{U}(s). 
\end{eqnarray*}
{From} these equations, we obtain 
\begin{eqnarray*}
e^{it\tilde{\mathcal{H}}(s)}\tilde{\mathcal{H}}_{0,Z}e^{-it\tilde{\mathcal{H}}(s)}&=&
\tilde{\mathcal{H}}_{0,Z}+\int_0^t dt'\;e^{it'\tilde{\mathcal{H}}(s)}i[\tilde{\mathcal{H}}(s),\tilde{\mathcal{H}}_{0,Z}]
e^{-it'\tilde{\mathcal{H}}(s)}\\
&=&\tilde{\mathcal{H}}_{0,Z}+\int_0^t dt'\;e^{it'\tilde{\mathcal{H}}(s)}i[\tilde{\mathcal{H}}_0,\tilde{\mathcal{H}}_{0,Z}]
e^{-it'\tilde{\mathcal{H}}(s)}+\tilde{\mathcal{M}}_{t,Z}(s)
\end{eqnarray*}
with 
\begin{equation}
\label{tildecalMtZs}
\tilde{\mathcal{M}}_{t,Z}(s):=\int_0^t dt'\;\tilde{\tau}_{s,t'}
\left(is[\tilde{U}^\ast(s)\tilde{V}\tilde{U}(s),\tilde{\mathcal{H}}_{0,Z}]+
\int_0^s ds'\;i[\tilde{U}^\ast(s')i[\tilde{\mathcal{H}}_0,\tilde{D}(s')]\tilde{U}(s'),\tilde{\mathcal{H}}_{0,Z}]\right).
\end{equation}
Here, we recall that the ground state $\tilde{\Phi}_0(0)$ of $\tilde{\mathcal{H}}_0$ is also 
the ground state of $\tilde{\mathcal{H}}(s)$, and $\tilde{\mathcal{H}}_{0,Z}\tilde{\Phi}_0(0)=0$ for all $Z$. 
Therefore, we have  
\begin{equation}
\tilde{\mathcal{M}}_{t,Z}(s)\tilde{P}_0(0)=0.
\end{equation}
In addition, one obtains 
\begin{equation}
\sum_{Z\subset\Lambda}\mathcal{R}_s(\tilde{\mathcal{H}}_{0,Z})=\tilde{\mathcal{H}}_0
+\sum_{Z\subset\Lambda}\int_{-\infty}^{+\infty}dt\;w_\gamma(t)\tilde{\mathcal{M}}_{t,Z}(s).  
\end{equation}
We write 
\begin{equation}
\label{tildecalWZ3s}
\tilde{\mathcal{W}}_Z^{(3)}(s):=\int_{-\infty}^{+\infty}dt\;w_\gamma(t)\tilde{\mathcal{M}}_{t,Z}(s).
\end{equation}
Consequently, we obtain 
\begin{equation}
\label{tildecalHs:tildecalH0+tildecalWs}
\tilde{\mathcal{H}}(s):=\tilde{\mathcal{H}}_0+\tilde{\mathcal{W}}(s)
\end{equation}
with 
$$
\tilde{\mathcal{W}}(s)=\sum_{Z\subset\Lambda}\tilde{\mathcal{W}}_Z(s)
$$
and 
$$
\tilde{\mathcal{W}}_Z(s):=\sum_{i=1}^3\tilde{\mathcal{W}}_Z^{(i)}(s).
$$
Here, we have dropped the constant term, which does not affect the spectral gap of the Hamiltonian $\tilde{\mathcal{H}}(s)$. 
The local interactions satisfy 
\begin{equation}
\label{tildeV'P000}
\tilde{\mathcal{W}}_Z^{(i)}(s)\tilde{P}_0(0)=0
\end{equation}
for $i=1,2,3$. 

Clearly, the operations of $\mathcal{R}_s$, $\tilde{U}(s)$ and $\tilde{\tau}_{s,t'}$ enlarge 
the original support $Z$ of the potentials $\tilde{\mathcal{W}}_Z^{(i)}(s)$ to a large region. 
Actually, they exhibit sub-exponentially decaying tails. 
Following \cite{Hastings,NSY,BHV}, we decompose them into a sum of terms $\tilde{\mathcal{W}}_{Z,n}^{(i)}(s)$ 
with a specific compact support,  
\begin{equation}
\tilde{\mathcal{W}}_Z^{(i)}(s)=\sum_{n=0}^\infty \tilde{\mathcal{W}}_{Z,n}^{(i)}(s), 
\end{equation}
so that each term $\tilde{\mathcal{W}}_{Z,n}^{(i)}(s)$ satisfies 
\begin{equation}
\tilde{\mathcal{W}}_{Z,n}^{(i)}(s)\tilde{P}_0(0)=0
\end{equation}
and 
\begin{equation}
\Vert\tilde{\mathcal{W}}_{Z,n}^{(i)}(s)\Vert \rightarrow 0 \quad \mbox{as \ }s\rightarrow 0. 
\end{equation}
The construction of $\tilde{\mathcal{W}}_{Z,n}^{(i)}(s)$ are given 
in Appendices~\ref{Sec:LocaApproObs}-\ref{Sec:BoundtildecalWZ3s}.

\Section{Relatively bounded perturbations}
\label{Sec:relativeBound}

Consider a Hamiltonian, 
\begin{equation}
\tilde{\mathcal{H}}=\tilde{\mathcal{H}}_0+\tilde{\mathcal{W}}.
\end{equation}
Here, we assume that the two self-adjoint operators, $\tilde{\mathcal{H}}_0$ 
and $\tilde{\mathcal{W}}$, satisfy the bound \cite{Kato},
\begin{equation}
\label{relativebound}
\tilde{\mathcal{W}}^2\le b^2(\tilde{\mathcal{H}}_0)^2,
\end{equation}
with the small constant $b\ge 0$. If the constant $b$ can be chosen to be sufficiently small compared 
to the spectral gap $\Delta E$ above the ground state of the Hamiltonian $\tilde{\mathcal{H}}_0$, then 
the spectral gap above the ground state of the total Hamiltonian can be expected to be still 
nonvanishing against the weak perturbation $\tilde{\mathcal{W}}$. We will prove that this statement is true 
in the present situation. 

{From} the argument of Sec.~\ref{Sec:TransFree}, we set 
\begin{equation}
\label{tildecalH0eta}
\tilde{\mathcal{H}}_0:=\sum_{\mu,\nu\in\{c,d\}}\sum_{x,y\in\Lambda}(\eta_x^\mu)^\dagger |A|_{x,y}^{\mu,\nu}\eta_y^\nu. 
\end{equation}
We write $\tilde{\Lambda}=\Lambda\times\{c,d\}$, and write $\xi_y=\eta_x^\mu$ with $y=(x,\mu)\in\tilde{\Lambda}$ 
for short. Then, the Hamiltonian $\tilde{\mathcal{H}}_0$ is written 
\begin{equation}
\label{H0xi}
\tilde{\mathcal{H}}_0=\sum_{y,z\in\tilde{\Lambda}}\xi_y^\dagger |A|_{y,z}\xi_z.
\end{equation}

\begin{pro}
\label{pro:relative}
Suppose that the perturbed potential $\tilde{\mathcal{W}}$ satisfies the relative bound (\ref{relativebound}) 
with a positive $b$, and $\tilde{\mathcal{W}}\tilde{P}_0=0$, 
where $\tilde{P}_0$ is the projection onto the ground state $\tilde{\Phi}_0$ of 
the unperturbed Hamiltonian $\tilde{\mathcal{H}}_0$ of (\ref{H0xi}). Then, for a small $b$, 
the ground state of the total Hamiltonian $\tilde{\mathcal{H}}=\tilde{\mathcal{H}}_0+\tilde{\mathcal{W}}$ is unique 
with the ground-state energy $\tilde{E}_0=0$. Further, the first excitation energy $\tilde{E}_1$ of $\tilde{\mathcal{H}}$ 
satisfies 
$$
\tilde{E}_1\ge \left[1-\frac{b}{\sqrt{1-2b}}\right]\Delta E 
$$
for $b$ satisfying 
$$
\frac{b}{\sqrt{1-2b}}\le \frac{1}{2}.
$$
Here, $\Delta E>0$ is the spectral gap above the ground state of the unperturbed Hamiltonian $\tilde{\mathcal{H}}_0$.  
\end{pro}

\begin{proof}{Proof}
{From} the assumptions, the expectation value of 
the total Hamiltonian $\tilde{\mathcal{H}}$ with respect to the ground state $\tilde{\Phi}_0$ 
of the unperturbed Hamiltonian $\tilde{\mathcal{H}}_0$ is 
given by 
$$
\langle \tilde{\Phi}_0,\tilde{\mathcal{H}}\tilde{\Phi}_0\rangle
=\langle \tilde{\Phi}_0,\tilde{\mathcal{H}}_0\tilde{\Phi}_0\rangle+\langle\tilde{\Phi}_0,\tilde{\mathcal{W}}\tilde{\Phi}_0\rangle=0.
$$
Therefore, the ground state energy $\tilde{E}_0$ of $\tilde{\mathcal{H}}$ is less than or equal to zero, 
i.e., $\tilde{E}_0\le 0$. 

We write $\tilde{\Psi}_1$ for the first excited state of $\tilde{\mathcal{H}}$ with the energy eigenvalue $\tilde{E}_1$, 
i.e., $\tilde{\mathcal{H}}\tilde{\Psi}_1=\tilde{E}_1\tilde{\Psi}_1$. Clearly, one has 
\begin{eqnarray*}
\tilde{E}_1^2&=&\langle\tilde{\Psi}_1,\tilde{\mathcal{H}}^2\tilde{\Psi}_1\rangle\\ 
&=&\langle \tilde{\Psi}_1,(\tilde{\mathcal{H}}_0)^2\tilde{\Psi}_1\rangle+\langle\tilde{\Psi}_1,
\tilde{\mathcal{W}}^2\tilde{\Psi}_1\rangle 
+\langle\tilde{\Psi}_1,\tilde{\mathcal{H}}_0\tilde{\mathcal{W}}\tilde{\Psi}_1\rangle+\langle\tilde{\Psi}_1,
\tilde{\mathcal{W}}\tilde{\mathcal{H}}_0\tilde{\Psi}_1\rangle.
\end{eqnarray*}
The third and fourth terms is evaluated as 
$$
|{\rm Re}\langle\tilde{\Psi}_1,\tilde{\mathcal{H}}_0\tilde{\mathcal{W}}\tilde{\Psi}_1\rangle|
\le \sqrt{\langle \tilde{\Psi}_1,(\tilde{\mathcal{H}}_0)^2\tilde{\Psi}_1\rangle\langle\tilde{\Psi}_1,
\tilde{\mathcal{W}}^2\tilde{\Psi}_1\rangle}
\le b\langle\tilde{\Psi}_1,(\tilde{\mathcal{H}}_0)^2\tilde{\Psi}_1\rangle,
$$
where we have used Schwarz inequality and the assumption (\ref{relativebound}). Therefore, one has 
$$
\tilde{E}_1^2\ge (1-2b)\langle\tilde{\Psi}_1,(\tilde{\mathcal{H}}_0)^2\tilde{\Psi}_1\rangle.
$$
Clearly,   
$$
\langle\tilde{\Psi}_1,(\tilde{\mathcal{H}}_0)^2\tilde{\Psi}_1\rangle\le\frac{1}{1-2b}\tilde{E}_1^2
$$
for a small $b$. Combining this with (\ref{relativebound}), one obtains
$$
\langle\tilde{\Psi}_1,\tilde{\mathcal{W}}^2\tilde{\Psi}_1\rangle\le\frac{b^2}{1-2b}\tilde{E}_1^2.
$$
On the other hand, one has 
$$
\langle\tilde{\Psi}_1,\tilde{\mathcal{W}}^2\tilde{\Psi}_1\rangle=\langle\tilde{\Psi}_1,
(\tilde{\mathcal{H}}-\tilde{\mathcal{H}}_0)^2\tilde{\Psi}_1\rangle
=\langle\tilde{\Psi}_1,(\tilde{E}_1-\tilde{\mathcal{H}}_0)^2\tilde{\Psi}_1\rangle. 
$$
{From} these two, one obtains
$$
\langle\tilde{\Psi}_1,(\tilde{E}_1-\tilde{\mathcal{H}}_0)^2\tilde{\Psi}_1\rangle\le \frac{b^2}{1-2b}\tilde{E}_1^2. 
$$
When $\tilde{E}_1\ge \Delta E$, the statement of the theorem holds. Therefore, it is enough to treat the case 
of $\tilde{E}_1<\Delta E$. Then, the above bound is written 
\begin{equation}
\label{PsiE1-H02bound}
\langle\tilde{\Psi}_1,(\tilde{E}_1-\tilde{\mathcal{H}}_0)^2\tilde{\Psi}_1\rangle\le \frac{b^2}{1-2b}\Delta{E}^2. 
\end{equation}
Let us consider varying the strength of the interaction continuously from zero, i.e., varying the parameter $b$ 
{from} $b=0$ to a positive value by relying on the continuity of the spectrum of the Hamiltonian with respect to 
the parameters. For such a small $b$, we can assume $\tilde{E}_1\ge \Delta E/2$.  
Clearly, this inequality implies $\tilde{E}_1\ge \Delta E-\tilde{E}_1>0$. Combining this with 
the assumption on the spectrum of $\tilde{\mathcal{H}}_0$, one has 
$$
(\Delta E-\tilde{E}_1)^2\le \langle\tilde{\Psi}_1,(\tilde{E}_1-\tilde{\mathcal{H}}_0)^2\tilde{\Psi}_1\rangle
$$
for the normalized vector $\tilde{\Psi}_1$. Further, by combining this with (\ref{PsiE1-H02bound}), we obtain   
$$
\left[1-\frac{b}{\sqrt{1-2b}}\right]\Delta E\le \tilde{E}_1<\Delta E
$$
for $b$ satisfying 
$$
\frac{b}{\sqrt{1-2b}}\le \frac{1}{2}.
$$
\end{proof}  

The resulting bound for the first excitation energy $\tilde{E}_1$ implies 
that for a small $b$, the model is still gapped against a weak perturbation $\tilde{\mathcal{W}}$.
Therefore, in order to prove the existence of a non-vanishing spectral gap above the ground state of the interacting system, 
it is sufficient to prove the bound (\ref{relativebound}) with a small positive constant $b$ which depends on 
the strength of the interactions $\tilde{\mathcal{W}}$. 

We assume that interaction $\tilde{\mathcal{W}}$ is written as a sum of local interactions, 
$$
\tilde{\mathcal{W}}=\sum_{X\subset\tilde \Lambda} \tilde{\mathcal{W}}_X,
$$
where $X$ is the compact support of the local interaction $\tilde{\mathcal{W}}_X$.  
We also assume 
\begin{equation}
\label{tildecalWXtildeP00}
\tilde{\mathcal{W}}_X\tilde{P}_0=0
\end{equation}
for the projection $\tilde{P}_0$ onto the ground state of $\tilde{\mathcal{H}}_0$. 
We write $q_x:=\xi_x^\dagger \xi_x$ and $\overline{q}_x:=\xi_x\xi_x^\dagger$. 
Since the support $X$ of $\tilde{\mathcal{W}}_X$ is compact, the condition (\ref{tildecalWXtildeP00}) implies 
\begin{equation}
\label{tildecalWXtildeP0X0}
\tilde{\mathcal{W}}_X\tilde{P}_{0,X}=0,
\end{equation}
where 
\begin{equation}
\tilde{P}_{0,X}:=\prod_{x\in X}\overline{q}_x. 
\end{equation}
Write ${X}=\{x_1,x_2,\ldots,x_{|X|}\}$. Then,  
\begin{eqnarray}
\label{qbarID}
\tilde{P}_{0,X}&=&\overline{q}_{x_1}\overline{q}_{x_2}\cdots\overline{q}_{x_{|X|}}\nonumber\\
&=&\overline{q}_{x_1}\overline{q}_{x_2}\cdots\overline{q}_{x_{|X|}}
-\overline{q}_{x_1}\overline{q}_{x_2}\cdots\overline{q}_{x_{|X|-1}}
+\overline{q}_{x_1}\overline{q}_{x_2}\cdots\overline{q}_{x_{|X|-1}}\nonumber\\
&=&\overline{q}_{x_1}\overline{q}_{x_2}\cdots\overline{q}_{x_{|X|-1}}(-q_{x_{|X|}})
+\overline{q}_{x_1}\overline{q}_{x_2}\cdots\overline{q}_{x_{|X|-1}}\nonumber\\
&=&\overline{q}_{x_1}\overline{q}_{x_2}\cdots\overline{q}_{x_{|X|-1}}(-q_{x_{|X|}})
+\overline{q}_{x_1}\overline{q}_{x_2}\cdots\overline{q}_{x_{|X|-2}}(-q_{x_{|X|-1}})
+\cdots+(-q_{x_1})+1.\nonumber\\
\end{eqnarray}
Therefore, we have 
\begin{equation}
1-\tilde{P}_{0,X}=\sum_{n=1}^{|X|}Q_{n-1}q_{x_n}
\end{equation}
with $Q_0=1$ and $Q_{n-1}=\overline{q}_{x_1}\overline{q}_{x_2}\cdots\overline{q}_{x_{n-1}}$ for $n=2,3,\ldots,|X|$. 
{From} this and the condition (\ref{tildecalWXtildeP0X0}), we obtain 
\begin{equation}
\label{WXdecomp}
\tilde{\mathcal{W}}_X=(1-\tilde{P}_{0,X})\tilde{\mathcal{W}}_X(1-\tilde{P}_{0,X})
=\sum_{m,n=1}^{|X|}q_mQ_{m-1}^\ast\tilde{\mathcal{W}}_XQ_{n-1}q_n. 
\end{equation}

\begin{lemma}
\label{lem:tildecalWXcondition}
Suppose that the local interactions $\tilde{\mathcal{W}}_X$ satisfies 
\begin{equation}
\tilde{\mathcal{W}}_X\tilde{P}_{0,X}=0 
\end{equation}
and 
\begin{equation}
\label{assumptionWX}
\sup_{\tilde{\Lambda}}\sup_{x\in\tilde{\Lambda}} \sum_{X\ni x}|X|\cdot\Vert \tilde{\mathcal{W}}_X\Vert =:\tilde{g}<\infty. 
\end{equation}
Then, the following bound is valid:   
\begin{equation}
\label{tildecalW2tildecalH02bound}
\langle \tilde{\Psi},\tilde{\mathcal{W}}^2\tilde{\Psi}\rangle\le\left(\frac{\tilde{g}}{\Delta E}\right)^2\langle\tilde{\Psi},
(\tilde{\mathcal{H}}_0)^2\tilde{\Psi}\rangle 
\end{equation}
for any vector $\tilde{\Psi}$. Here, $\Delta E$ is the lower bound for the spectral gap above the ground state of 
the unperturbed Hamiltonian $\tilde{\mathcal{H}}_0$. 
\end{lemma}

\begin{proof}{Proof} Our aim is to estimate the quantity,
\begin{equation}
\label{PsiW2Psiexpand}
\langle \tilde{\Psi},\tilde{\mathcal{W}}^2\tilde{\Psi}\rangle=\sum_{X,Y}\langle \tilde{\Psi},
\tilde{\mathcal{W}}_X \tilde{\mathcal{W}}_Y\tilde{\Psi}\rangle,
\end{equation}
for a wavefunction $\tilde{\Psi}$. 
\medskip

\noindent
(i) Consider first the summand $\langle \tilde{\Psi},\tilde{\mathcal{W}}_X\tilde{\mathcal{W}}_Y\tilde{\Psi}\rangle$ 
in the right-hand side of (\ref{PsiW2Psiexpand})
in the case that $X\cap Y=\emptyset$. We write $X=\{x_1,x_2,\ldots,x_{|X|}\}$ 
and $Y=\{y_1,y_2,\ldots,y_{|Y|}\}$. By using the above decomposition (\ref{WXdecomp}) and 
Schwarz inequality, we have
\begin{eqnarray*}
|\langle\tilde{\Psi},\tilde{\mathcal{W}}_X\tilde{\mathcal{W}}_Y\tilde{\Psi}\rangle|&\le&
\sum_{m,n,k,\ell}
|\langle\tilde{\Psi},q_{x_m}Q_{m-1}^\ast \tilde{\mathcal{W}}_XQ_{n-1}q_{x_n}q_{y_k}Q_{k-1}^\ast \tilde{\mathcal{W}}_YQ_{\ell-1}q_{y_\ell}\tilde{\Psi}\rangle|\\
&=&\sum_{m,n,k,\ell}
|\langle\tilde{\Psi},q_{x_m}q_{y_k}Q_{m-1}^\ast \tilde{\mathcal{W}}_XQ_{n-1}Q_{k-1}^\ast \tilde{\mathcal{W}}_YQ_{\ell-1}q_{x_n}q_{y_\ell}\tilde{\Psi}\rangle|\\
&\le&\sum_{m,n,k,\ell}
\sqrt{\langle\tilde{\Psi},q_{x_m}q_{y_k}Q_{m-1}^\ast \tilde{\mathcal{W}}_XQ_{n-1}Q_{n-1}^\ast \tilde{\mathcal{W}}_X^\ast Q_{m-1}q_{y_k}q_{x_m}
\tilde{\Psi}\rangle}\\
& &\times 
\sqrt{\langle\tilde{\Psi},q_{y_\ell}q_{x_n}Q_{\ell-1}^\ast \tilde{\mathcal{W}}_Y^\ast Q_{k-1}Q_{k-1}^\ast \tilde{\mathcal{W}}_YQ_{\ell-1}q_{x_n}q_{y_\ell}\tilde{\Psi}\rangle}\\
&\le& \sum_{m,n,k,\ell}
\Vert \tilde{\mathcal{W}}_X\Vert \cdot \Vert \tilde{\mathcal{W}}_Y\Vert \sqrt{\langle\tilde{\Psi},q_{x_m}q_{y_k}\tilde{\Psi}\rangle}
\sqrt{\langle\tilde{\Psi},q_{y_\ell}q_{x_n}\tilde{\Psi}\rangle}.
\end{eqnarray*}
Further, by using the inequality $2ab\le(a^2+b^2)$ for $a,b>0$, we obtain 
\begin{eqnarray}
\label{WXWYempty}
|\langle\tilde{\Psi},\tilde{\mathcal{W}}_X\tilde{\mathcal{W}}_Y\tilde{\Psi}\rangle|
&\le& \frac{1}{2}\sum_{m,n,k,\ell}
\Vert \tilde{\mathcal{W}}_X\Vert \cdot \Vert \tilde{\mathcal{W}}_Y\Vert(\langle\tilde{\Psi},q_{x_m}q_{y_k}\tilde{\Psi}\rangle+\langle\tilde{\Psi},q_{y_\ell}q_{x_n}\tilde{\Psi}\rangle)\nonumber\\
&=&\sum_{m,k}|X|\cdot|Y|\cdot\Vert \tilde{\mathcal{W}}_X\Vert \cdot \Vert \tilde{\mathcal{W}}_Y\Vert\langle\tilde{\Psi},q_{x_m}q_{y_k}\tilde{\Psi}\rangle\nonumber\\
&=&|X|\cdot|Y|\cdot\Vert \tilde{\mathcal{W}}_X\Vert \cdot \Vert \tilde{\mathcal{W}}_Y\Vert\sum_{x_m\in X,y_k\in Y}\langle\tilde{\Psi},q_{x_m}q_{y_k}\tilde{\Psi}\rangle.
\end{eqnarray}
\smallskip

\noindent
(ii) Next consider the case that $X\cap Y\ne\emptyset$. In this case, 
we have 
\begin{eqnarray}
\label{WXWYestimate}
|\langle\tilde{\Psi},\tilde{\mathcal{W}}_X\tilde{\mathcal{W}}_Y\tilde{\Psi}\rangle|&\le&\sum_{m,\ell}
\Vert \tilde{\mathcal{W}}_X\Vert \cdot \Vert \tilde{\mathcal{W}}_Y\Vert \sqrt{\langle\tilde{\Psi},q_{x_m}\tilde{\Psi}\rangle}
\sqrt{\langle\tilde{\Psi},q_{y_\ell}\tilde{\Psi}\rangle}\nonumber\\
&\le&\frac{1}{2}\sum_{m,\ell}\Vert \tilde{\mathcal{W}}_X\Vert \cdot \Vert \tilde{\mathcal{W}}_Y\Vert
(\langle\tilde{\Psi},q_{x_m}\tilde{\Psi}\rangle+\langle\tilde{\Psi},q_{y_\ell}\tilde{\Psi}\rangle)\nonumber\\
&=&\frac{1}{2}|Y|\cdot\Vert \tilde{\mathcal{W}}_X\Vert \cdot \Vert \tilde{\mathcal{W}}_Y\Vert
\sum_{x_m\in X}\langle\tilde{\Psi},q_{x_m}\tilde{\Psi}\rangle
\nonumber\\&+&
\frac{1}{2}|X|\cdot\Vert \tilde{\mathcal{W}}_X\Vert \cdot \Vert \tilde{\mathcal{W}}_Y\Vert
\sum_{y_\ell\in Y}\langle\tilde{\Psi},q_{y_\ell}\tilde{\Psi}\rangle
\end{eqnarray}
because of the noncommutativity of the operators. 

Note that 
\begin{eqnarray}
\label{PsiW2Psiexpand2}
\langle \tilde{\Psi},\tilde{\mathcal{W}}^2\tilde{\Psi}\rangle&=&\sum_{X,Y}
\langle \tilde{\Psi},\tilde{\mathcal{W}}_X\tilde{\mathcal{W}}_Y\tilde{\Psi}\rangle\nonumber\\
&\le&\sum_{X,Y}
|\langle \tilde{\Psi},\tilde{\mathcal{W}}_X\tilde{\mathcal{W}}_Y\tilde{\Psi}\rangle|\nonumber\\
&\le&\sum_{X,Y\; : \; X\cap Y\ne\emptyset}
|\langle \tilde{\Psi},\tilde{\mathcal{W}}_X\tilde{\mathcal{W}}_Y\tilde{\Psi}\rangle|+\sum_{X,Y\; : \; X\cap Y=\emptyset}
|\langle \tilde{\Psi},\tilde{\mathcal{W}}_X\tilde{\mathcal{W}}_Y\tilde{\Psi}\rangle|.
\end{eqnarray}
By using the bound (\ref{WXWYestimate}) and the assumption (\ref{assumptionWX}), 
the first sum in the right-hand side can be evaluated as 
\begin{eqnarray*}
\sum_{X,Y\; : \; X\cap Y\ne\emptyset}
|\langle \tilde{\Psi},\tilde{\mathcal{W}}_X\tilde{\mathcal{W}}_Y\tilde{\Psi}\rangle|&\le& \sum_{X,Y\; : \; X\cap Y\ne\emptyset}
|Y|\cdot \Vert \tilde{\mathcal{W}}_X\Vert\cdot\Vert \tilde{\mathcal{W}}_Y\Vert \sum_{x\in X}\langle \tilde{\Psi},q_x\tilde{\Psi}\rangle\\
&\le&\sum_{x\in\tilde{\Lambda}}\;\sum_{X\ni x}\; \sum_{Y: X\cap Y\ne \emptyset}
|Y|\cdot \Vert \tilde{\mathcal{W}}_X\Vert\cdot\Vert \tilde{\mathcal{W}}_Y\Vert\langle \tilde{\Psi},q_x\tilde{\Psi}\rangle\\
&\le&\sum_{x\in\tilde{\Lambda}}\;\sum_{X\ni x}\;\sum_{y\in X}\; \sum_{Y\ni y}
|Y|\cdot \Vert \tilde{\mathcal{W}}_X\Vert\cdot\Vert \tilde{\mathcal{W}}_Y\Vert\langle \tilde{\Psi},q_x\tilde{\Psi}\rangle\\
&\le&\tilde{g} \sum_{x\in\tilde{\Lambda}}\;\sum_{X\ni x}|X|\cdot\Vert \tilde{\mathcal{W}}_X\Vert\langle \tilde{\Psi},q_x\tilde{\Psi}\rangle\\
&\le&(\tilde{g})^2\sum_{x\in\tilde{\Lambda}}\;\langle \tilde{\Psi},q_x\tilde{\Psi}\rangle. 
\end{eqnarray*}
Similarly, by using the bound (\ref{WXWYempty}), 
the second sum in the right-hand side of (\ref{PsiW2Psiexpand2}) can be estimated as 
\begin{eqnarray*}
\sum_{X,Y\; : \; X\cap Y=\emptyset}
|\langle \tilde{\Psi},\tilde{\mathcal{W}}_X\tilde{\mathcal{W}}_Y\tilde{\Psi}\rangle|
&\le&\sum_{X,Y\; : \; X\cap Y=\emptyset}
|X|\cdot|Y|\cdot\Vert \tilde{\mathcal{W}}_X\Vert \cdot \Vert \tilde{\mathcal{W}}_Y\Vert\sum_{x\in X,y\in Y}\langle\tilde{\Psi},q_xq_y\tilde{\Psi}\rangle\\
&\le&\sum_{x,y\in\tilde{\Lambda}\; :\; x\ne y}\; \sum_{X\ni x}\; \sum_{Y\ni y}\; 
|X|\cdot|Y|\cdot\Vert \tilde{\mathcal{W}}_X\Vert \cdot \Vert \tilde{\mathcal{W}}_Y\Vert\langle\tilde{\Psi},q_xq_y\tilde{\Psi}\rangle\\
&\le&(\tilde{g})^2\sum_{x,y\in\tilde{\Lambda}\; :\; x\ne y}\; \langle\tilde{\Psi},q_xq_y\tilde{\Psi}\rangle. 
\end{eqnarray*}
Combining these with  
$$
(\tilde{\mathcal{N}})^2=\Bigl(\sum_{x\in\tilde{\Lambda}}q_x\Bigr)^2
=\sum_{x\in\tilde{\Lambda}}q_x+\sum_{y,z\;:\;y\ne z}q_yq_z,
$$
we obtain 
$$
\langle \tilde{\Psi},\tilde{\mathcal{W}}^2\tilde{\Psi}\rangle\le(\tilde{g})^2\langle\tilde{\Psi},
(\tilde{\mathcal{N}})^2\tilde{\Psi}\rangle . 
$$
Further, by using the inequality $(\tilde{\mathcal{H}}_0)^2\ge (\Delta E)^2(\tilde{\mathcal{N}})^2$ of (\ref{H02N2}), 
we obtain the desired result (\ref{tildecalW2tildecalH02bound}). 
\end{proof}

\Section{Proof of Theorem~\ref{mainTheorem}}
\label{Sec:ProofMainTheorem}

In order to prove the statements of Theorem~\ref{mainTheorem}, from the argument of the preceding section, 
it is sufficient to prove the following: 
There exists a positive function $b(s)$ of the parameter $s$ such that $b(s)$ satisfies 
$b(s)\rightarrow 0$ as $s\rightarrow 0$ and  
\begin{equation}
(\tilde{\mathcal{W}}(s))^2\le b^2(s)(\tilde{\mathcal{H}}_0)^2
\end{equation}
for the Hamiltonian $\tilde{\mathcal{H}}(s)$ of (\ref{tildecalHs:tildecalH0+tildecalWs}).

Write
$$
\tilde{\mathcal{W}}^{(i)}(s):=\sum_{Z,n}\tilde{\mathcal{W}}_{Z,n}^{(i)}(s)
$$
for $i=1,2,3$. Then, 
$$
\tilde{\mathcal{W}}(s)=\sum_{i=1}^3 \tilde{\mathcal{W}}^{(i)}(s). 
$$
For any vector $\tilde{\Psi}$, one has 
\begin{eqnarray*}
\langle\tilde{\Psi},(\tilde{\mathcal{W}}(s))^2\tilde{\Psi} \rangle &\le&
\sum_{i=1}^3 \langle\tilde{\Psi},(\tilde{\mathcal{W}}^{(i)}(s))^2\tilde{\Psi} \rangle
+\sum_{i,j:i\ne j}\langle\tilde{\Psi},\tilde{\mathcal{W}}^{(i)}(s)\tilde{\mathcal{W}}^{(j)}(s)\tilde{\Psi}\rangle. 
\end{eqnarray*}
The summand in the second sum in the right-hand side is estimated by 
\begin{eqnarray*}
\left|\langle\tilde{\Psi},\tilde{\mathcal{W}}^{(i)}(s)\tilde{\mathcal{W}}^{(j)}(s)\tilde{\Psi}\rangle\right|
&\le&\sqrt{\langle\tilde{\Psi},(\tilde{\mathcal{W}}^{(i)}(s))^2\tilde{\Psi}\rangle
\langle\tilde{\Psi},(\tilde{\mathcal{W}}^{(j)}(s))^2\tilde{\Psi}\rangle}\\
&\le& \frac{1}{2}\left[\langle\tilde{\Psi},(\tilde{\mathcal{W}}^{(i)}(s))^2\tilde{\Psi}\rangle 
+ \langle\tilde{\Psi},(\tilde{\mathcal{W}}^{(j)}(s))^2\tilde{\Psi}\rangle\right].
\end{eqnarray*}
Therefore, one has 
\begin{equation}
\label{tildecalWsWis}
\langle\tilde{\Psi},(\tilde{\mathcal{W}}(s))^2\tilde{\Psi} \rangle \le
3\sum_{i=1}^3 \langle\tilde{\Psi},(\tilde{\mathcal{W}}^{(i)}(s))^2\tilde{\Psi} \rangle.
\end{equation}

On the other hand, all of the local potentials, $\tilde{\mathcal{W}}_{Z,n}^{(i)}(s)$, 
satisfy the two conditions in Lemma~\ref{lem:tildecalWXcondition}. 
In particular, the boundedness is written 
\begin{equation}
\label{tildecalWZnsboundedness}
\sup_\Lambda \sup_{x\in\Lambda}\sum_Z\sum_{n=0}^\infty {\rm Ind}[x\in {\rm supp}\; \tilde{\mathcal{W}}_{Z,n}^{(i)}(s)]\cdot
|{\rm supp}\; \tilde{\mathcal{W}}_{Z,n}^{(i)}(s)|\cdot \Vert\tilde{\mathcal{W}}_{Z,n}^{(i)}(s)\Vert<\tilde{g}^{(i)}|s|,
\end{equation}
where ${\rm Ind}[\cdots]$ is the indicator function, and $\tilde{g}^{(i)}$ is the positive constants which are independent of 
the size $|\Lambda|$ of the lattice. 
The construction of $\tilde{\mathcal{W}}_{Z,n}^{(i)}(s)$ and the proofs of the above bounds (\ref{tildecalWZnsboundedness}) 
are given in Appendix~\ref{Sec:LocaApproObs}-\ref{Sec:BoundtildecalWZ3s}. 
By applying Lemma~\ref{lem:tildecalWXcondition} to (\ref{tildecalWsWis}), we have 
\begin{equation}
\langle\tilde{\Psi},(\tilde{\mathcal{W}}^{(i)}(s))^2\tilde{\Psi} \rangle \le\left(\frac{\tilde{g}^{(i)}}{\Delta E}\right)^2|s|^2
\langle\tilde{\Psi},(\tilde{\mathcal{H}}_0)^2\tilde{\Psi}\rangle 
\end{equation}
for any vector $\tilde{\Psi}$. Therefore, we obtain 
\begin{equation}
\langle\tilde{\Psi},(\tilde{\mathcal{W}}(s))^2\tilde{\Psi} \rangle \le 3\sum_{i=1}^3
\left(\frac{\tilde{g}^{(i)}}{\Delta E}\right)^2|s|^2
\langle\tilde{\Psi},(\tilde{\mathcal{H}}_0)^2\tilde{\Psi}\rangle.  
\end{equation}
This implies the existence of the non-vanishing spectral gap above the unique ground state 
of the total Hamiltonian $\tilde{\mathcal{H}}(s)$ 
for a small $|s|$ by Proposition~\ref{pro:relative} in the preceding section.

\appendix 

\Section{Local boundedness of the transformed potential $\tilde{V}$}
\label{append:boundedtildeV}

In this appendix, we prove the bound (\ref{boundednesstildeV0}) for the transformed potential $\tilde{V}$. 

To begin with, let us consider the unperturbed Hamiltonian $\tilde{\mathcal{H}}_0$. 
{From} the expression (\ref{tildecalH0eta}) of the unperturbed Hamiltonian $\tilde{\mathcal{H}}_0$, 
the local Hamiltonian can be written 
\begin{equation}
\tilde{\mathcal{H}}_{0,Z}=\sum_{\mu,\nu}(\eta_x^\mu)^\dagger |A|_{x,y}^{\mu,\nu}\eta_y^\nu,
\end{equation}
where the support $Z$ is given by $Z=\{x,y\}$ if $x\ne y$ and $Z=\{x\}$ if $x=y$. Since the matrix $|A|$ can be written 
$|A|=As(A)$, the norm can be evaluated as 
\begin{equation}
\label{tildecalH0Zexpdecay}
\Vert\tilde{\mathcal{H}}_{0,Z}\Vert\le \sum_{\mu,\nu}\left||A|_{x,y}^{\mu,\nu}\right|
\le \sum_{\mu,\nu}\sum_{z,\kappa}|A_{x,z}^{\mu,\kappa}|\cdot|s(A)_{z,y}^{\kappa,\nu}|
\le C_{t,h}\cdot e^{-m_{t,h}{\rm dist}(x,y)}
\end{equation}
with the positive constants, $C_{t,h}$ and $m_{t,h}$, where we have used the decay bound (\ref{texpdecay}) 
for the hopping amplitudes $t_{x,y}$ which yields an exponential decay bound for $|A_{x,z}^{\mu,\kappa}|$ 
because the matrix elements of $A$ are written in terms of $t_{x,y}$, and we have also used the assumption on  
the spectral gap above the Fermi level which yields an exponential decay bound for $|s(A)_{z,y}^{\kappa,\nu}|$. 
{From} this bound (\ref{tildecalH0Zexpdecay}), one has 
\begin{equation}
\label{boundednessH0Zx}
\sum_{Z\ni x}\Vert\tilde{\mathcal{H}}_{0,Z}\Vert\le\sum_{y}C_{t,h}e^{-m_{t,h}{\rm dist}(x,y)}
\le C_{t,h}K_{t,h}
\end{equation}
and 
\begin{equation}
\label{boundednessH0Zxy}
\sum_{Z\ni x,y}\Vert\tilde{\mathcal{H}}_{0,Z}\Vert\le C_{t,h}K_{t,h}\cdot e^{-m_{t,h}{\rm dist}(x,y)}
\end{equation}
where the constant $K_{t,h}$ is given by  
$$
K_{t,h}:=\sum_{y\in\ze^d}e^{-m_{t,h}{\rm dist}(0,y)}. 
$$

In the following, we will obtain an analogue of these bounds for the local interactions $\tilde{V}_Z$. 
To begin with, we recall the expressions (\ref{abycd}) and (\ref{gammacd}). Using these, one has 
$$
a_x^\dagger=\frac{1}{\sqrt{2}}(\gamma_{1,x}^c-i\gamma_{1,x}^d),
\quad 
a_x=\frac{1}{\sqrt{2}}(\gamma_{1,x}^c+i\gamma_{1,x}^d).
$$
Further, from (\ref{tildeU}) and (\ref{tildegamma}), one obtains 
$$
\gamma_{1,x}^\mu=\frac{1}{\sqrt{2}}
\left[\tilde{\gamma}_{1,x}^\mu+\sum_{y,\nu}is(A)_{x,y}^{\mu,\nu}\tilde{\gamma}_{2,y}^{\nu}\right].
$$
These expressions yield  
$$
a_x^\dagger=\frac{1}{2}\left[\tilde{\gamma}_{1,x}^c-i\tilde{\gamma}_{1,x}^d
+\sum_{y,\nu}\{is(A)_{x,y}^{c,\nu}+s(A)_{x,y}^{d,\nu}\}\tilde{\gamma}_{2,y}^\nu\right]
$$
and 
$$
a_x=\frac{1}{2}\left[\tilde{\gamma}_{1,x}^c+i\tilde{\gamma}_{1,x}^d
+\sum_{y,\nu}\{is(A)_{x,y}^{c,\nu}-s(A)_{x,y}^{d,\nu}\}\tilde{\gamma}_{2,y}^\nu\right].
$$
We write 
$$
\zeta_{1,x}:=\frac{1}{\sqrt{2}}(\tilde{\gamma}_{1,x}^c+i\tilde{\gamma}_{1,x}^d)
$$
and 
\begin{equation}
\label{omega}
\omega_{2,x}:=\frac{1}{\sqrt{2}}\sum_{y,\nu}\{is(A)_{x,y}^{c,\nu}-s(A)_{x,y}^{d,\nu}\}\tilde{\gamma}_{2,y}^\nu.
\end{equation}
Then, we have
\begin{equation}
\label{transazo}
a_x=\frac{1}{\sqrt{2}}(\zeta_{1,x}+\omega_{2,x}),\quad 
a_x^\dagger=\frac{1}{\sqrt{2}}(\zeta_{1,x}^\dagger +\omega_{2,x}^\dagger).
\end{equation}
The Majorana fermion operators, $\tilde{\gamma}_{1,x}^\mu$ and $\tilde{\gamma}_{2,x}^\mu$, can be expressed in terms of $\eta_x^\mu$, 
i.e.,  
$$
\tilde{\gamma}_{1,x}^\mu=\frac{1}{\sqrt{2}}[(\eta_x^\mu)^\dagger+\eta_x^\mu],\quad\mbox{and}\quad 
\tilde{\gamma}_{2,x}^\mu=\frac{i}{\sqrt{2}}[(\eta_x^\mu)^\dagger-\eta_x^\mu]. 
$$
Therefore, the fermion operators $a_x$ can be also expressed in terms of $\eta_x^\mu$.

The local interaction $V_X$ with the support $X$ is written as a sum of many types of monomials of 
the fermion operators, $a_x$ and $a_y^\dagger$, for $x,y\in\Lambda$ as follows: 
\begin{equation}
V_X=\sum_{I=I(X)}V_X^{(I)},
\end{equation}
where $V_X^{(I)}$ is the interaction of a monomial of the fermion operators, $a_x$ and $a_y^\dagger$, 
and we have introduced the superscript $I$ to distinguish many types of the monomials. 
Clearly, the superscript $I$ depends on the support $X$, and therefore we will often write $I=I(X)$. 
Specifically, $V_X^{(I)}$ has the form, 
\begin{equation}
\label{VXI}
V_X^{(I)}=C_X^{(I)}a_{x_1}^\dagger a_{x_2}^\dagger \cdots a_{x_m}^\dagger a_{y_1}a_{y_2}\cdots a_{y_n} 
a_{z_1}^\dagger a_{z_1} a_{z_2}^\dagger a_{z_2}\cdots a_{z_\ell}^\dagger a_{z_\ell},
\end{equation}
where $C_X^{(I)}$ is a constant, and the three sets of the sites, $\{x_1,x_2,\ldots,x_m\}$, $\{y_1,y_2,\ldots,y_n\}$ 
and $\{z_1,z_2,\ldots,z_\ell\}$, are disjoint each other, and satisfy 
$$
X=\{x_1,x_2,\ldots,x_m\}\cup\{y_1,y_2,\ldots,y_n\}\cup\{z_1,z_2,\ldots,z_\ell\}.
$$
Note that the operator norm $\Vert \mathcal{A}\Vert$ for an operator $\mathcal{A}$ can be written  
\begin{equation}
\Vert \mathcal{A}\Vert=\sup_{\phi,\psi}\frac{|\langle \phi,\mathcal{A}\psi\rangle|}{\Vert \phi\Vert \Vert \psi\Vert}
\end{equation}
By using this, one has 
\begin{equation}
\label{VXVXICXIbound}
\Vert V_X\Vert\ge \Vert V_X^{(I)}\Vert=|C_X^{(I)}|.
\end{equation}

Using the transformations (\ref{transazo}), the interaction $V_X^{(I)}$ can be written  
\begin{equation}
V_X^{(I)}=\sum_{\tilde{I}:I=I(X)\rightarrow\tilde{I}}{V}_{X}^{(I\rightarrow\tilde{I})},
\end{equation}
where ${V}_{X}^{(I\rightarrow\tilde{I})}$ is the monomial of the fermion operators, $\zeta_{1,x}$, $\zeta_{1,y}^\dagger$, 
$\omega_{2,z}$, $\omega_{2,w}^\dagger$, for $x,y,z,w\in\Lambda$   
with the support $X$ and the type $\tilde{I}$ of the monomial; we have written $I(X)\rightarrow\tilde{I}$ 
when $\tilde{I}$ is obtained from $I=I(X)$.
Further, by using (\ref{omega}), we have  
\begin{equation}
{V}_{X}^{(I\rightarrow\tilde{I})}
=\sum_{\tilde{X}:X\rightarrow \tilde{X},I(X)\rightarrow\tilde{I}} \tilde{V}_{\tilde{X}}^{(\tilde{I})},
\end{equation}
where the support $\tilde{X}$ is determined by the supports of the operators, $\zeta_{1,x}$, $\zeta_{1,y}^\dagger$ 
and $\tilde{\gamma}_{2,z}^\mu$, for $x,y,z\in\Lambda$ and $\mu=c,d$. 
One can easily check that $\tilde{V}_{\tilde{X}}^{(\tilde{I})}$ is the monomial of the operators, 
$\zeta_{1,x}$, $\zeta_{1,y}^\dagger$, $\tilde{\omega}_{2,u,z}^\mu$ and $\tilde{\omega}_{2,v,w}^\nu$ 
for $x,y,z,w,u,v\in\Lambda$ and $\mu,\nu\in\{c,d\}$, where we have written  
\begin{equation}
\label{tildeomega}
\tilde{\omega}_{2,x,y}:=\sum_{\nu=c,d}M_{x,y}^\nu\tilde{\gamma}_{2,y}^\nu
\end{equation}
with
$$
M_{x,y}^\nu:=\frac{1}{\sqrt{2}}\{is(A)_{x,y}^{c,\nu}-s(A)_{x,y}^{d,\nu}\}.
$$
{From} the assumption on the spectral gap of the single fermion Hamiltonian $h$, one has 
\begin{equation}
\left|M_{x,y}^\nu\right|\le\frac{C_h}{\sqrt{2}}e^{-m_h{\rm dist}(x,y)},
\end{equation} 
where the positive constants, $C_h$ and $m_h$, are determined by the Hamiltonian $h$. Combining this bound, 
(\ref{tildeomega}) and $\Vert\tilde{\gamma}_{2,y}^\nu\Vert=1/\sqrt{2}$, one obtains 
\begin{equation}
\label{tildeomegabound}
\Vert \tilde{\omega}_{2,x,y}\Vert=\Vert \tilde{\omega}_{2,x,y}^\dagger\Vert
\le \sum_{\nu=c.d}|M_{x,y}^\nu|\Vert\tilde{\gamma}_{2,y}^\nu\Vert
\le C_he^{-m_h{\rm dist}(x,y)}.
\end{equation}

{From} (\ref{omega}), (\ref{transazo}), (\ref{VXI}) and (\ref{tildeomega}), 
the interaction $\tilde{V}_{\tilde{X}}^{(\tilde{I})}$ can be written in the form,
\begin{equation}
\label{tildeVtildeXtildeIexpr}
\tilde{V}_{\tilde{X}}^{(\tilde{I})}=\pm C_X^{(I)}\left(\frac{1}{\sqrt{2}}\right)^{m+n}
\left(\frac{1}{2}\right)^\ell\zeta_{1,u_1}^\bullet\zeta_{1,u_2}^\bullet\cdots\zeta_{1,u_i}^\bullet
\tilde{\omega}_{2,v_1,w_1}^\bullet\tilde{\omega}_{2,v_2,w_2}^\bullet\cdots\tilde{\omega}_{2,v_j,w_j}^\bullet,
\end{equation}
where $\mathcal{A}^\bullet$ denotes $\mathcal{A}^\dagger$ or $\mathcal{A}$ for an operator $\mathcal{A}$, 
and the sites for the operators satisfy 
$X=\{u_1,u_2,\ldots,u_i\}\cup\{v_1,v_2,\ldots,v_j\}$ 
and $\tilde{X}=\{u_1,u_2,\ldots,u_i\}\cup\{w_1,w_2,\ldots,w_j\}$. 
Therefore, by using (\ref{VXVXICXIbound}) and (\ref{tildeomegabound}), we obtain 
\begin{equation}
\label{tildeVtildeXtildeIbound}
\Vert\tilde{V}_{\tilde{X}}^{(\tilde{I})}\Vert\le \left(\frac{1}{\sqrt{2}}\right)^{m+n}
\left(\frac{1}{2}\right)^\ell\Vert V_X\Vert (C_h)^j\prod_{k=1}^j e^{-m_h{\rm dist}(v_k,w_k)}. 
\end{equation}

The transformed local interaction $\tilde{V}_{\tilde{X}}$ 
with the support $\tilde{X}$ can be written as a sum of the monomials $\tilde{V}_{\tilde{X}}^{(\tilde{I})}$, 
\begin{equation}
\tilde{V}_{\tilde{X}}=\sum_{X: X\rightarrow \tilde{X}}\sum_{I=I(X):X\rightarrow\tilde{X}}
\sum_{\tilde{I}:X\rightarrow \tilde{X},I(X)\rightarrow\tilde{I}}
\tilde{V}_{\tilde{X}}^{(\tilde{I})}.
\end{equation}
{From} this expression, one has 
\begin{eqnarray*}
\sum_{\tilde{X}\ni x,y}|\tilde{X}|^3\Vert \tilde{V}_{\tilde{X}}\Vert 
&\le& \sum_{\tilde{X}\ni x,y}\sum_{X: X\rightarrow \tilde{X}}\sum_{I=I(X):X\rightarrow\tilde{X}}
\sum_{\tilde{I}:X\rightarrow \tilde{X},I(X)\rightarrow\tilde{I}}|\tilde{X}|^3
\Vert \tilde{V}_{\tilde{X}}^{(\tilde{I})}\Vert\\
&\le& J_1+J_2+J_3+J_4
\end{eqnarray*}
with 
\begin{equation}
J_1:=\sum_{\tilde{X}}\sum_{X: X\rightarrow \tilde{X}}\sum_{I=I(X):X\rightarrow\tilde{X}}
\sum_{\tilde{I}:X\rightarrow \tilde{X},I(X)\rightarrow\tilde{I}}|\tilde{X}|^3
\Vert \tilde{V}_{\tilde{X}}^{(\tilde{I})}\Vert\times{\rm Ind}[x,y\in X],
\end{equation}
\begin{equation}
J_2:=\sum_{\tilde{X}}\sum_{X: X\rightarrow \tilde{X}}\sum_{I=I(X):X\rightarrow\tilde{X}}
\sum_{\tilde{I}:X\rightarrow \tilde{X},I(X)\rightarrow\tilde{I}}|\tilde{X}|^3
\Vert \tilde{V}_{\tilde{X}}^{(\tilde{I})}\Vert\times{\rm Ind}[x\in X,y\in \tilde{X}\backslash X],
\end{equation}
\begin{equation}
J_3:=\sum_{\tilde{X}}\sum_{X: X\rightarrow \tilde{X}}\sum_{I=I(X):X\rightarrow\tilde{X}}
\sum_{\tilde{I}:X\rightarrow \tilde{X},I(X)\rightarrow\tilde{I}}|\tilde{X}|^3
\Vert \tilde{V}_{\tilde{X}}^{(\tilde{I})}\Vert\times{\rm Ind}[y\in X,x\in \tilde{X}\backslash X],
\end{equation}
and 
\begin{equation}
J_4:=\sum_{\tilde{X}}\sum_{X: X\rightarrow \tilde{X}}\sum_{I=I(X):X\rightarrow\tilde{X}}
\sum_{\tilde{I}:X\rightarrow \tilde{X},I(X)\rightarrow\tilde{I}}|\tilde{X}|^3
\Vert \tilde{V}_{\tilde{X}}^{(\tilde{I})}\Vert\times{\rm Ind}[x,y\in \tilde{X}\backslash X],
\end{equation}
where ${\rm Ind}[\cdots]$ is the indicator function. 

Consider first the quantity $J_1$. One can easily show that it can be written as 
\begin{equation}
\label{J1bound}
J_1=\sum_{{X}\ni x,y}\sum_{\tilde{X}: X\rightarrow \tilde{X}}\sum_{I=I(X):X\rightarrow\tilde{X}}
\sum_{\tilde{I}:X\rightarrow \tilde{X},I(X)\rightarrow\tilde{I}}|\tilde{X}|^3
\Vert \tilde{V}_{\tilde{X}}^{(\tilde{I})}\Vert.
\end{equation}
The sums in the right-hand side except for the first sum generate all of $\tilde{V}_{\tilde{X}}^{(\tilde{I})}$ 
{from} a given $V_X$. Therefore, we have 
\begin{eqnarray}
\label{sumtildeVtildeXtildeIbound}
& &\sum_{\tilde{X}: X\rightarrow \tilde{X}}\sum_{I=I(X):X\rightarrow\tilde{X}}
\sum_{\tilde{I}:X\rightarrow \tilde{X},I(X)\rightarrow\tilde{I}}|\tilde{X}|^3
\Vert \tilde{V}_{\tilde{X}}^{(\tilde{I})}\Vert\nonumber\\
&\le& \left[\sqrt{2}+\frac{1}{2}+(\sqrt{2}+1)\hat{K}_h+\frac{1}{2}\hat{K}_h^2\right]^{|X|}|\tilde{X}|^3\Vert V_X\Vert
\end{eqnarray}
by using the inequality (\ref{tildeVtildeXtildeIbound}), where the constant $\hat{K}_h$ is given by 
$$
\hat{K}_h:=\sum_{w\in\ze^d} C_h e^{-m_h{\rm dist}(0,w)}.
$$
For short, we write 
$$
\tilde{K}_h:=\sqrt{2}+\frac{1}{2}+(\sqrt{2}+1)\hat{K}_h+\frac{1}{2}\hat{K}_h^2.
$$
By substituting this result (\ref{sumtildeVtildeXtildeIbound}) into (\ref{J1bound}), we obtain 
\begin{equation}
J_1\le \sum_{X\ni x,y}(\tilde{K}_h)^{|X|}|\tilde{X}|^3\Vert V_X\Vert.
\end{equation}
Note that there exist positive constants, $C_h'$ and $K_h>\hat{K}_h$, such that 
$$
4(\tilde{K}_h)^{|X|}|{X}|^3\le C_h'(K_h)^{|X|}. 
$$
Combining this, $|\tilde{X}|\le 2|X|$ and Assumption~\ref{AssumptionVX}, we have 
\begin{equation}
J_1\le C_h'C_{h,V}e^{-m_{h,V}{\rm dist}(x,y)}. 
\end{equation}
 
Next, consider the quantity $J_2$. In this case, since $y\in \tilde{X}\backslash X$, 
the right-hand side of (\ref{tildeVtildeXtildeIexpr}) contains an operator $\tilde{\omega}_{2,z,y}^\bullet$ with $z\in X$.  
Except for this, we can deal with this case in the same way. As a result, we have 
\begin{eqnarray}
J_2&\le& \sum_{X\ni x}\sum_{z\in X}(\tilde{K}_h)^{|X|}|\tilde{X}|^3\Vert V_X\Vert
\frac{(\sqrt{2}+1)+K_h}{\tilde{K}_h}C_he^{-m_h{\rm dist}(z,y)}\nonumber\\
&\le&2\sum_{z\in \Lambda} \sum_{X\ni x,z}(\tilde{K}_h)^{|X|}|\tilde{X}|^3\Vert V_X\Vert e^{-m_h{\rm dist}(z,y)}\nonumber\\
&\le&2\sum_{z\in \ze^d}C_h'C_{h,V}e^{-m_{h,V}{\rm dist}(x,z)}e^{-m_h{\rm dist}(z,y)}.  
\end{eqnarray}
In order to estimate the last sum in the right-hand side, we use the inequality \cite{HK},    
\begin{equation}
\sum_{z\in \ze^d}e^{-m_{h,V}{\rm dist}(x,z)}e^{-m_h{\rm dist}(z,y)}\le K_{h,V}e^{-\tilde{m}_{h,V}{\rm dist}(x,y)},
\end{equation}
where $K_{h,V}$ is a positive constant, and $\tilde{m}_{h,V}$ is a positive constant 
which satisfies $\tilde{m}_{h,V}<\min\{m_{h,V},m_h\}$. Substituting this into the above right-hand side, we obtain 
\begin{equation}
J_2\le 2K_{h,V}C_h'C_{h,V}e^{-\tilde{m}_{h,V}{\rm dist}(x,y)}.
\end{equation}
Clearly, one has 
\begin{equation}
J_3\le 2K_{h,V}C_h'C_{h,V}e^{-\tilde{m}_{h,V}{\rm dist}(x,y)}
\end{equation}
in the same way. As to $J_4$, the right-hand side of (\ref{tildeVtildeXtildeIexpr}) contains two operators, 
$\tilde{\omega}_{2,z,x}^\bullet$ and $\tilde{\omega}_{2,w,y}^\bullet$ with $z,w\in X$. Therefore, 
\begin{eqnarray}
J_4&\le&4\sum_X \sum_{z,w\in X}(\tilde{K}_h)^{|X|}|\tilde{X}|^3\Vert V_X\Vert e^{-m_h{\rm dist}(x,z)}e^{-m_h{\rm dist}(w,y)}\nonumber\\
&\le& 4\sum_{z,w\in\Lambda}\sum_{X\ni z,w}(\tilde{K}_h)^{|X|}|\tilde{X}|^3\Vert V_X\Vert 
e^{-m_h{\rm dist}(x,z)}e^{-m_h{\rm dist}(w,y)}\nonumber\\
&\le& 4\sum_{z,w\in\ze^d} C_h'C_{h,V} e^{-m_{h,V}{\rm dist}(z,w)}e^{-m_h{\rm dist}(x,z)}e^{-m_h{\rm dist}(w,y)}\nonumber\\
&\le& 4K_{h,V}'C_h'C_{h,V}e^{-\tilde{m}_{h,V}{\rm dist}(x,y)}
\end{eqnarray}
with a positive constant $K_{h,V}'$ and the same positive constant $\tilde{m}_{h,V}$. 
Consequently, from these observations, we obtain 
\begin{equation}
\label{boundedtildeVX}
\sum_{\tilde{X}\ni x,y}|\tilde{X}|^3\Vert\tilde{V}_{\tilde X}\Vert\le \tilde{C}_{h,V}e^{-\tilde{m}_{h,V}{\rm dist}(x,y)}
\end{equation}
with a positive constant $\tilde{C}_{h,V}$. Clearly, the transformed potential $\tilde{V}$ is written 
$$
\tilde{V}=\sum_{\tilde{X}\subset\Lambda}\tilde{V}_{\tilde{X}}.
$$
The bound (\ref{boundedtildeVX}) implies \cite{BMNS} 
\begin{equation}
\label{boundedtildeVX2}
\sum_{\tilde{X}\ni x}|\tilde{X}|^3\Vert\tilde{V}_{\tilde X}\Vert\le \tilde{C}_{h,V}.
\end{equation}

\Section{Local approximation of observables}
\label{Sec:LocaApproObs}

In this appendix, we recall the method by \cite{BMNS,NSY} which approximates an observable 
with a large support by an observable with a smaller support. 

We write $q_x=\xi_x^\dagger \xi_x$ and $\overline{q}_x=\xi_x\xi_x^\dagger$. Then, one has $(1-2q_x)^2=1$, and 
\begin{equation}
(1-2q_x)\xi_x(1-2q_x)=-\xi_x, \quad \mbox{and} \quad (1-2q_x)\xi_x^\dagger(1-2q_x)=-\xi_x^\dagger. 
\end{equation}
Further, $(\xi_x+\xi_x^\dagger)^2=1$, and 
\begin{equation}
(\xi_x+\xi_x^\dagger)\xi_x(\xi_x+\xi_x^\dagger)=\xi_x^\dagger, \quad \mbox{and}\quad  
(\xi_x+\xi_x^\dagger)\xi_x^\dagger(\xi_x+\xi_x^\dagger)=\xi_x.
\end{equation}
These imply 
\begin{equation}
(\xi_x+\xi_x^\dagger)q_x(\xi_x+\xi_x^\dagger)=\overline{q}_x=1-q_x, \quad \mbox{and}\quad  
(\xi_x+\xi_x^\dagger)\overline{q}_x(\xi_x+\xi_x^\dagger)=q_x.
\end{equation}
Let $\sigma\in\{(0,0),(1,0),(0,1),(1,1)\}$, and define 
\begin{equation}
\tilde{U}_x(\sigma)=\begin{cases}1, & \mbox{for \ } \sigma=(0,0); \\ 
                (1-2q_x), & \mbox{for \ } \sigma=(1,0); \\
                 (\xi_x+\xi_x^\dagger), & \mbox{for \ } \sigma=(0,1);\\  
                 (1-2q_x)(\xi_x+\xi_x^\dagger), & \mbox{for \ }\sigma=(1,1).\\ \end{cases}
\end{equation} 
We also define 
\begin{equation} 
\overline{\Pi}_x(\mathcal{A}_x):=\frac{1}{4}\sum_\sigma \tilde{U}_x^\ast(\sigma)\mathcal{A}_x\tilde{U}_x(\sigma)
\end{equation}
for an operator $\mathcal{A}_x$ on the single site $x$. Then, one has 
$$
\overline{\Pi}_x(\mathcal{A}_x)=0 \quad \mbox{for \ } \mathcal{A}_x=\xi_x \ \mbox{or \ } \xi_x^\dagger, 
$$
and 
$$
\overline{\Pi}_x(\mathcal{A}_x)=\frac{1}{2} \quad \mbox{for \ } \mathcal{A}_x=q_x \ \mbox{or \ } \overline{q}_x.
$$
For a general site set $\tilde{X}\subset\tilde{\Lambda}$ and an operator $\mathcal{A}$ with even parity, we define 
$$
\overline{\Pi}_{\tilde{X}}(\mathcal{A}):=\left(\frac{1}{4}\right)^{|\tilde{X}|}\sum_{\sigma_{x_1},\sigma_{x_2},\ldots,
\sigma_{|\tilde{X}|}}
\tilde{U}_{\tilde{X}}^\ast(\sigma_{x_1},\sigma_{x_2},\ldots,\sigma_{x_{|\tilde{X}|}})
\mathcal{A}\tilde{U}_{\tilde{X}}(\sigma_{x_1},\sigma_{x_2},\ldots,\sigma_{x_{|\tilde{X}|}}), 
$$
with 
$$
\tilde{U}_{\tilde{X}}(\sigma_{x_1},\sigma_{x_2},\ldots,\sigma_{x_{|\tilde{X}|}})
:=\tilde{U}_{x_1}(\sigma_{x_1})\tilde{U}_{x_2}(\sigma_{x_2})\cdots\tilde{U}_{x_{|\tilde{X}|}}(\sigma_{x_{|\tilde{X}|}}),
$$
where $\{x_1,x_2,\ldots,x_{|\tilde{X}|}\}=\tilde{X}$. 

\begin{lemma}
Let $\mathcal{A}$ be an operator with even parity. 
Then, ${\rm supp}\; \overline{\Pi}_{\tilde{X}}(\mathcal{A})\subset \tilde{\Lambda}\backslash \tilde{X}$. 
\end{lemma}

\begin{proof}{Proof}
It is enough to consider the case that the operator $\mathcal{A}$ is a monomial. 
We assume $x\in \tilde{X}$. In this case, 
the operator $\mathcal{A}$ can be written 
$$
\mathcal{A}=\mathcal{A}'\mathcal{A}_x\mathcal{A}'',
$$
where ${\rm supp}\;\mathcal{A}'$, ${\rm supp}\;\mathcal{A}_x=\{x\}$ and ${\rm supp}\;\mathcal{A}''$ are 
disjoint with each other.  

When $\mathcal{A}_x=\xi_x$ or $\xi_x^\dagger$, the operation for $(1-2q_x)$ changes $\mathcal{A}_x$ to 
$-\mathcal{A}_x$. This operation causes a cancellation, and therefore the contribution is vanishing.  

When $\mathcal{A}_x=q_x$ or $\overline{q}_x$, the parity of $\mathcal{A}'$ is equal to that of $\mathcal{A}''$. Then, one has 
\begin{eqnarray*}
\mathcal{A}+(\xi_x+\xi_x^\dagger)\mathcal{A}(\xi_x+\xi_x^\dagger)
&=&\mathcal{A}+(\xi_x+\xi_x^\dagger)\mathcal{A}'\mathcal{A}_x\mathcal{A}''(\xi_x+\xi_x^\dagger)\\
&=&\mathcal{A}'\mathcal{A}_x\mathcal{A}''
+\mathcal{A}'(\xi_x+\xi_x^\dagger)\mathcal{A}_x(\xi_x+\xi_x^\dagger)\mathcal{A}''\\
&=&\mathcal{A}'[\mathcal{A}_x+(\xi_x+\xi_x^\dagger)\mathcal{A}_x(\xi_x+\xi_x^\dagger)]\mathcal{A}''\\
&=&\mathcal{A}'\mathcal{A}''.
\end{eqnarray*}
\end{proof}

\noindent
{\it Remark:} This property does not hold for an odd parity operator $\mathcal{A}$. For example, 
consider $\mathcal{A}=q_x\xi_y$ for $x\ne y$, and $\overline{\Pi}_{\{x\}}$ on the single site $x$. Actually, one has 
\begin{eqnarray*}
q_x\xi_y+(\xi_x+\xi_x^\dagger)q_x\xi_y(\xi_x+\xi_x^\dagger)&=&q_x\xi_y-(\xi_x+\xi_x^\dagger)q_x(\xi_x+\xi_x^\dagger)\xi_y\\
&=&q_x\xi_y-(1-q_x)\xi_y=(2q_x-1)\xi_y.
\end{eqnarray*}
\medskip

We define 
\begin{equation}
\tilde{\Pi}_{\tilde{X}}:=\overline{\Pi}_{\tilde{\Lambda}\backslash \tilde{X}}.
\end{equation}
Then, one has ${\rm supp}\;\tilde{\Pi}_{\tilde{X}}(\mathcal{A})\subset \tilde{X}$. Roughly speaking, we can cut the tail of 
the observables. 

The following lemma is a fermion analogue of Lemma~3.2 in \cite{BMNS}:

\begin{lemma} {\bf \cite{NSY}}
\label{fermionPilem}
Let $\mathcal{A}$ be a fermion operator with even parity such that  
$$
\Vert [\mathcal{A},\mathcal{B}]\Vert \le \epsilon \Vert \mathcal{B}\Vert 
$$
with a small $\epsilon\ge 0$, for all the operators $\mathcal{B}$ 
satisfying ${\rm supp}\;\mathcal{B}\subset \tilde{\Lambda}\backslash \tilde{X}$. Then, 
\begin{equation}
\Vert \mathcal{A}-\tilde{\Pi}_{\tilde{X}}(\mathcal{A})\Vert \le \epsilon.
\end{equation}
\end{lemma}

\begin{proof}{Proof}
We write $\tilde{Y}=\tilde{\Lambda}\backslash \tilde{X}$ for short. {From} the definition of $\tilde{\Pi}_{\tilde{X}}$, one has 
\begin{eqnarray*}
\tilde{\Pi}_{\tilde{X}}(\mathcal{A})-\mathcal{A}&=&\left(\frac{1}{4}\right)^{|\tilde{Y}|}
\sum_{\sigma_{x_1},\sigma_{x_2},\ldots,\sigma_{|\tilde{Y}|}}\left[
\tilde{U}_{\tilde{Y}}^\ast(\sigma_{x_1},\sigma_{x_2},\ldots,\sigma_{x_{|\tilde{Y}|}})
\mathcal{A}\tilde{U}_{\tilde{Y}}(\sigma_{x_1},\sigma_{x_2},\ldots,\sigma_{x_{|\tilde{Y}|}})-\mathcal{A}\right]\\
&=&\left(\frac{1}{4}\right)^{|\tilde{Y}|}
\sum_{\sigma_{x_1},\sigma_{x_2},\ldots,\sigma_{|\tilde{Y}|}}
\tilde{U}_Y^\ast(\sigma_{x_1},\sigma_{x_2},\ldots,\sigma_{x_{|\tilde{Y}|}})
\left[\mathcal{A},\tilde{U}_{\tilde{Y}}(\sigma_{x_1},\sigma_{x_2},\ldots,\sigma_{x_{|\tilde{Y}|}})\right].
\end{eqnarray*}
Therefore, 
\begin{eqnarray*}
\Vert \tilde{\Pi}_{\tilde{X}}(\mathcal{A})-\mathcal{A}\Vert&\le&\left(\frac{1}{4}\right)^{|\tilde{Y}|}
\sum_{\sigma_{x_1},\sigma_{x_2},\ldots,\sigma_{|Y|}}
\Vert\tilde{U}_{\tilde{Y}}^\ast(\sigma_{x_1},\sigma_{x_2},\ldots,\sigma_{x_{|\tilde{Y}|}})\Vert\\
&\times&\left\Vert\left[\mathcal{A},\tilde{U}_{\tilde{Y}}(\sigma_{x_1},\sigma_{x_2},\ldots,\sigma_{x_{|\tilde{Y}|}})\right]\right
\Vert\le\epsilon.
\end{eqnarray*}
\end{proof}

\Section{Local approximation of the time evolution $\tilde{\tau}_{s,t}$ by the Hamiltonian $\tilde{\mathcal{H}}(s)$}
\label{Sec:UVZUDelta}

Our aim of this appendix is to obtain a local approximation bound (\ref{taudifPi}) below for 
the time evolution $\tilde{\tau}_{s,t}$ by the Hamiltonian $\tilde{\mathcal{H}}(s)$. 
The Hamiltonian $\tilde{\mathcal{H}}(s)$ is written 
\begin{equation}
\tilde{\mathcal{H}}(s)=\sum_Z\tilde{U}^\ast(s)\tilde{\mathcal{H}}_{0,Z}\tilde{U}(s)+s\sum_Z\tilde{U}^\ast(s)\tilde{V}_Z\tilde{U}(s).
\end{equation}
In order to treat the time evolution $\tilde{\tau}_{s,t}$ by the Hamiltonian $\tilde{\mathcal{H}}(s)$,  
we need to use local approximation of the two operators, $\tilde{U}^\ast(s)\tilde{\mathcal{H}}_{0,Z}\tilde{U}(s)$ 
and $\tilde{U}^\ast(s)\tilde{V}_Z\tilde{U}(s)$, in the manner of the preceding section. 

Consider first the transformed interactions $\tilde{U}^\ast(s)\tilde{V}_Z\tilde{U}(s)$. 
Clearly, the operation of $\tilde{U}(s)$ enlarges the support $Z$ of $\tilde{V}_Z$ to a large region. 
Let $n\in\{0,1,2,\ldots\}$, and define  
$$
Z_n:=\{z\in{\Lambda}\;|\;{\rm dist}(z,Z)\le n\} 
$$
for $Z\subset{\Lambda}$. Following \cite{BMNS}, we can write 
\begin{equation}
\tilde{U}^\ast(s)\tilde{V}_Z\tilde{U}(s)=\sum_{n=0}^\infty \Delta_{\tilde{U}}^n(\tilde{V}_Z,s),
\end{equation}
where 
\begin{equation}
\label{Deltan}
\Delta_{\tilde{U}}^n(\tilde{V}_Z,s):=\Pi_{Z_n}(\tilde{U}^\ast(s)\tilde{V}_Z\tilde{U}(s))
-\Pi_{Z_{n-1}}(\tilde{U}^\ast(s)\tilde{V}_Z\tilde{U}(s))\quad \mbox{for \ } n\ge 1
\end{equation}
and 
\begin{equation}
\label{Delta0}
\Delta_{\tilde{U}}^0(\tilde{V}_Z,s):=\Pi_Z(\tilde{U}^\ast(s)\tilde{V}_Z\tilde{U}(s)),
\end{equation}
where 
$$
\Pi_Z:=\tilde{\Pi}_{\tilde{Z}}\quad \mbox{with \ } \tilde{Z}=Z\times\{c,d\}. 
$$ 
Therefore, the corresponding interaction Hamiltonian is written 
$$
\sum_{Z\subset{\Lambda}}\tilde{U}^\ast(s)\tilde{V}_Z\tilde{U}(s)
=\sum_{Z\subset{\Lambda}}\sum_{n=0}^\infty \Delta_{\tilde{U}}^n(\tilde{V}_Z,s).
$$
In order to derive the Lieb-Robinson bound for the time evolution $\tilde{\tau}_{s,t'}$, we need to estimate 
\begin{equation}
\label{Deltasum}
\sum_{Z\subset{\Lambda}}\;\sum_{n=0:\;x,y\in Z_n}^\infty \Vert\Delta_{\tilde{U}}^n(\tilde{V}_Z,s)\Vert
\end{equation}
for the summand $\Delta_{\tilde{U}}^n(\tilde{V}_Z,s)$ in the above right-hand side, where $x$ and $y$ are 
any given two sites, $x,y\in{\Lambda}$. 

Since the Hamiltonian $\tilde{\mathcal{H}}_s$ of (\ref{tildecalHs}) satisfies 
the exponential decay conditions, (\ref{tildecalH0Zexpdecay}) and (\ref{boundedtildeVX}), 
the application of Theorem~4.5 in \cite{BMNS} yields   
\begin{equation}
\Vert[\tilde{U}^\ast(s)\mathcal{A}\tilde{U}(s),\mathcal{B}] \Vert
\le C_{\tilde{U}}\Vert \mathcal{A}\Vert\Vert\mathcal{B}\Vert 
\sum_{x\in{\rm supp}\;\mathcal{A},\;y\in{\rm supp}\;\mathcal{B}}F_{\mu_0,r_0}({\rm dist}(x,y)),
\end{equation}
where the positive constant $C_{\tilde{U}}$ is determined by the unitary operator $\tilde{U}(s)$ 
or the Hamiltonian $\tilde{\mathcal{H}}_s$, 
and the function $F_{\mu_0,r_0}(r)$ is given by 
\begin{equation}
\label{Fmu0r0}
F_{\mu_0,r_0}(r):=\frac{\tilde{u}_{\mu_0}(r/r_0)}{[1+(r/r_0)]^{d+1}}.
\end{equation}
Here, $\mu_0$ and $r_0$ are positive constants which are also determined by 
the Hamiltonian $\tilde{\mathcal{H}}_s$, and the function $\tilde{u}_\mu(r)$ is given by 
\begin{equation}
\label{tildeumu}
\tilde{u}_\mu(r):=\begin{cases}u_\mu(e^2), & \mbox{for \ } 0\le r\le e^2; \\ 
                 u_\mu(r), & \mbox{otherwise}.\\ \end{cases}
\end{equation}
with 
\begin{equation}
\label{umu}
u_\mu(r):=\exp[-\mu r/(\log r)^2].
\end{equation}
Therefore, by applying Lemma~\ref{fermionPilem}, one has 
\begin{equation}
\left\Vert\tilde{U}^\ast(s)\tilde{V}_Z\tilde{U}(s)-\Pi_{Z_n}(\tilde{U}^\ast(s)\tilde{V}_Z\tilde{U}(s))\right\Vert
\le C_{\tilde{U}}\Vert\tilde{V}_Z\Vert \sum_{x\in Z}\sum_{y\in {\Lambda}\backslash Z_n} 
F_{\mu_0,r_0}({\rm dist}(x,y))
\end{equation}
Note that 
\begin{equation}
\sum_{y\in {\Lambda}\backslash Z_n}F_{\mu_0,r_0}({\rm dist}(x,y))
\le \mathcal{C}_1F_{\mu_1,r_0}(n),
\end{equation}
where the positive constant $\mathcal{C}_1$ and the positive constant $\mu_1\in (0,\mu_0)$ are 
independent of $n$ and ${\Lambda}$. By using this inequality, one has
\begin{equation}
\label{UVUPibound}
\left\Vert\tilde{U}^\ast(s)\tilde{V}_Z\tilde{U}(s)-\Pi_{Z_n}(\tilde{U}^\ast(s)\tilde{V}_Z\tilde{U}(s))\right\Vert
\le C_{\tilde{U}}\mathcal{C}_1|Z|\Vert\tilde{V}_Z\Vert F_{\mu_1,r_0}(n).
\end{equation}
This enables us to estimate $\Delta_{\tilde{U}}^n(\tilde{V}_Z,s)$ of (\ref{Deltan}). Actually, one obtains 
\begin{equation}
\label{Deltanbound}
\Vert\Delta_{\tilde{U}}^n(\tilde{V}_Z,s)\Vert \le 2C_{\tilde{U}}\mathcal{C}_1|Z|\Vert\tilde{V}_Z\Vert F_{\mu_1,r_0}(n-1)
\end{equation}
for $n\ge 1$. As to $\Delta_{\tilde{U}}^0(\tilde{V}_Z,s)$ of (\ref{Delta0}), one has 
\begin{equation}
\label{Delta0bound}
\Vert\Delta_{\tilde{U}}^0(\tilde{V}_Z,s)\Vert \le\Vert\tilde{V}_Z\Vert. 
\end{equation}
By using these inequalities, the quantity (\ref{Deltasum}) is estimated as  
\begin{equation}
\label{Deltasum2}
\sum_{Z\subset\tilde{\Lambda}}\;\sum_{n=0:\;x,y\in Z_n}^\infty \Vert\Delta_{\tilde{U}}^n(\tilde{V}_Z,s)\Vert\le 
\sum_{Z\ni x,y} \Vert\tilde{V}_Z\Vert+2C_{\tilde{U}}\mathcal{C}_1\sum_Z \;\sum_{n=1:\;x,y\in Z_n}^\infty 
|Z|\Vert\tilde{V}_Z\Vert F_{\mu_1,r_0}(n-1). 
\end{equation}
The first sum in the right-hand side satisfies the desired bound, 
\begin{equation}
\label{Deltasum21}
\sum_{Z\ni x,y} \Vert\tilde{V}_Z\Vert\le \tilde{C}_{h,V}e^{-\tilde{m}_{h,V}{\rm dist}(x,y)},
\end{equation}
{from} (\ref{boundedtildeVX}). 

Therefore, it is enough to estimate the second double sum in the right-hand side of (\ref{Deltasum2}). 
Since $x\in Z_n$, there exists $u\in Z$ such that ${\rm dist}(x,u)={\rm dist}(x,Z)\le n$.
Similarly, since $y\in Z_n$, there exists $v\in Z$ such that ${\rm dist}(y,v)={\rm dist}(y,Z)\le n$. 
Therefore, we have 
$$
{\rm dist}(x,y)\le {\rm dist}(x,u)+{\rm dist}(u,v)+{\rm dist}(v,y)\le 2n+{\rm dist}(u,v).
$$ 
{From} these observations, the second double sum in the right-hand side of (\ref{Deltasum2}) is estimated by 
\begin{eqnarray}
\label{Deltasum22}
& &\sum_Z \;\sum_{n=1:\;x,y\in Z_n}^\infty |Z|\Vert\tilde{V}_Z\Vert F_{\mu_1,r_0}(n-1)\nonumber\\
&\le& \sum_{u,v}\sum_{Z\ni u,v}\sum_{n=1}^\infty |Z|\Vert\tilde{V}_Z\Vert F_{\mu_1,r_0}(n-1)
{\rm Ind}[{\rm dist}(x,u)\le n]\cdot{\rm Ind}[{\rm dist}(y,v)\le n]\nonumber\\
&\times&{\rm Ind}[{\rm dist}(x,y)\le 2n+{\rm dist}(u,v)]\nonumber\\
&\le& \tilde{C}_{h,V}\sum_{u,v}\sum_{n=1}^\infty e^{-\tilde{m}_{h,V}{\rm dist}(u,v)} F_{\mu_1,r_0}(n-1)
{\rm Ind}[{\rm dist}(x,u)\le n]\cdot{\rm Ind}[{\rm dist}(y,v)\le n]\nonumber\\
&\times&{\rm Ind}[{\rm dist}(x,y)\le 2n+{\rm dist}(u,v)],
\end{eqnarray}
where we have used (\ref{boundedtildeVX}). 

In order to estimate the right-hand side of (\ref{Deltasum22}), we note the following: 
$$ 
F_{\mu_1,r_0}(n-1)=\tilde{u}_\delta((n-1)/r_0)F_{\mu_1-\delta,r_0}(n-1)
$$
with $\delta>0$ satisfying $\mu_1-\delta>0$. From ${\rm dist}(x,y)\le 2n+{\rm dist}(u,v)$ and 
$$
\frac{a}{\log^2 a}+\frac{b}{\log^2 b}\ge \frac{a+b}{\log^2(a+b)}
$$
for $a,b>0$, one has 
\begin{equation}
e^{-\tilde{m}_{h,V}{\rm dist}(u,v)}F_{\mu_1-\delta,r_0}(n-1)\le \mathcal{C}_2
F_{\mu_2,2r_0}({\rm dist}(x,y)),
\end{equation}
where $\mu_2$ is a positive constant satisfying $\mu_2<\mu_1-\delta$, and $\mathcal{C}_2$ is a positive constant. 
Substituting these into the right-hand side of (\ref{Deltasum22}), we obtain 
\begin{eqnarray}
& &\sum_{u,v}\sum_{n=1}^\infty e^{-\tilde{m}_{h,V}{\rm dist}(u,v)} F_{\mu_1,r_0}(n-1)
{\rm Ind}[{\rm dist}(x,u)\le n]\cdot{\rm Ind}[{\rm dist}(y,v)\le n]\nonumber\\
&\times&{\rm Ind}[{\rm dist}(x,y)\le 2n+{\rm dist}(u,v)]\nonumber\\
&\le& \mathcal{C}_2\sum_{u,v}\sum_{n=1}^\infty \tilde{u}_\delta((n-1)/r_0)F_{\mu_2,2r_0}({\rm dist}(x,y))
{\rm Ind}[{\rm dist}(x,u)\le n]\cdot{\rm Ind}[{\rm dist}(y,v)\le n]\nonumber\\
&\le&\mathcal{C}_2(\kappa_d)^2\sum_{n=1}^\infty n^{2d}\tilde{u}_\delta((n-1)/r_0)\cdot F_{\mu_2,2r_0}({\rm dist}(x,y)), 
\end{eqnarray}
where $\kappa_d$ is the positive constant which is related to the volume of the $d$-dimensional ball with radius $r$, 
which is given by $\kappa_d r^d$. 
Combining this, (\ref{Deltasum2}), (\ref{Deltasum21}) and (\ref{Deltasum22}), we obtain 
\begin{equation}
\sum_{Z\subset{\Lambda}}\;\sum_{n=0:\;x,y\in Z_n}^\infty \Vert\Delta_{\tilde{U}}^n(\tilde{V}_Z,s)\Vert
\le C_{\tilde{U},\tilde{V}}F_{\mu_2,2r_0}({\rm dist}(x,y))
\end{equation}
for the interaction $\tilde{V}(s)$ of (\ref{tildeVs}). 
Here, the positive constant $C_{\tilde{U},\tilde{V}}$ can be taken to be independent of the parameter $s$. 

Clearly, we can deal with $\tilde{\mathcal{H}}_0(s)$ in the same way. 
Therefore, we can apply Theorem~4.6 in \cite{BMNS} for the time evolution by the Hamiltonian $\tilde{\mathcal{H}}(s)$. 
As a result, we obtain the Lieb-Robinson bound, 
\begin{equation}
\left\Vert[\tilde{\tau}_{s,t}(\mathcal{A}),\mathcal{B}] \right\Vert 
\le C_{\tilde{U},\tilde{\mathcal{H}}}\Vert\mathcal{A}\Vert \Vert \mathcal{B}\Vert\times \exp[v_{\rm LR}|t|]
\sum_{x\in {\rm supp}\;\mathcal{A},\; y\in{\rm supp}\;\mathcal{B}}F_{\mu_2,2r_0}({\rm dist}(x,y)),
\end{equation}
where $C_{\tilde{U},\tilde{\mathcal{H}}}$ is a positive constant, and $v_{\rm LR}$ is the Lieb-Robinson velocity. 
Clearly, these constants can be chosen to be independent of the parameter $s$.   
Therefore, from Lemma~\ref{fermionPilem}, we obtain the desired result,  
\begin{equation}
\label{taudifPi}
\Vert\tilde{\tau}_{s,t}(\mathcal{A})-\Pi_{X_n}(\tilde{\tau}_{s,t}(\mathcal{A}))\Vert
\le C_{\tilde{U},\tilde{\mathcal{H}}}\Vert\mathcal{A}\Vert \times \exp[v_{\rm LR}|t|]
\sum_{x\in {\rm supp}\;\mathcal{A}}\;\sum_{y\in{\Lambda}\backslash X_n}F_{\mu_2,2r_0}({\rm dist}(x,y)).
\end{equation}


By using the inequality (\ref{taudifPi}), one can obtain local approximation of $\mathcal{R}_s$ of (\ref{calRs}), 
which is given by 
$$
\mathcal{R}_s(\cdots)=\int_{-\infty}^{+\infty}dt\;w_\gamma(t)\tilde{\tau}_{s,t}(\cdots).
$$
Note that 
\begin{eqnarray*}
\Vert\mathcal{R}_s(\mathcal{A})-\Pi_{X_n}(\mathcal{R}_s(\mathcal{A}))\Vert
&\le&\int_{-\infty}^{+\infty}dt\;w_\gamma(t)\Vert\tilde{\tau}_{s,t}(\mathcal{A})-\Pi_{X_n}(\tilde{\tau}_{s,t}(\mathcal{A}))\Vert\\
&\le&\int_{-T}^{+T}dt\;w_\gamma(t)\Vert\tilde{\tau}_{s,t}(\mathcal{A})-\Pi_{X_n}(\tilde{\tau}_{s,t}(\mathcal{A}))\Vert\\
&+&\int_{|t|\ge T}dt\;w_\gamma(t)\Vert\tilde{\tau}_{s,t}(\mathcal{A})-\Pi_{X_n}(\tilde{\tau}_{s,t}(\mathcal{A}))\Vert
\end{eqnarray*}
for any $T>0$. The second integral is estimated by 
\begin{equation}
\int_{|t|\ge T}dt\;w_\gamma(t)\Vert\tilde{\tau}_{s,t}(\mathcal{A})-\Pi_{X_n}(\tilde{\tau}_{s,t}(\mathcal{A}))\Vert
\le 4 \Vert\mathcal{A}\Vert\int_T^\infty dt\;w_\gamma(t).
\end{equation}
For the first integral, one has  
\begin{eqnarray*}
& &\int_{-T}^{+T}dt\;w_\gamma(t)\Vert\tilde{\tau}_{s,t}(\mathcal{A})-\Pi_{X_n}(\tilde{\tau}_{s,t}(\mathcal{A}))\Vert\\
&\le& 2C_{\tilde{U},\tilde{\mathcal{H}}}\Vert\mathcal{A}\Vert \times 
\sum_{x\in {\rm supp}\;\mathcal{A}}\;\sum_{y\in{\Lambda}\backslash X_n}F_{\mu_2,2r_0}({\rm dist}(x,y))
\int_{0}^{+T}dt\;w_\gamma(t)\exp[v_{\rm LR}t]\\
&\le& \frac{2c_\gamma C_{\tilde{U},\tilde{\mathcal{H}}}}{v_{\rm LR}}\Vert\mathcal{A}\Vert \times 
\sum_{x\in {\rm supp}\;\mathcal{A}}\;\sum_{y\in{\Lambda}\backslash X_n}F_{\mu_2,2r_0}({\rm dist}(x,y))
\exp[v_{\rm LR}T]
\end{eqnarray*}
by using (\ref{taudifPi}) and the bound $w_\gamma(t)\le c_\gamma$ with a positive constant $c_\gamma$ for all $t$ \cite{BMNS}. 
In consequence, we obtain 
\begin{eqnarray}
\label{calRsPibound}
& &\Vert\mathcal{R}_s(\mathcal{A})-\Pi_{X_n}(\mathcal{R}_s(\mathcal{A}))\Vert\nonumber\\
&\le& 4 \Vert\mathcal{A}\Vert\int_T^\infty dt\;w_\gamma(t)
+\frac{2c_\gamma C_{\tilde{U},\tilde{\mathcal{H}}}}{v_{\rm LR}}\Vert\mathcal{A}\Vert \times 
\sum_{x\in {\rm supp}\;\mathcal{A}}\;\sum_{y\in{\Lambda}\backslash X_n}F_{\mu_2,2r_0}({\rm dist}(x,y))
\exp[v_{\rm LR}T]\nonumber\\
\end{eqnarray}
for any positive $T$. 

\Section{Local boundedness of the interaction $\tilde{\mathcal{W}}^{(1)}(s)$} 

In order to prove the statement of Theorem~\ref{mainTheorem}, 
we have to prove the inequalities (\ref{tildecalWZnsboundedness}). 
In this appendix, we will treat $\tilde{\mathcal{W}}_Z^{(1)}(s)$, which is given by 
$$
\tilde{\mathcal{W}}_Z^{(1)}(s)=s\mathcal{R}_s(:\tilde{U}^\ast(s)\tilde{V}_Z\tilde{U}(s):).
$$ 

To begin with, we define 
$$
\mathcal{Z}_n:=\{z\in{\Lambda}\;|\; {\rm dist}(z,Z_n)\le n\},
$$
where $Z_n=\{z\in\Lambda\;|\; {\rm dist}(z,Z)\le n\}$. Note that
\begin{eqnarray}
\label{RUVZdifPi}
& &\left\Vert\mathcal{R}_s(\tilde{U}^\ast(s)\tilde{V}_Z\tilde{U}(s))
-\Pi_{\mathcal{Z}_n}(\mathcal{R}_s(\tilde{U}^\ast(s)\tilde{V}_Z\tilde{U}(s)))\right\Vert\nonumber\\
&\le&\left\Vert \mathcal{R}_s(\tilde{U}^\ast(s)\tilde{V}_Z\tilde{U}(s))
-\mathcal{R}_s(\Pi_{Z_n}(\tilde{U}^\ast(s)\tilde{V}_Z\tilde{U}(s)))\right\Vert\nonumber\\
&+&\left\Vert\mathcal{R}_s(\Pi_{Z_n}(\tilde{U}^\ast(s)\tilde{V}_Z\tilde{U}(s)))
-\Pi_{\mathcal{Z}_n}(\mathcal{R}_s(\Pi_{Z_n}(\tilde{U}^\ast(s)\tilde{V}_Z\tilde{U}(s))))\right\Vert\nonumber\\
&+&\left\Vert\Pi_{\mathcal{Z}_n}(\mathcal{R}_s(\Pi_{Z_n}(\tilde{U}^\ast(s)\tilde{V}_Z\tilde{U}(s))))
-\Pi_{\mathcal{Z}_n}(\mathcal{R}_s(\tilde{U}^\ast(s)\tilde{V}_Z\tilde{U}(s)))\right\Vert. 
\end{eqnarray}
{From} the inequality (\ref{UVUPibound}), the first term in the right-hand side is estimated as 
\begin{eqnarray}
\label{RUVZdifPi1}
& &\left\Vert \mathcal{R}_s(\tilde{U}^\ast(s)\tilde{V}_Z\tilde{U}(s))
-\mathcal{R}_s(\Pi_{Z_n}(\tilde{U}^\ast(s)\tilde{V}_Z\tilde{U}(s)))\right\Vert\nonumber\\
&\le& 
\left\Vert \tilde{U}^\ast(s)\tilde{V}_Z\tilde{U}(s)-\Pi_{Z_n}(\tilde{U}^\ast(s)\tilde{V}_Z\tilde{U}(s))\right\Vert
\le C_{\tilde{U}}\mathcal{C}_1|Z|\Vert\tilde{V}_Z\Vert F_{\mu_1,r_0}(n).
\end{eqnarray}
In the same way, third term is estimated as 
\begin{eqnarray}
\label{RUVZdifPi3}
& &\left\Vert\Pi_{\mathcal{Z}_n}(\mathcal{R}_s(\Pi_{Z_n}(\tilde{U}^\ast(s)\tilde{V}_Z\tilde{U}(s))))
-\Pi_{\mathcal{Z}_n}(\mathcal{R}_s(\tilde{U}^\ast(s)\tilde{V}_Z\tilde{U}(s)))\right\Vert\nonumber\\
&\le&\left\Vert\mathcal{R}_s(\Pi_{Z_n}(\tilde{U}^\ast(s)\tilde{V}_Z\tilde{U}(s)))
-\mathcal{R}_s(\tilde{U}^\ast(s)\tilde{V}_Z\tilde{U}(s))\right\Vert\nonumber\\
&\le&\left\Vert\Pi_{Z_n}(\tilde{U}^\ast(s)\tilde{V}_Z\tilde{U}(s))
-\tilde{U}^\ast(s)\tilde{V}_Z\tilde{U}(s)\right\Vert\le C_{\tilde{U}}\mathcal{C}_1|Z|\Vert\tilde{V}_Z\Vert F_{\mu_1,r_0}(n). 
\end{eqnarray}
As to the second term, we use the inequality (\ref{calRsPibound}). As a result, we have 
\begin{eqnarray*}
& &\left\Vert\mathcal{R}_s(\Pi_{Z_n}(\tilde{U}^\ast(s)\tilde{V}_Z\tilde{U}(s)))
-\Pi_{\mathcal{Z}_n}(\mathcal{R}_s(\Pi_{Z_n}(\tilde{U}^\ast(s)\tilde{V}_Z\tilde{U}(s))))\right\Vert\\
&\le&4 \Vert\tilde{V}_Z\Vert\int_T^\infty dt\;w_\gamma(t)
+\frac{2c_\gamma C_{\tilde{U},\tilde{\mathcal{H}}}}{v_{\rm LR}}\Vert\tilde{V}_Z\Vert \times 
\sum_{x\in Z_n}\;\sum_{y\in{\Lambda}\backslash \mathcal{Z}_n}F_{\mu_2,2r_0}({\rm dist}(x,y))
\exp[v_{\rm LR}T]\\
&\le&4 \Vert\tilde{V}_Z\Vert\int_T^\infty dt\;w_\gamma(t)
+\frac{2c_\gamma C_{\tilde{U},\tilde{\mathcal{H}}}\mathcal{C}_3}{v_{\rm LR}}\Vert\tilde{V}_Z\Vert\cdot 
|Z_n|\cdot F_{\mu_3,2r_0}(n)\exp[v_{\rm LR}T]\\
&\le&4 \Vert\tilde{V}_Z\Vert\int_T^\infty dt\;w_\gamma(t)
+\frac{2c_\gamma C_{\tilde{U},\tilde{\mathcal{H}}}\mathcal{C}_3\kappa_d}{v_{\rm LR}}\Vert\tilde{V}_Z\Vert\cdot 
|Z|n^d\cdot F_{\mu_3,2r_0}(n)\exp[v_{\rm LR}T],
\end{eqnarray*}
where we have used 
\begin{equation}
\label{sumFmu2bound}
\sum_{y\in{\Lambda}\backslash \mathcal{Z}_n}F_{\mu_2,2r_0}({\rm dist}(x,y))\le\mathcal{C}_3F_{\mu_3,2r_0}(n)
\end{equation}
with a constant $\mu_3\in(0,\mu_2)$ and a positive constant $\mathcal{C}_3$, and we have also used 
$|Z_n|\le\kappa_d|Z|n^d$. {From} (\ref{Fmu0r0}), (\ref{tildeumu}) and (\ref{umu}), one has 
$$
F_{\mu_3,2r_0}(r)\exp[v_{\rm LR}T]
=\frac{1}{[1+(r/2r_0)]^{d+1}}\exp[-\mu_3(r/2r_0)/(\log(r/2r_0))^2]\exp[v_{\rm LR}T]
$$
for $r\ge 2e^2r_0$. We choose $T$ as a function of $n$,  
\begin{equation}
\label{T(n)}
T=T(n):=\frac{\mu_3}{2v_{\rm LR}}\times\begin{cases}(n/2r_0)/(\log(n/2r_0))^2, & \mbox{for \ } n\ge 2e^2r_0; \\ 
                 e^2/4, & \mbox{otherwise}.\\ \end{cases}
\end{equation}
Then, one has
\begin{equation}
\label{FexpTnF}
F_{\mu_3,2r_0}(n)\exp[v_{\rm LR}T(n)]=F_{\mu_3/2,2r_0}(n).
\end{equation}
Therefore, we have 
\begin{eqnarray*}
& &\left\Vert\mathcal{R}_s(\Pi_{Z_n}(\tilde{U}^\ast(s)\tilde{V}_Z\tilde{U}(s)))
-\Pi_{\mathcal{Z}_n}(\mathcal{R}_s(\Pi_{Z_n}(\tilde{U}^\ast(s)\tilde{V}_Z\tilde{U}(s))))\right\Vert\\
&\le&4 \Vert\tilde{V}_Z\Vert\int_{T(n)}^\infty dt\;w_\gamma(t)
+\frac{2c_\gamma C_{\tilde{U},\tilde{\mathcal{H}}}\mathcal{C}_3\kappa_d}{v_{\rm LR}}\Vert\tilde{V}_Z\Vert\cdot 
|Z|n^d\cdot F_{\mu_3/2,2r_0}(n)
\end{eqnarray*}
According to \cite{BMNS}, the integral of the first term is estimated as 
\begin{equation}
\label{WgammaTdecay}
\int_{T}^\infty dt\;w_\gamma(t) \le 35e^2(\gamma T)^4u_{2/7}(\gamma T)
\end{equation}
for $\gamma T\ge 561$. (See the equation~(2.18) in \cite{BMNS}.) We write 
\begin{equation}
\label{WgammaTn}
W_\gamma(T(n))=\int_{T(n)}^\infty dt\;w_\gamma(t),
\end{equation}
following \cite{BMNS}. Then, the above bound is written 
\begin{eqnarray}
\label{calRsPiUVZUdif}
& &\left\Vert\mathcal{R}_s(\Pi_{Z_n}(\tilde{U}^\ast(s)\tilde{V}_Z\tilde{U}(s)))
-\Pi_{\mathcal{Z}_n}(\mathcal{R}_s(\Pi_{Z_n}(\tilde{U}^\ast(s)\tilde{V}_Z\tilde{U}(s))))\right\Vert\nonumber\\
&\le&4\Vert\tilde{V}_Z\Vert\left[W_\gamma(T(n))+\mathcal{C}_4 
|Z|n^d\cdot F_{\mu_3/2,2r_0}(n)\right]
\end{eqnarray}
with the positive constant $\mathcal{C}_4$. Substituting this, (\ref{RUVZdifPi1}) and (\ref{RUVZdifPi3}) 
into (\ref{RUVZdifPi}), we obtain 
\begin{eqnarray}
& &\left\Vert\mathcal{R}_s(\tilde{U}^\ast(s)\tilde{V}_Z\tilde{U}(s))
-\Pi_{\mathcal{Z}_n}(\mathcal{R}_s(\tilde{U}^\ast(s)\tilde{V}_Z\tilde{U}(s)))\right\Vert \nonumber\\
&\le& 2C_{\tilde{U}}\mathcal{C}_1|Z|\Vert\tilde{V}_Z\Vert F_{\mu_1,r_0}(n)
+4\Vert\tilde{V}_Z\Vert\left[W_\gamma(T(n))+\mathcal{C}_4|Z|n^d\cdot F_{\mu_3/2,2r_0}(n)\right].
\end{eqnarray}
{From} (\ref{WgammaTdecay}) and (\ref{WgammaTn}), this bound can be written 
\begin{equation}
\label{RUVZPidifbound}
\left\Vert\mathcal{R}_s(\tilde{U}^\ast(s)\tilde{V}_Z\tilde{U}(s))
-\Pi_{\mathcal{Z}_n}(\mathcal{R}_s(\tilde{U}^\ast(s)\tilde{V}_Z\tilde{U}(s)))\right\Vert
\le \Vert\tilde{V}_Z\Vert\left[G_1(n)+|Z|G_2(n)\right]
\end{equation}
in terms of two sub-exponentially decaying functions, $G_1$ and $G_2$. 

We recall the property of the local interaction $\tilde{\mathcal{W}}_Z^{(1)}(s)$ of (\ref{tildecalWZ1s}) that is given by 
(\ref{RUVZUP0}), i.e.,  
\begin{equation}
\label{RsUVZUP00}
\mathcal{R}_s(:\tilde{U}^\ast(s)\tilde{V}_Z\tilde{U}(s):)\tilde{P}_0(0)=0,
\end{equation}
where $\tilde{P}_0(0)$ is the projection onto the ground state of the unperturbed Hamiltonian $\tilde{\mathcal{H}}_0$. 
Combining this with the above bound (\ref{RUVZPidifbound}), we have
\begin{eqnarray}
& &\left\Vert\Pi_{\mathcal{Z}_n}(\mathcal{R}_s(:\tilde{U}^\ast(s)\tilde{V}_Z\tilde{U}(s):))\tilde{P}_0(0)\right\Vert\nonumber\\
&\le&\left\Vert\left[\mathcal{R}_s(:\tilde{U}^\ast(s)\tilde{V}_Z\tilde{U}(s):)
-\Pi_{\mathcal{Z}_n}(\mathcal{R}_s(:\tilde{U}^\ast(s)\tilde{V}_Z\tilde{U}(s):))\right]\tilde{P}_0(0)\right\Vert\nonumber\\
&\le& \Vert\tilde{V}_Z\Vert\left[G_1(n)+|Z|G_2(n)\right].
\end{eqnarray}
Since one has ${\rm supp}\;\Pi_{\mathcal{Z}_n}(\mathcal{R}_s(:\tilde{U}^\ast(s)\tilde{V}_Z\tilde{U}(s):))\subset \mathcal{Z}_n$ 
by definition, the following equality is valid: 
$$ 
\left\Vert\Pi_{\mathcal{Z}_n}(\mathcal{R}_s(:\tilde{U}^\ast(s)\tilde{V}_Z\tilde{U}(s):))\tilde{P}_0(0)\right\Vert=
\left\Vert\Pi_{\mathcal{Z}_n}(\mathcal{R}_s(:\tilde{U}^\ast(s)\tilde{V}_Z\tilde{U}(s):))\tilde{P}_{0,\mathcal{Z}_n}\right\Vert,
$$
where the local operator $\tilde{P}_{0,\mathcal{Z}_n}$ with the support $\mathcal{Z}_n\times\{c,d\}$ is defined by 
\begin{equation}
\tilde{P}_{0,\mathcal{Z}_n}:=\prod_{x\in\mathcal{Z}_n\times\{c,d\}}\xi_x\xi_x^\dagger.
\end{equation}
Therefore, we obtain 
\begin{equation}
\label{PiZnRsUVUP0bound}
\left\Vert\Pi_{\mathcal{Z}_n}(\mathcal{R}_s(:\tilde{U}^\ast(s)\tilde{V}_Z\tilde{U}(s):))\tilde{P}_{0,\mathcal{Z}_n}\right\Vert
\le \Vert\tilde{V}_Z\Vert\left[G_1(n)+|Z|G_2(n)\right].
\end{equation}

Similarly to Appendix~\ref{Sec:UVZUDelta}, we introduce 
\begin{equation}
\label{DeltatildeURs0V}
\Delta_{\tilde{U},\mathcal{R}_s}^0(\tilde{V}_Z,s):=\Pi_Z(\mathcal{R}_s(:\tilde{U}^\ast(s)\tilde{V}_Z\tilde{U}(s):))
\end{equation}
and 
\begin{equation}
\label{DeltatildeURsnV}
\Delta_{\tilde{U},\mathcal{R}_s}^n(\tilde{V}_Z,s):=\Pi_{\mathcal{Z}_n}(\mathcal{R}_s(:\tilde{U}^\ast(s)\tilde{V}_Z\tilde{U}(s):))
-\Pi_{\mathcal{Z}_{n-1}}(\mathcal{R}_s(:\tilde{U}^\ast(s)\tilde{V}_Z\tilde{U}(s):))
\end{equation}
for $n\ge 1$. Then, 
\begin{equation}
\mathcal{R}_s(:\tilde{U}^\ast(s)\tilde{V}_Z\tilde{U}(s):)=\sum_{n=0}^\infty \Delta_{\tilde{U},\mathcal{R}_s}^n(\tilde{V}_Z,s).
\end{equation}
However, this decomposition may ruin the above good property (\ref{RsUVZUP00}). Namely, it may occur that  
$$
\Delta_{\tilde{U},\mathcal{R}_s}^n(\tilde{V}_Z,s)\tilde{P}_0(0)\ne 0
$$
for some $n$. 

In order to cure this defect, we further introduce a decomposition of the quantities as follows.   
Note that 
\begin{eqnarray}
\label{DeltaURsnDecomp}
\Delta_{\tilde{U},\mathcal{R}_s}^n(\tilde{V}_Z,s)&=&[(1-\tilde{P}_{0,\mathcal{Z}_n})+\tilde{P}_{0,\mathcal{Z}_n}]
\Delta_{\tilde{U},\mathcal{R}_s}^n(\tilde{V}_Z,s)[(1-\tilde{P}_{0,\mathcal{Z}_n})+\tilde{P}_{0,\mathcal{Z}_n}]\nonumber\\
&=&(1-\tilde{P}_{0,\mathcal{Z}_n})\Delta_{\tilde{U},\mathcal{R}_s}^n(\tilde{V}_Z,s)(1-\tilde{P}_{0,\mathcal{Z}_n})
+(1-\tilde{P}_{0,\mathcal{Z}_n})\Delta_{\tilde{U},\mathcal{R}_s}^n(\tilde{V}_Z,s)\tilde{P}_{0,\mathcal{Z}_n}\nonumber\\
& &+\tilde{P}_{0,\mathcal{Z}_n}\Delta_{\tilde{U},\mathcal{R}_s}^n(\tilde{V}_Z,s)(1-\tilde{P}_{0,\mathcal{Z}_n})
+\tilde{P}_{0,\mathcal{Z}_n}\Delta_{\tilde{U},\mathcal{R}_s}^n(\tilde{V}_Z,s)\tilde{P}_{0,\mathcal{Z}_n}.
\end{eqnarray}
The second term in the right-hand side is written 
\begin{eqnarray*}
(1-\tilde{P}_{0,\mathcal{Z}_n})\Delta_{\tilde{U},\mathcal{R}_s}^n(\tilde{V}_Z,s)\tilde{P}_{0,\mathcal{Z}_n}
&=&(1-\tilde{P}_{0,\mathcal{Z}_n})\Pi_{\mathcal{Z}_n}(\mathcal{R}_s(:\tilde{U}^\ast(s)\tilde{V}_Z\tilde{U}(s):))
\tilde{P}_{0,\mathcal{Z}_n}\\
& &-(1-\tilde{P}_{0,\mathcal{Z}_n})\Pi_{\mathcal{Z}_{n-1}}(\mathcal{R}_s(:\tilde{U}^\ast(s)\tilde{V}_Z\tilde{U}(s):))
\tilde{P}_{0,\mathcal{Z}_n}
\end{eqnarray*}
for $n\ge 1$. Let us consider 
\begin{eqnarray*}
& &(1-\tilde{P}_{0,\mathcal{Z}_{n-1}})\Pi_{\mathcal{Z}_{n-1}}(\mathcal{R}_s(:\tilde{U}^\ast(s)\tilde{V}_Z\tilde{U}(s):))
\tilde{P}_{0,\mathcal{Z}_{n-1}}\\
& &-(1-\tilde{P}_{0,\mathcal{Z}_n})\Pi_{\mathcal{Z}_{n-1}}(\mathcal{R}_s(:\tilde{U}^\ast(s)\tilde{V}_Z\tilde{U}(s):))
\tilde{P}_{0,\mathcal{Z}_n}.
\end{eqnarray*}
The first term is that in 
$(1-\tilde{P}_{0,\mathcal{Z}_{n-1}})\Delta_{\tilde{U},\mathcal{R}_s}^{n-1}(\tilde{V}_Z,s)\tilde{P}_{0,\mathcal{Z}_{n-1}}$. 
Note that 
$$
\tilde{P}_{0,\mathcal{Z}_n}=\tilde{P}_{0,\mathcal{Z}_{n-1}}\tilde{P}_{0,\mathcal{Z}_n\backslash \mathcal{Z}_{n-1}}
$$
and 
$$
(1-\tilde{P}_{0,\mathcal{Z}_n})\tilde{P}_{0,\mathcal{Z}_n\backslash \mathcal{Z}_{n-1}}=
(1-\tilde{P}_{0,\mathcal{Z}_{n-1}})\tilde{P}_{0,\mathcal{Z}_n\backslash \mathcal{Z}_{n-1}}.
$$ 
By using these relations, one has 
\begin{eqnarray*}
& &(1-\tilde{P}_{0,\mathcal{Z}_{n-1}})\Pi_{\mathcal{Z}_{n-1}}(\mathcal{R}_s(:\tilde{U}^\ast(s)\tilde{V}_Z\tilde{U}(s):))
\tilde{P}_{0,\mathcal{Z}_{n-1}}\\
& &-(1-\tilde{P}_{0,\mathcal{Z}_n})\Pi_{\mathcal{Z}_{n-1}}(\mathcal{R}_s(:\tilde{U}^\ast(s)\tilde{V}_Z\tilde{U}(s):))
\tilde{P}_{0,\mathcal{Z}_n}\\
&=&(1-\tilde{P}_{0,\mathcal{Z}_{n-1}})\Pi_{\mathcal{Z}_{n-1}}(\mathcal{R}_s(:\tilde{U}^\ast(s)\tilde{V}_Z\tilde{U}(s):))
\tilde{P}_{0,\mathcal{Z}_{n-1}}\\
& &-(1-\tilde{P}_{0,\mathcal{Z}_{n-1}})\Pi_{\mathcal{Z}_{n-1}}(\mathcal{R}_s(:\tilde{U}^\ast(s)\tilde{V}_Z\tilde{U}(s):))
\tilde{P}_{0,\mathcal{Z}_{n-1}}\tilde{P}_{0,\mathcal{Z}_n\backslash \mathcal{Z}_{n-1}}\\
&=&(1-\tilde{P}_{0,\mathcal{Z}_{n-1}})\Pi_{\mathcal{Z}_{n-1}}(\mathcal{R}_s(:\tilde{U}^\ast(s)\tilde{V}_Z\tilde{U}(s):))
\tilde{P}_{0,\mathcal{Z}_{n-1}}(1-\tilde{P}_{0,\mathcal{Z}_n\backslash \mathcal{Z}_{n-1}}). 
\end{eqnarray*}
Similarly, for the fourth term in the right-hand side in the second equality of (\ref{DeltaURsnDecomp}), we have 
\begin{eqnarray*} 
& &\tilde{P}_{0,\mathcal{Z}_{n-1}}
\Pi_{\mathcal{Z}_{n-1}}(\mathcal{R}_s(:\tilde{U}^\ast(s)\tilde{V}_Z\tilde{U}(s):))
\tilde{P}_{0,\mathcal{Z}_{n-1}}
-\tilde{P}_{0,\mathcal{Z}_n}
\Pi_{\mathcal{Z}_{n-1}}(\mathcal{R}_s(:\tilde{U}^\ast(s)\tilde{V}_Z\tilde{U}(s):))
\tilde{P}_{0,\mathcal{Z}_n}\\
&=&\tilde{P}_{0,\mathcal{Z}_{n-1}}
\Pi_{\mathcal{Z}_{n-1}}(\mathcal{R}_s(:\tilde{U}^\ast(s)\tilde{V}_Z\tilde{U}(s):))
\tilde{P}_{0,\mathcal{Z}_{n-1}}\\
& &-\tilde{P}_{0,\mathcal{Z}_n\backslash \mathcal{Z}_{n-1}}\tilde{P}_{0,\mathcal{Z}_{n-1}}
\Pi_{\mathcal{Z}_{n-1}}(\mathcal{R}_s(:\tilde{U}^\ast(s)\tilde{V}_Z\tilde{U}(s):))
\tilde{P}_{0,\mathcal{Z}_{n-1}}\tilde{P}_{0,\mathcal{Z}_n\backslash \mathcal{Z}_{n-1}}\\
&=&(1-\tilde{P}_{0,\mathcal{Z}_n\backslash \mathcal{Z}_{n-1}})\tilde{P}_{0,\mathcal{Z}_{n-1}}
\Pi_{\mathcal{Z}_{n-1}}(\mathcal{R}_s(:\tilde{U}^\ast(s)\tilde{V}_Z\tilde{U}(s):))
\tilde{P}_{0,\mathcal{Z}_{n-1}}(1-\tilde{P}_{0,\mathcal{Z}_n\backslash \mathcal{Z}_{n-1}}). 
\end{eqnarray*}
{From} these observations, we have 
\begin{eqnarray*}
& &\mathcal{R}_s(:\tilde{U}^\ast(s)\tilde{V}_Z\tilde{U}(s):)\\
&=&\sum_{n=0}^\infty \Delta_{\tilde{U},\mathcal{R}_s}^n(\tilde{V}_Z,s)\\
&=&\sum_{n=0}^\infty(1-\tilde{P}_{0,\mathcal{Z}_n})\Delta_{\tilde{U},\mathcal{R}_s}^n(\tilde{V}_Z,s)(1-\tilde{P}_{0,\mathcal{Z}_n})\\
&+&\sum_{n=1}^\infty(1-\tilde{P}_{0,\mathcal{Z}_{n-1}})\Pi_{\mathcal{Z}_{n-1}}(\mathcal{R}_s(:\tilde{U}^\ast(s)\tilde{V}_Z\tilde{U}(s):))
\tilde{P}_{0,\mathcal{Z}_{n-1}}(1-\tilde{P}_{0,\mathcal{Z}_n\backslash \mathcal{Z}_{n-1}})\\
&+&\sum_{n=1}^\infty(1-\tilde{P}_{0,\mathcal{Z}_n\backslash \mathcal{Z}_{n-1}})\tilde{P}_{0,\mathcal{Z}_{n-1}}
\Pi_{\mathcal{Z}_{n-1}}(\mathcal{R}_s(:\tilde{U}^\ast(s)\tilde{V}_Z\tilde{U}(s):))(1-\tilde{P}_{0,\mathcal{Z}_{n-1}})\\
&+&\sum_{n=1}^\infty(1-\tilde{P}_{0,\mathcal{Z}_n\backslash \mathcal{Z}_{n-1}})\tilde{P}_{0,\mathcal{Z}_{n-1}}
\Pi_{\mathcal{Z}_{n-1}}(\mathcal{R}_s(:\tilde{U}^\ast(s)\tilde{V}_Z\tilde{U}(s):))
\tilde{P}_{0,\mathcal{Z}_{n-1}}(1-\tilde{P}_{0,\mathcal{Z}_n\backslash \mathcal{Z}_{n-1}}). 
\end{eqnarray*}
Consequently, we obtain the expression, 
\begin{equation}
\tilde{\mathcal{W}}_Z^{(1)}(s)=s\mathcal{R}_s(:\tilde{U}^\ast(s)\tilde{V}_Z\tilde{U}(s):)
=\sum_{n=0}^\infty \tilde{\mathcal{W}}_{Z,n}^{(1)}(s),
\end{equation}
with 
$$
\tilde{\mathcal{W}}_{Z,0}^{(1)}(s):=
s(1-\tilde{P}_{0,\mathcal{Z}})\Delta_{\tilde{U},\mathcal{R}_s}^0(\tilde{V}_Z,s)(1-\tilde{P}_{0,\mathcal{Z}})
$$
and 
\begin{eqnarray*}
\tilde{\mathcal{W}}_{Z,n}^{(1)}(s)&:=&s\left[
(1-\tilde{P}_{0,\mathcal{Z}_n})\Delta_{\tilde{U},\mathcal{R}_s}^n(\tilde{V}_Z,s)(1-\tilde{P}_{0,\mathcal{Z}_n})\right.\\
&+&(1-\tilde{P}_{0,\mathcal{Z}_{n-1}})\Pi_{\mathcal{Z}_{n-1}}(\mathcal{R}_s(:\tilde{U}^\ast(s)\tilde{V}_Z\tilde{U}(s):))
\tilde{P}_{0,\mathcal{Z}_{n-1}}(1-\tilde{P}_{0,\mathcal{Z}_n\backslash \mathcal{Z}_{n-1}})\\
&+&(1-\tilde{P}_{0,\mathcal{Z}_n\backslash \mathcal{Z}_{n-1}})\tilde{P}_{0,\mathcal{Z}_{n-1}}
\Pi_{\mathcal{Z}_{n-1}}(\mathcal{R}_s(:\tilde{U}^\ast(s)\tilde{V}_Z\tilde{U}(s):))(1-\tilde{P}_{0,\mathcal{Z}_{n-1}})\\
&+&\left.(1-\tilde{P}_{0,\mathcal{Z}_n\backslash \mathcal{Z}_{n-1}})\tilde{P}_{0,\mathcal{Z}_{n-1}}
\Pi_{\mathcal{Z}_{n-1}}(\mathcal{R}_s(:\tilde{U}^\ast(s)\tilde{V}_Z\tilde{U}(s):))
\tilde{P}_{0,\mathcal{Z}_{n-1}}(1-\tilde{P}_{0,\mathcal{Z}_n\backslash \mathcal{Z}_{n-1}})\right]
\end{eqnarray*}
for $n\ge 1$. Clearly, all the terms satisfy 
\begin{equation}
\tilde{\mathcal{W}}_{Z,n}^{(1)}(s)\tilde{P}_0(0)=0,
\end{equation}
and the bound, 
\begin{equation}
\label{tildecalWZn1bound}
\Vert\tilde{\mathcal{W}}_{Z,n}^{(1)}(s)\Vert \le 5|s|\cdot\Vert \tilde{V}_Z\Vert\cdot[G_1(n)+|Z|G_2(n)], 
\end{equation}
{from} (\ref{RUVZPidifbound}), (\ref{PiZnRsUVUP0bound}), (\ref{DeltatildeURs0V}) and (\ref{DeltatildeURsnV}). 

In order to prove (\ref{tildecalWZnsboundedness}), it is sufficient to prove 
\begin{equation}
\label{tildecalWZn1boundness}
\sum_Z \sum_{n=0}^\infty {\rm Ind}[x\in\mathcal{Z}_n]|\mathcal{Z}_n|\Vert\tilde{\mathcal{W}}_{Z,n}^{(1)}(s)\Vert<\infty. 
\end{equation}
For simplicity, we assume 
$$
\Vert\tilde{\mathcal{W}}_{Z,n}^{(1)}(s)\Vert \le|Z|\Vert \tilde{V}_Z\Vert G_{\tilde{\mathcal{W}}^{(1)}}(n)
$$
with a sub-exponentially decaying function $G_{\tilde{\mathcal{W}}^{(1)}}(n)$, although this is weaker bound than the above bound 
(\ref{tildecalWZn1bound}). Note that 
$$
|\mathcal{Z}_n|\le \kappa_d|Z_n|n^d\le (\kappa_d)^2|Z|n^{2d}. 
$$ 
{From} these, it is enough to prove 
\begin{equation}
\label{tildecalWZn1boundness2}
\sum_Z \sum_{n=0}^\infty {\rm Ind}[x\in\mathcal{Z}_n]|Z|^2n^{2d}\Vert \tilde{V}_Z\Vert G_{\tilde{\mathcal{W}}^{(1)}}(n)<\infty. 
\end{equation}

When the summand in the left-hand side of (\ref{tildecalWZn1boundness2}) is non-vanishing, 
the site $x$ must be contained in $\mathcal{Z}_n$. Therefore, there exists $z\in Z$ such that 
${\rm dist}(x,z)={\rm dist}(x,Z)\le 2n$. By using this, we have 
\begin{eqnarray}
& &\sum_Z \sum_{n=0}^\infty {\rm Ind}[x\in\mathcal{Z}_n]|Z|^2n^{2d}\Vert \tilde{V}_Z\Vert G_{\tilde{\mathcal{W}}^{(1)}}(n)\nonumber\\
&\le& \sum_z \sum_{Z\ni z}\sum_n {\rm Ind}[{\rm dist}(x,z)\le 2n]|Z|^2n^{2d}\Vert \tilde{V}_Z\Vert G_{\tilde{\mathcal{W}}^{(1)}}(n)\nonumber\\
&\le&\tilde{C}_{h,V}\sum_z \sum_n {\rm Ind}[{\rm dist}(x,z)\le 2n]n^{2d}\ G_{\tilde{\mathcal{W}}^{(1)}}(n)\nonumber\\
&\le& 2^{d}\kappa_d\tilde{C}_{h,V}\sum_n n^{3d}G_{\tilde{\mathcal{W}}^{(1)}}(n)<\infty, 
\end{eqnarray}
where we have used (\ref{boundedtildeVX2}).

\Section{Local approximation of $[\tilde{\mathcal{H}}_{0,Z},\tilde{D}(r)]$}

In this appendix, we consider local approximation of $[\tilde{\mathcal{H}}_{0,Z},\tilde{D}(r)]$ as a preparation for 
treating the potential $\tilde{\mathcal{W}}_Z^{(2)}(s)$ of (\ref{tildecalWZ2s}). 

For this purpose, we first recall the expression of $\tilde{D}(r)$ in \cite{BMNS}, 
\begin{equation}
\label{tildeDr}
\tilde{D}(r)=\sum_Y \sum_{m=0}^\infty \Delta_{\tilde{D}}^m(\tilde{V}_Y,r),
\end{equation}
where 
$$
\Delta_{\tilde{D}}^0(\tilde{V}_Y,s):=\int_{-\infty}^{+\infty}dt\; W_\gamma(t)\cdot\Pi_Y(\tau_t^{\tilde{\mathcal{H}}_s}(\tilde{V}_Y))
$$
and 
$$
\Delta_{\tilde{D}}^m(\tilde{V}_Y,s):=\int_{-\infty}^{+\infty}dt\; W_\gamma(t)\cdot\Pi_{Y_n}(\tau_t^{\tilde{\mathcal{H}}_s}(\tilde{V}_Y))
-\int_{-\infty}^{+\infty}dt\; W_\gamma(t)\cdot\Pi_{Y_{n-1}}(\tau_t^{\tilde{\mathcal{H}}_s}(\tilde{V}_Y))
$$
for $m\ge 1$; the weight function $W_\gamma(t)$ is given by the equation (2.12) in \cite{BMNS}, and 
$\tau_t^{\tilde{\mathcal{H}}_s}(\cdots)$ is the time evolution 
by the Hamiltonian $\tilde{\mathcal{H}}_s$ of (\ref{tildecalHs}), i.e., it is given by 
$$
\tau_t^{\tilde{\mathcal{H}}_s}(\cdots):=\exp[it\tilde{\mathcal{H}}_s](\cdots)\exp[-it\tilde{\mathcal{H}}_s].
$$
The Hamiltonian $\tilde{\mathcal{H}}_s=\tilde{\mathcal{H}}_0+s\tilde{V}$ satisfies 
the exponential decay bounds, (\ref{tildecalH0Zexpdecay}) and (\ref{boundedtildeVX}).  
Therefore, by Lemma~4.7 in \cite{BMNS}, the following bound holds: 
\begin{equation}
\label{DeltaDmbound}
\Vert\Delta_{\tilde{D}}^m(\tilde{V}_Y,r)\Vert\le C_{\tilde{D}} |Y|\Vert \tilde{V}_Y\Vert G_{\tilde{D}}(m)
\end{equation}
with the positive constant $C_{\tilde{D}}$ and a sub-exponentially decaying function $G_{\tilde{D}}(m)$. 
By using the expression of $\tilde{D}(r)$, one has 
\begin{eqnarray*}
[\tilde{\mathcal{H}}_{0,Z},\tilde{D}(r)]
&=&\sum_Y \sum_{m=0}^\infty \;[\tilde{\mathcal{H}}_{0,Z},\Delta_{\tilde{D}}^m(\tilde{V}_Y,r)]\\
&=&\sum_Y \sum_{m=0}^\infty \;{\rm Ind}[Z\cap Y_m\ne \emptyset]\cdot
[\tilde{\mathcal{H}}_{0,Z},\Delta_{\tilde{D}}^m(\tilde{V}_Y,r)].
\end{eqnarray*}
We write 
\begin{equation}
\label{commutildecalH0ZtildeDZn}
[\tilde{\mathcal{H}}_{0,Z},\tilde{D}(r)]_{Z_n}:=\sum_Y \sum_{m=0}^\infty \;{\rm Ind}[Z\cap Y_m\ne \emptyset]\cdot
{\rm Ind}[Y_m\subset Z_n]\cdot[\tilde{\mathcal{H}}_{0,Z},\Delta_{\tilde{D}}^m(\tilde{V}_Y,r)].
\end{equation}
When the summand in the right-hand side is non-vanishing, $Z\cap Y_m\ne \emptyset$. Therefore, there exists $z\in Z\cap Y_m$ 
and $y\in Y$ such that ${\rm dist}(z,y)={\rm dist}(z,Y)\le m$. 
By using this and the above inequality (\ref{DeltaDmbound}), the norm of 
$[\tilde{\mathcal{H}}_{0,Z},\tilde{D}(r)]_{Z_n}$ is estimated as follows: 
\begin{eqnarray}
\label{commutildecalH0ZtildeDZnfinite}
\Vert[\tilde{\mathcal{H}}_{0,Z},\tilde{D}(r)]_{Z_n}\Vert&\le&\sum_{z\in Z}\sum_y
\sum_{Y\ni y} \sum_{m=0}^\infty \;
{\rm Ind}[{\rm dist}(z,y)\le m]\cdot\Vert[\tilde{\mathcal{H}}_{0,Z},\Delta_{\tilde{D}}^m(\tilde{V}_Y,r)]\Vert\nonumber\\
&\le& 2C_{\tilde{D}}\Vert\tilde{\mathcal{H}}_{0,Z}\Vert 
\sum_{z\in Z}\sum_y \sum_{Y\ni y} \sum_{m=0}^\infty \;
{\rm Ind}[{\rm dist}(z,y)\le m]|Y|\Vert \tilde{V}_Y\Vert G_{\tilde{D}}(m)\nonumber\\
&\le&2C_{\tilde{D}}\tilde{C}_{h,V}\Vert\tilde{\mathcal{H}}_{0,Z}\Vert\sum_{z\in Z}\sum_y \sum_{m=0}^\infty \;{\rm Ind}[{\rm dist}(z,y)\le m] G_{\tilde{D}}(m)\nonumber\\
&\le&2C_{\tilde{D}}\tilde{C}_{h,V}\kappa_d\Vert\tilde{\mathcal{H}}_{0,Z}\Vert\sum_{z\in Z}\sum_{m=0}^\infty \;m^d G_{\tilde{D}}(m)\nonumber\\
&\le&4C_{\tilde{D}}\tilde{C}_{h,V}\kappa_d\Vert\tilde{\mathcal{H}}_{0,Z}\Vert\sum_{m=0}^\infty \;m^d G_{\tilde{D}}(m)
\le C_{[\tilde{\mathcal{H}}_{0,Z},\tilde{D}]}\Vert\tilde{\mathcal{H}}_{0,Z}\Vert
\end{eqnarray}
with the positive constant $C_{[\tilde{\mathcal{H}}_{0,Z},\tilde{D}]}$, 
where we have used (\ref{boundedtildeVX2}). 

{From} the definition of (\ref{commutildecalH0ZtildeDZn}), one has 
\begin{eqnarray}
\label{difcommuH0ZtildeDr}
& &\left\Vert[\tilde{\mathcal{H}}_{0,Z},\tilde{D}(r)]-[\tilde{\mathcal{H}}_{0,Z},\tilde{D}(r)]_{Z_n}\right\Vert\nonumber\\
&\le& \sum_Y \sum_{m=0}^\infty \;{\rm Ind}[Z\cap Y_m\ne \emptyset]\cdot
{\rm Ind}[Y_m\not\subset Z_n]\cdot\Vert[\tilde{\mathcal{H}}_{0,Z},\Delta_{\tilde{D}}^m(\tilde{V}_Y,r)]\Vert\nonumber\\
&\le& 2\Vert\tilde{\mathcal{H}}_{0,Z}\Vert\sum_Y \sum_{m=0}^\infty \;{\rm Ind}[Z\cap Y_m\ne \emptyset]\cdot
{\rm Ind}[Y_m\not\subset Z_n]\cdot\Vert \Delta_{\tilde{D}}^m(\tilde{V}_Y,r)\Vert\nonumber\\
&\le& 2C_{\tilde{D}}\Vert\tilde{\mathcal{H}}_{0,Z}\Vert \sum_Y \sum_{m=0}^\infty \;{\rm Ind}[Z\cap Y_m\ne \emptyset]\cdot
{\rm Ind}[Y_m\not\subset Z_n]\cdot|Y|\Vert \tilde{V}_Y\Vert G_{\tilde{D}}(m). 
\end{eqnarray}
Since the set $Z$ at most consists of two sites, we write $Z=\{x,y\}$. From $Z\cap Y_m\ne \emptyset$, we can assume $x\in Y_m$. 
Then, there exists $u\in Y$ such that ${\rm dist}(x,u)={\rm dist}(x,Y)\le m$. 
On the other hand, from $Y_m\not\subset Z_n$, there exists $z\in Y_m$ such that $z\not\in Z_n$. 
Therefore, ${\rm dist}(x,z)>n$. In addition, 
there exists $v\in Y$ such that ${\rm dist}(v,z)={\rm dist}(z,Y)\le m$. From these observations, we have 
$$
n<{\rm dist}(x,z)\le {\rm dist}(x,u)+{\rm dist}(u,v)+{\rm dist}(v,z)\le 2m +{\rm dist}(u,v).
$$ 
The quantity in the right-hand side of (\ref{difcommuH0ZtildeDr}) can be estimated as  
\begin{eqnarray}
\label{difcommuH0ZtildeDr2}
& &\sum_Y \sum_{m=0}^\infty \;{\rm Ind}[Z\cap Y_m\ne \emptyset]\cdot
{\rm Ind}[Y_m\not\subset Z_n]\cdot|Y|\Vert \tilde{V}_Y\Vert G_{\tilde{D}}(m)\nonumber\\
&\le &\sum_{u,v,z}\sum_{Y\ni u,v} \sum_{m=0}^\infty \;{\rm Ind}[2m +{\rm dist}(u,v)>n,{\rm dist}(x,u)\le m, {\rm dist}(z,v)\le m]
|Y|\Vert \tilde{V}_Y\Vert G_{\tilde{D}}(m)\nonumber\\
&\le& \kappa_d\tilde{C}_{h,V} \sum_{u,v}\sum_{m=0}^\infty\; m^d\cdot {\rm Ind}[2m +{\rm dist}(u,v)>n,{\rm dist}(x,u)\le m]\cdot
e^{-\tilde{m}_{h,V}{\rm dist}(u,v)} G_{\tilde{D}}(m),\nonumber\\
\end{eqnarray}
where we have used the inequality (\ref{boundedtildeVX}). From Lemmas~2.6 and 4.7 in \cite{BMNS}, one has 
$$
G_{\tilde{D}}(m)\le\mathcal{C}_5 u_{\tilde{\mu}}(m/r_1)\le \mathcal{C}_5u_{\tilde{\mu}/2}(m/r_1)u_{\tilde{\mu}/2}(m/r_1) 
$$
with the positive constants, $\tilde{\mu}\in (0,2/7)$, $\mathcal{C}_5$ and $r_1$. In addition, the following bound holds 
under the condition $2m +{\rm dist}(u,v)>n$:
$$
e^{-\tilde{m}_{h,V}{\rm dist}(u,v)/2}u_{\tilde{\mu}/2}(m/r_1)
\le \mathcal{C}_6u_{\tilde{\mu}/2}(n/2r_1)
$$
with a positive constant $\mathcal{C}_6$. 
Using this, the right-hand side of (\ref{difcommuH0ZtildeDr2}) is estimated as 
\begin{eqnarray}
\label{difcommuH0ZtildeDr3}
& &\sum_{u,v}\sum_{m=0}^\infty\; m^d\cdot {\rm Ind}[2m +{\rm dist}(u,v)>n,{\rm dist}(x,u)\le m]\cdot
e^{-\tilde{m}_{h,V}{\rm dist}(u,v)} G_{\tilde{D}}(m)\nonumber\\
&\le& \mathcal{C}_5 \sum_{u,v}\sum_{m=0}^\infty\; m^d\cdot u_{\tilde{\mu}/2}(m/r_1) 
e^{-\tilde{m}_{h,V}{\rm dist}(u,v)/2}\cdot{\rm Ind}[{\rm dist}(x,u)\le m]\nonumber\\
& &\times {\rm Ind}[2m +{\rm dist}(u,v)>n]\cdot e^{-\tilde{m}_{h,V}{\rm dist}(u,v)/2}u_{\tilde{\mu}/2}(m/r_1)\nonumber\\
&\le&\mathcal{C}_5 \mathcal{C}_6 \sum_{u,v}\sum_{m=0}^\infty\; m^d\cdot u_{\tilde{\mu}/2}(m/r_1) 
e^{-\tilde{m}_{h,V}{\rm dist}(u,v)/2}\cdot{\rm Ind}[{\rm dist}(x,u)\le m]\cdot u_{\tilde{\mu}/2}(n/2r_1)\nonumber\\
&\le&\mathcal{C}_5\mathcal{C}_6\hat{K}_{h,V}  \sum_{u}\sum_{m=0}^\infty\; m^d\cdot u_{\tilde{\mu}/2}(m/r_1) 
\cdot{\rm Ind}[{\rm dist}(x,u)\le m]\cdot u_{\tilde{\mu}/2}(n/2r_1)\nonumber\\
&\le&\mathcal{C}_5\mathcal{C}_6 \hat{K}_{h,V}\kappa_d \left[\sum_{m=0}^\infty\; m^{2d}u_{\tilde{\mu}/2}(m/r_1)\right] 
u_{\tilde{\mu}/2}(n/2r_1), 
\end{eqnarray}
where the positive constant $\hat{K}_{h,V}$ is given by   
\begin{equation}
\label{hatKhV}
\hat{K}_{h,V}:=\sum_{v\in\ze^d}e^{-\tilde{m}_{h,V}{\rm dist}(0,v)/2}.
\end{equation}
Substituting (\ref{difcommuH0ZtildeDr2}) and (\ref{difcommuH0ZtildeDr3}) into (\ref{difcommuH0ZtildeDr}), we obtain 
\begin{equation}
\label{diftildecalH0ZtildeDZn}
\left\Vert[\tilde{\mathcal{H}}_{0,Z},\tilde{D}(r)]-[\tilde{\mathcal{H}}_{0,Z},\tilde{D}(r)]_{Z_n}\right\Vert
\le C_{[\tilde{\mathcal{H}}_{0,Z},\tilde{D}]}'\Vert\tilde{\mathcal{H}}_{0,Z}\Vert u_{\tilde{\mu}/2}(n/2r_1)
\end{equation}
with the positive constant $C_{[\tilde{\mathcal{H}}_{0,Z},\tilde{D}]}'$.

\Section{Local boundedness of the interaction $\tilde{\mathcal{W}}^{(2)}(s)$}

We recall the expression, 
$$
\tilde{\mathcal{W}}_Z^{(2)}(s)=\int_0^s ds'\mathcal{R}_s(:\tilde{U}^\ast(s')i[\tilde{\mathcal{H}}_{0,Z},\tilde{D}(s')]
\tilde{U}(s'):).
$$
In order to prove the inequality (\ref{tildecalWZnsboundedness}) for $\tilde{\mathcal{W}}_Z^{(2)}(s)$, 
it is sufficient to prove the inequality (\ref{difcalRstildeUtildecalH0ZtildeDtildeU}) below which is an analogue 
of (\ref{RUVZPidifbound}) in the case of $\tilde{\mathcal{W}}_{Z,n}^{(1)}(s)$. 
Actually, we can construct $\tilde{\mathcal{W}}_{Z,n}^{(2)}(s)$ which satisfies the desired properties 
in the same way as in the case of $\tilde{\mathcal{W}}_{Z,n}^{(1)}(s)$. 

For $Z\subset{\Lambda}$, we define 
$$
\mathcal{Z}_n:=\{z\;|\; {\rm dist}(z,Z_n)\le n\}
$$
and 
$$
\mathcal{Z}_n':=\{z\;|\; {\rm dist}(z,\mathcal{Z}_n)\le n\}.
$$
Note that 
\begin{eqnarray}
& &\left\Vert \mathcal{R}_s(\tilde{U}^\ast(r)[\tilde{\mathcal{H}}_{0,Z},\tilde{D}(r)]\tilde{U}(r))
-\Pi_{\mathcal{Z}_n'}(\mathcal{R}_s(\tilde{U}^\ast(r)[\tilde{\mathcal{H}}_{0,Z},\tilde{D}(r)]\tilde{U}(r)))\right\Vert
\nonumber\\
&\le&\left\Vert \mathcal{R}_s(\tilde{U}^\ast(r)[\tilde{\mathcal{H}}_{0,Z},\tilde{D}(r)]\tilde{U}(r))
-\mathcal{R}_s(\tilde{U}^\ast(r)[\tilde{\mathcal{H}}_{0,Z},\tilde{D}(r)]_{Z_n}\tilde{U}(r))\right\Vert\nonumber\\
&+&\left\Vert \mathcal{R}_s(\tilde{U}^\ast(r)[\tilde{\mathcal{H}}_{0,Z},\tilde{D}(r)]_{Z_n}\tilde{U}(r))
-\mathcal{R}_s(\Pi_{\mathcal{Z}_n}(\tilde{U}^\ast(r)[\tilde{\mathcal{H}}_{0,Z},\tilde{D}(r)]_{Z_n}\tilde{U}(r)))\right\Vert
\nonumber\\
&+&\left\Vert \mathcal{R}_s(\Pi_{\mathcal{Z}_n}(\tilde{U}^\ast(r)[\tilde{\mathcal{H}}_{0,Z},\tilde{D}(r)]_{Z_n}\tilde{U}(r)))
-\Pi_{\mathcal{Z}_n'}(\mathcal{R}_s(\Pi_{\mathcal{Z}_n}(\tilde{U}^\ast(r)[\tilde{\mathcal{H}}_{0,Z},
\tilde{D}(r)]_{Z_n}\tilde{U}(r))))\right\Vert\nonumber\\
&+&\left\Vert \Pi_{\mathcal{Z}_n'}(\mathcal{R}_s(\Pi_{\mathcal{Z}_n}(\tilde{U}^\ast(r)[\tilde{\mathcal{H}}_{0,Z},
\tilde{D}(r)]_{Z_n}\tilde{U}(r))))
-\Pi_{\mathcal{Z}_n'}(\mathcal{R}_s(\tilde{U}^\ast(r)[\tilde{\mathcal{H}}_{0,Z},
\tilde{D}(r)]_{Z_n}\tilde{U}(r)))\right\Vert\nonumber\\
&+&\left\Vert \Pi_{\mathcal{Z}_n'}(\mathcal{R}_s(\tilde{U}^\ast(r)[\tilde{\mathcal{H}}_{0,Z},
\tilde{D}(r)]_{Z_n}\tilde{U}(r)))
-\Pi_{\mathcal{Z}_n'}(\mathcal{R}_s(\tilde{U}^\ast(r)[\tilde{\mathcal{H}}_{0,Z},
\tilde{D}(r)]\tilde{U}(r)))\right\Vert.
\end{eqnarray}
By using the inequality (\ref{diftildecalH0ZtildeDZn}), the first term in the right-hand side is estimated as follows: 
\begin{eqnarray}
& &\left\Vert \mathcal{R}_s(\tilde{U}^\ast(r)[\tilde{\mathcal{H}}_{0,Z},\tilde{D}(r)]\tilde{U}(r))
-\mathcal{R}_s(\tilde{U}^\ast(r)[\tilde{\mathcal{H}}_{0,Z},\tilde{D}(r)]_{Z_n}\tilde{U}(r))\right\Vert\nonumber\\
&\le& \left\Vert \tilde{U}^\ast(r)[\tilde{\mathcal{H}}_{0,Z},\tilde{D}(r)]\tilde{U}(r)
-\tilde{U}^\ast(r)[\tilde{\mathcal{H}}_{0,Z},\tilde{D}(r)]_{Z_n}\tilde{U}(r)\right\Vert\nonumber\\
&\le& \left\Vert [\tilde{\mathcal{H}}_{0,Z},\tilde{D}(r)]
-[\tilde{\mathcal{H}}_{0,Z},\tilde{D}(r)]_{Z_n}\right\Vert\le C_{[\tilde{\mathcal{H}}_{0,Z},\tilde{D}]}'
\Vert \tilde{\mathcal{H}}_{0,Z}\Vert u_{\tilde{\mu}/2}(n/2r_1).
\end{eqnarray}
Similarly, the fifth term satisfies 
\begin{eqnarray}
& &\left\Vert \Pi_{\mathcal{Z}_n'}(\mathcal{R}_s(\tilde{U}^\ast(r)[\tilde{\mathcal{H}}_{0,Z},
\tilde{D}(r)]_{Z_n}\tilde{U}(r)))
-\Pi_{\mathcal{Z}_n'}(\mathcal{R}_s(\tilde{U}^\ast(r)[\tilde{\mathcal{H}}_{0,Z},
\tilde{D}(r)]\tilde{U}(r)))\right\Vert\nonumber\\
&\le&\left\Vert [\tilde{\mathcal{H}}_{0,Z},\tilde{D}(r)]_{Z_n}-[\tilde{\mathcal{H}}_{0,Z},\tilde{D}(r)]\right\Vert
\le C_{[\tilde{\mathcal{H}}_{0,Z},\tilde{D}]}'\Vert \tilde{\mathcal{H}}_{0,Z}\Vert u_{\tilde{\mu}/2}(n/2r_1).
\end{eqnarray}
For the second and fourth terms, we have 
\begin{eqnarray}
& &\left\Vert \mathcal{R}_s(\tilde{U}^\ast(r)[\tilde{\mathcal{H}}_{0,Z},\tilde{D}(r)]_{Z_n}\tilde{U}(r))
-\mathcal{R}_s(\Pi_{\mathcal{Z}_n}(\tilde{U}^\ast(r)[\tilde{\mathcal{H}}_{0,Z},\tilde{D}(r)]_{Z_n}\tilde{U}(r)))\right\Vert\nonumber\\
&\le&\left\Vert \tilde{U}^\ast(r)[\tilde{\mathcal{H}}_{0,Z},\tilde{D}(r)]_{Z_n}\tilde{U}(r)
-\Pi_{\mathcal{Z}_n}(\tilde{U}^\ast(r)[\tilde{\mathcal{H}}_{0,Z},\tilde{D}(r)]_{Z_n}\tilde{U}(r))\right\Vert\nonumber\\
&\le& 2C_{\tilde{U}}\mathcal{C}_1\kappa_d\Vert [\tilde{\mathcal{H}}_{0,Z},\tilde{D}(r)]_{Z_n}\Vert\cdot n^d F_{\mu_1,r_0}(n)
\end{eqnarray}
and  
\begin{eqnarray}
& &\left\Vert \Pi_{\mathcal{Z}_n'}(\mathcal{R}_s(\Pi_{\mathcal{Z}_n}(\tilde{U}^\ast(r)[\tilde{\mathcal{H}}_{0,Z},
\tilde{D}(r)]_{Z_n}\tilde{U}(r))))
-\Pi_{\mathcal{Z}_n'}(\mathcal{R}_s(\tilde{U}^\ast(r)[\tilde{\mathcal{H}}_{0,Z},
\tilde{D}(r)]_{Z_n}\tilde{U}(r)))\right\Vert\nonumber\\
&\le&\left\Vert \Pi_{\mathcal{Z}_n}(\tilde{U}^\ast(r)[\tilde{\mathcal{H}}_{0,Z},\tilde{D}(r)]_{Z_n}\tilde{U}(r))
-\tilde{U}^\ast(r)[\tilde{\mathcal{H}}_{0,Z},\tilde{D}(r)]_{Z_n}\tilde{U}(r)\right\Vert\nonumber\\
&\le& 2C_{\tilde{U}}\mathcal{C}_1\kappa_d\Vert [\tilde{\mathcal{H}}_{0,Z},\tilde{D}(r)]_{Z_n}\Vert\cdot n^d F_{\mu_1,r_0}(n)
\end{eqnarray}
in the same way as that for deriving the inequality (\ref{UVUPibound}), where we have used $|Z_n|\le \kappa_d|Z|n^d=2\kappa_d n^d$.  
Further, by using the inequality (\ref{calRsPibound}) in the same way as that for deriving the inequality (\ref{calRsPiUVZUdif}), 
the remaining third term is evaluated as follows: 
\begin{eqnarray}
& &\left\Vert \mathcal{R}_s(\Pi_{\mathcal{Z}_n}(\tilde{U}^\ast(r)[\tilde{\mathcal{H}}_{0,Z},\tilde{D}(r)]_{Z_n}\tilde{U}(r)))
-\Pi_{\mathcal{Z}_n'}(\mathcal{R}_s(\Pi_{\mathcal{Z}_n}(\tilde{U}^\ast(r)[\tilde{\mathcal{H}}_{0,Z},
\tilde{D}(r)]_{Z_n}\tilde{U}(r))))\right\Vert\nonumber\\
&\le&8\Vert [\tilde{\mathcal{H}}_{0,Z},\tilde{D}(r)]_{Z_n}\Vert
\left[W_\gamma(T(n))+\mathcal{C}_4\kappa_dn^{2d}\cdot F_{\mu_3/2,2r_0}(n)\right], 
\end{eqnarray}
where we have used $|\mathcal{Z}_n|\le \kappa_d|Z_n|n^d\le (\kappa_d)^2 |Z|n^{2d}=2(\kappa_d)^2 n^{2d}$. 
Combining these bounds with (\ref{commutildecalH0ZtildeDZnfinite}), we obtain 
\begin{equation}
\label{difcalRstildeUtildecalH0ZtildeDtildeU}
\left\Vert \mathcal{R}_s(\tilde{U}^\ast(r)[\tilde{\mathcal{H}}_{0,Z},\tilde{D}(r)]\tilde{U}(r))
-\Pi_{\mathcal{Z}_n'}(\mathcal{R}_s(\tilde{U}^\ast(r)[\tilde{\mathcal{H}}_{0,Z},\tilde{D}(r)]\tilde{U}(r)))\right\Vert
\le \Vert\tilde{\mathcal{H}}_{0,Z}\Vert G_3(n)
\end{equation}
with a sub-exponentially decaying function $G_3(n)$.

\Section{Local approximation of $[\tilde{U}^\ast(r)[\tilde{\mathcal{H}}_0,\tilde{D}(r)]\tilde{U}(r),\tilde{\mathcal{H}}_{0,Z}]$}

In order to deal with $\tilde{\mathcal{W}}_Z^{(3)}(s)$, we have to consider local approximation of 
the two operators, $[\tilde{U}^\ast(s)\tilde{V}\tilde{U}(s),\tilde{\mathcal{H}}_{0,Z}]$ and 
$[\tilde{U}^\ast(r)[\tilde{\mathcal{H}}_0,\tilde{D}(r)]\tilde{U}(r),\tilde{\mathcal{H}}_{0,Z}]$ 
in the expression (\ref{tildecalWZ3s}) of $\tilde{\mathcal{W}}_Z^{(3)}(s)$ with $\tilde{\mathcal{M}}_{t,Z}(s)$ 
of (\ref{tildecalMtZs}). In this appendix, we will treat the latter operator. 

To begin with, we recall the facts \cite{BMNS} that for any operator $\mathcal{A}_Z$ with a compact support $Z$, 
the following decomposition and the inequality hold: 
\begin{equation}
\label{decomptildeUscalAZtildeUs}
\tilde{U}^\ast(s)\mathcal{A}_Z\tilde{U}(s)=\sum_{n=0}^\infty \Delta_{\tilde{U}}^n(\mathcal{A}_Z,s),
\end{equation}
where the support of $\Delta_{\tilde{U}}^n(\mathcal{A}_Z,s)$ is contained in $Z_n$, and 
the operator $\Delta_{\tilde{U}}^n(\mathcal{A}_Z,s)$ satisfies 
\begin{equation}
\label{Deltanbound2}
\Vert \Delta_{\tilde{U}}^n(\mathcal{A}_Z,s)\Vert \le \tilde{C}_{\tilde{U}}|Z|\Vert\mathcal{A}_Z\Vert F_{\mu_1,r_0}(n-1). 
\end{equation}
Here, $\tilde{C}_{\tilde{U}}$ is the positive constant, and $F_{\mu_1,r_0}(n)$ is the function in (\ref{Deltanbound}); 
for simplicity, we take $F_{\mu_1,r_0}(-1)$ to be some positive constant. 
These facts were already explained in Appendix~\ref{Sec:UVZUDelta} for $\mathcal{A}_Z=\tilde{V}_Z$. 
Actually, these hold for  \cite{BMNS} any operator $\mathcal{A}_Z$. 
By using the above decomposition (\ref{decomptildeUscalAZtildeUs}) and the decomposition (\ref{tildeDr}) of $\tilde{D}(r)$, one has  
\begin{eqnarray}
[\tilde{U}^\ast(r)[\tilde{\mathcal{H}}_0,\tilde{D}(r)]\tilde{U}(r),\tilde{\mathcal{H}}_{0,Z}]
&=&\sum_{X}\sum_Y \sum_\ell [\tilde{U}^\ast(r)[\tilde{\mathcal{H}}_{0,X},\Delta_{\tilde{D}}^\ell(\tilde{V}_Y,r)]
\tilde{U}(r),\tilde{\mathcal{H}}_{0,Z}]\nonumber\\
&=&\sum_{X}\sum_Y \sum_\ell\sum_k [\Delta_{\tilde{U}}^k([\tilde{\mathcal{H}}_{0,X},\Delta_{\tilde{D}}^\ell(\tilde{V}_Y,r)],r),
\tilde{\mathcal{H}}_{0,Z}],\nonumber\\
\end{eqnarray}
where we have written 
\begin{equation}
\tilde{U}^\ast(r)[\tilde{\mathcal{H}}_{0,X},\Delta_{\tilde{D}}^\ell(\tilde{V}_Y,r)]
\tilde{U}(r)=\sum_{k=0}^\infty 
\Delta_{\tilde{U}}^k([\tilde{\mathcal{H}}_{0,X},\Delta_{\tilde{D}}^\ell(\tilde{V}_Y,r)],r).
\end{equation}
The support of the summand $\Delta_{\tilde{U}}^k([\tilde{\mathcal{H}}_{0,X},\Delta_{\tilde{D}}^\ell(\tilde{V}_Y,r)],r)$ 
in the right-hand side is contained in $(X\cup Y_\ell)_k$ if $X\cap Y_\ell\ne \emptyset$. In fact, if $X\cap Y_\ell=\emptyset$, 
then the summand is vanishing. 

We write 
$$
\Delta^{k,\ell}(X,Y,Z)=[\Delta_{\tilde{U}}^k([\tilde{\mathcal{H}}_{0,X},\Delta_{\tilde{D}}^\ell(\tilde{V}_Y,r)],r),
\tilde{\mathcal{H}}_{0,Z}]
$$
for short. Then, the above quantity can be written
\begin{equation}
[\tilde{U}^\ast(r)[\tilde{\mathcal{H}}_0,\tilde{D}]\tilde{U}(r),\tilde{\mathcal{H}}_{0,Z}]
=\sum_{X}\sum_Y \sum_\ell\sum_k\Delta^{k,\ell}(X,Y,Z).
\end{equation}
We also define 
\begin{equation}
\label{UH_0DUH0Zn}
[\tilde{U}^\ast(r)[\tilde{\mathcal{H}}_0,\tilde{D}]\tilde{U}(r),\tilde{\mathcal{H}}_{0,Z}]_{Z_n}:=
\sum_{X}\sum_Y \sum_\ell\sum_k\Delta^{k,\ell}(X,Y,Z)
{\rm Ind}[{\rm supp}\; \Delta^{k,\ell}(X,Y,Z)\subset Z_n].
\end{equation}
Then, one has 
\begin{eqnarray}
\label{UH_0DUH0dif}
& &\Vert [\tilde{U}^\ast(r)[\tilde{\mathcal{H}}_0,\tilde{D}]\tilde{U}(r),\tilde{\mathcal{H}}_{0,Z}]
-[\tilde{U}^\ast(r)[\tilde{\mathcal{H}}_0,\tilde{D}]\tilde{U}(r),\tilde{\mathcal{H}}_{0,Z}]_{Z_n}\Vert\nonumber\\
&\le&\sum_{X}\sum_Y \sum_\ell\sum_k \Vert \Delta^{k,\ell}(X,Y,Z)\Vert\cdot
{\rm Ind}[{\rm supp}\; \Delta^{k,\ell}(X,Y,Z)\not\subset {Z}_n].\nonumber\\
\end{eqnarray}
As mentioned above, 
${\rm supp}\;\Delta_{\tilde{U}}^k([\tilde{\mathcal{H}}_{0,X},\Delta_{\tilde{D}}^\ell(\tilde{V}_Y,r)])=(X\cup Y_\ell)_k$ 
{from} their definitions.  
When the summand in the right-hand of (\ref{UH_0DUH0dif}) is non-vanishing, there exist sites, $z$ and $\zeta$, such that 
$z\in Z\cap(X\cup Y_\ell)_k$ and $\zeta\in (X\cup Y_\ell)_k\backslash {Z}_n$. 
Then, one has ${\rm dist}(z,\zeta)>n$. Further, there exists $x\in X\cup Y_\ell$ 
such that ${\rm dist}(z,x)={\rm dist}(z,X\cup Y_\ell)\le k$. 
Similarly, there exists $y\in X\cup Y_\ell$ such that ${\rm dist}(\zeta,y)={\rm dist}(\zeta, X\cup Y_\ell)\le k$. 
Clearly, there are the following four cases: 
\begin{description}
\item{(i)}\quad $x\in X\quad\mbox{and} \quad y\in Y_\ell$
\item{(ii)}\quad $x\in Y_\ell\quad\mbox{and} \quad y\in X$
\item{(iii)}\quad $x\in X\quad\mbox{and} \quad y\in X$
\item{(iv)}\quad $x\in Y_\ell\quad\mbox{and} \quad y\in Y_\ell$
\end{description}

On the other hand, if $[\tilde{\mathcal{H}}_{0,X},\Delta_{\tilde{D}}^\ell(\tilde{V}_Y,r)]$ is vanishing, 
then the summand in the right-hand side of (\ref{UH_0DUH0dif}) is vanishing as well. 
Therefore, we can assume that there exists $u\in X\cap Y_\ell$  in all the cases.  
\smallskip

Let us consider first Case~(i). Since $u,y\in Y_\ell$, there exist $v,w\in Y$ 
such that ${\rm dist}(u,v)={\rm dist}(u,Y)\le \ell$ and ${\rm dist}(y,w)={\rm dist}(y,Y)\le \ell$. 
In this case, we have 
\begin{eqnarray*}
n&<&{\rm dist}(z,\zeta)\\
&\le& {\rm dist}(z,x)+{\rm dist}(x,u)+{\rm dist}(u,v)+{\rm dist}(v,w)+{\rm dist}(w,y)+{\rm dist}(y,\zeta)\nonumber\\
&\le& 2k+2\ell+{\rm dist}(x,u)+{\rm dist}(v,w).
\end{eqnarray*}
Note that 
\begin{eqnarray}
\label{normDeltakellXYZ}
\Vert \Delta^{k,\ell}(X,Y,Z)\Vert
&=&\Vert[\Delta_{\tilde{U}}^k([\tilde{\mathcal{H}}_{0,X},\Delta_{\tilde{D}}^\ell(\tilde{V}_Y,r)]),
\tilde{\mathcal{H}}_{0,Z}]\Vert\nonumber\\
&\le& 2\Vert\tilde{\mathcal{H}}_{0,Z}\Vert\cdot\Vert 
\Delta_{\tilde{U}}^k([\tilde{\mathcal{H}}_{0,X},\Delta_{\tilde{D}}^\ell(\tilde{V}_Y,r)])\Vert\nonumber\\
&\le& 2\tilde{C}_{\tilde{U}}\Vert\tilde{\mathcal{H}}_{0,Z}\Vert\cdot|X\cup Y_\ell|\cdot
\Vert[\tilde{\mathcal{H}}_{0,X},\Delta_{\tilde{D}}^\ell(\tilde{V}_Y,r)]\Vert \cdot F_{\mu_1,r_0}(k-1)\nonumber\\
&\le&8\tilde{C}_{\tilde{U}}\kappa_d\Vert\tilde{\mathcal{H}}_{0,Z}\Vert\cdot|Y|\ell^d\cdot \Vert \tilde{\mathcal{H}}_{0,X}\Vert\cdot
\Vert\Delta_{\tilde{D}}^\ell(\tilde{V}_Y,r)\Vert \cdot F_{\mu_1,r_0}(k-1)\nonumber\\
&\le&8\tilde{C}_{\tilde{U}}C_{\tilde{D}}\kappa_d\Vert\tilde{\mathcal{H}}_{0,Z}\Vert\cdot|Y|^2\ell^d\cdot 
\Vert \tilde{\mathcal{H}}_{0,X}\Vert\cdot\Vert \tilde{V}_Y\Vert \cdot F_{\mu_1,r_0}(k-1)G_{\tilde{D}}(\ell),\nonumber\\
\end{eqnarray}
where we have used (\ref{Deltanbound2}) and (\ref{DeltaDmbound}). By using this inequality, we have 
\begin{eqnarray}
\label{DeltakellXYZbound1}
& &\sum_{X}\sum_Y \sum_\ell\sum_k \Vert \Delta^{k,\ell}(X,Y,Z)\Vert\cdot
{\rm Ind}[{\rm supp}\; \Delta^{k,\ell}(X,Y,Z)\not\subset{Z}_n]\cdot{\rm Ind}[\mbox{Case~(i) occurs}]\nonumber\\
&\le& \sum_{z,x,u,v,w,y,\zeta}\sum_{X\ni x,u}\sum_{Y\ni v,w}
 \sum_\ell\sum_k \Vert \Delta^{k,\ell}(X,Y,Z)\Vert\cdot{\rm Ind}[z\in Z]\nonumber\\
&\times&{\rm Ind}[{\rm dist}(x,z)\le k]\cdot{\rm Ind}[{\rm dist}(\zeta,y)\le k]\cdot
{\rm Ind}[{\rm dist}(u,v)\le \ell]\cdot{\rm Ind}[{\rm dist}(y,w)\le \ell]\nonumber\\
&\times&{\rm Ind}[n<2k+2\ell+{\rm dist}(x,u)+{\rm dist}(v,w)]\nonumber\\
&\le&8\tilde{C}_{\tilde{U}}C_{\tilde{D}}\kappa_d\Vert\tilde{\mathcal{H}}_{0,Z}\Vert
\sum_{z,x,u,v,w,y,\zeta}\sum_{X\ni x,u}\sum_{Y\ni v,w}
 \sum_{\ell,k} \ell^d\cdot\Vert \tilde{\mathcal{H}}_{0,X}\Vert|Y|^2
\Vert \tilde{V}_Y\Vert F_{\mu_1,r_0}(k-1)G_{\tilde{D}}(\ell)\nonumber\\
&\times&{\rm Ind}[z\in Z]\cdot{\rm Ind}[{\rm dist}(x,z)\le k]\cdot{\rm Ind}[{\rm dist}(\zeta,y)\le k]\nonumber\\
&\times&{\rm Ind}[{\rm dist}(u,v)\le \ell]\cdot
{\rm Ind}[{\rm dist}(y,w)\le \ell]\cdot{\rm Ind}[n<2k+2\ell+{\rm dist}(x,u)+{\rm dist}(v,w)].\nonumber\\
\end{eqnarray}
Further, by using (\ref{tildecalH0Zexpdecay}) and (\ref{boundedtildeVX}) for the summand in the above right-hand side, we obtain 
\begin{eqnarray}
\label{DeltakellXYZbound2}
& &\sum_{x,u,v,w,y,\zeta} \sum_{X\ni x,u}\sum_{Y\ni v,w}
 \sum_{\ell,k} \ell^d\cdot\Vert \tilde{\mathcal{H}}_{0,X}\Vert\cdot|Y|^2
\Vert \tilde{V}_Y\Vert F_{\mu_1,r_0}(k-1)G_{\tilde{D}}(\ell)\nonumber\\
&\times&{\rm Ind}[{\rm dist}(x,z)\le k]\cdot{\rm Ind}[{\rm dist}(\zeta,y)\le k]\cdot{\rm Ind}[{\rm dist}(u,v)\le \ell]\cdot
{\rm Ind}[{\rm dist}(y,w)\le \ell]\nonumber\\
&\times&{\rm Ind}[n<2k+2\ell+{\rm dist}(x,u)+{\rm dist}(v,w)]\nonumber\\
&\le& C_{t,h}\tilde{C}_{h,V}\sum_{x,u,v,w,y,\zeta} 
 \sum_{\ell,k} \ell^d e^{-m_{t,h}{\rm dist}(x,u)}e^{-\tilde{m}_{h,V}{\rm dist}(v,w)}
F_{\mu_1,r_0}(k-1)G_{\tilde{D}}(\ell)\nonumber\\
&\times&{\rm Ind}[{\rm dist}(x,z)\le k]\cdot
{\rm Ind}[{\rm dist}(\zeta,y)\le k]\cdot{\rm Ind}[{\rm dist}(u,v)\le \ell]\cdot
{\rm Ind}[{\rm dist}(y,w)\le \ell]\cdot\nonumber\\
&\times&{\rm Ind}[n<2k+2\ell+{\rm dist}(x,u)+{\rm dist}(v,w)]\nonumber\\
&\le& C_{t,h}\tilde{C}_{h,V}(\kappa_d)^2\sum_{x,u,v,w}
 \sum_{\ell,k} \ell^{2d}k^d e^{-m_{t,h}{\rm dist}(x,u)}e^{-\tilde{m}_{h,V}{\rm dist}(v,w)}
F_{\mu_1,r_0}(k-1)G_{\tilde{D}}(\ell)\nonumber\\
&\times&{\rm Ind}[{\rm dist}(x,z)\le k]\cdot{\rm Ind}[{\rm dist}(u,v)\le \ell]\cdot
{\rm Ind}[n<2k+2\ell+{\rm dist}(x,u)+{\rm dist}(v,w)].
\nonumber\\
\end{eqnarray}
In order to further estimate the right-hand side, we recall the follow facts:  
For $a,b,\mu>0$, 
$$
e^{-\mu a/\log^2 a}e^{-\mu b/\log^2 b}\le \exp[-\mu(a+b)/\log^2(a+b)].
$$
By using this inequality, one can show that there exist sub-exponentially decaying functions, $\tilde{F}_1(k)$, 
$\tilde{G}_1(\ell)$ and $\tilde{G}_2(n)$, such that they satisfy   
\begin{eqnarray}
\label{additiveinequ}
& &e^{-m_{t,h}{\rm dist}(x,u)}e^{-\tilde{m}_{h,V}{\rm dist}(v,w)}\cdot
F_{\mu_1,r_0}(k-1)G_{\tilde{D}}(\ell)\nonumber\\
&\le& 
e^{-m_{t,h}{\rm dist}(x,u)/2}e^{-\tilde{m}_{h,V}{\rm dist}(v,w)/2}\cdot\tilde{F}_1(k)\tilde{G}_1(\ell)\tilde{G}_2(n)
\end{eqnarray}
under the condition, $n<2k+2\ell+{\rm dist}(x,u)+{\rm dist}(v,w)$. 
By relying on this bound, we have 
\begin{eqnarray}
\label{DeltakellXYZbound3}
& &\sum_{x,u,v,w}\sum_{\ell,k} \ell^{2d} k^d e^{-m_{t,h}{\rm dist}(x,u)}e^{-\tilde{m}_{h,V}{\rm dist}(v,w)}\cdot
F_{\mu_1,r_0}(k-1)G_{\tilde{D}}(\ell)\nonumber\\
&\times&{\rm Ind}[{\rm dist}(x,z)\le k]\cdot{\rm Ind}[{\rm dist}(u,v)\le \ell]\cdot
{\rm Ind}[n<2k+2\ell+{\rm dist}(x,u)+{\rm dist}(v,w)]\nonumber\\
&\le& \sum_{x,u,v,w} 
 \sum_{\ell,k} \ell^{2d}k^d e^{-m_{t,h}{\rm dist}(x,u)/2}e^{-\tilde{m}_{h,V}{\rm dist}(v,w)/2}\cdot
\tilde{F}_1(k)\tilde{G}_1(\ell)\cdot\tilde{G}_2(n)\nonumber\\
&\times&{\rm Ind}[{\rm dist}(x,z)\le k]\cdot{\rm Ind}[{\rm dist}(u,v)\le \ell]\nonumber\\
&\le&\hat{K}_{h,V}\sum_{x,u,v}
 \sum_{\ell,k} \ell^{2d}k^d e^{-m_{t,h}{\rm dist}(x,u)/2}\cdot\tilde{F}_1(k)\tilde{G}_1(\ell)\cdot\tilde{G}_2(n)\cdot
{\rm Ind}[{\rm dist}(x,z)\le k]\nonumber\\
&\times&{\rm Ind}[{\rm dist}(u,v)\le \ell]\nonumber\\
&\le& \kappa_d\hat{K}_{h,V}\sum_{x,u}
 \sum_{\ell,k} \ell^{3d}k^d e^{-m_{t,h}{\rm dist}(x,u)/2}\cdot\tilde{F}_1(k)\tilde{G}_1(\ell)\cdot\tilde{G}_2(n)\cdot
{\rm Ind}[{\rm dist}(x,z)\le k]\nonumber\\
&\le& \kappa_d\hat{K}_{h,V}\hat{K}_{t,h}\sum_{x}
 \sum_{\ell,k} \ell^{3d}k^d \cdot\tilde{F}_1(k)\tilde{G}_1(\ell)\cdot\tilde{G}_2(n)\cdot
{\rm Ind}[{\rm dist}(x,z)\le k]\nonumber\\
&\le& (\kappa_d)^2\hat{K}_{h,V}\hat{K}_{t,h}
\sum_{\ell,k} \ell^{3d}k^{2d} \cdot\tilde{F}_1(k)\tilde{G}_1(\ell)\cdot\tilde{G}_2(n)\le \hat{K}_{t,h,V}\tilde{G}_2(n),
\end{eqnarray}
where the positive constant $\hat{K}_{h,V}$ is given by (\ref{hatKhV}), and 
$$
\hat{K}_{t,h}:=\sum_{u\in\ze^d}e^{-m_{t,h}{\rm dist}(0,u)/2},
$$
and 
$$
\hat{K}_{t,h,V}:=(\kappa_d)^2\hat{K}_{h,V}\hat{K}_{t,h}
\sum_{\ell,k} \ell^{3d}k^{2d} \cdot\tilde{F}_1(k)\tilde{G}_1(\ell).
$$
Combining (\ref{DeltakellXYZbound1}), (\ref{DeltakellXYZbound2}) and (\ref{DeltakellXYZbound3}), we obtain 
\begin{eqnarray}
\label{DeltakellXYZboundi}
& &\sum_{X,Y}\sum_{\ell,k}\Vert \Delta^{k,\ell}(X,Y,Z)\Vert\cdot
{\rm Ind}[{\rm supp}\; \Delta^{k,\ell}(X,Y,Z)\not\subset{Z}_n]\cdot{\rm Ind}[\mbox{Case~(i) occurs}]\nonumber\\
& &\le C_{\Delta}^{\rm (i)}\Vert \tilde{\mathcal{H}}_{0,Z}\Vert \tilde{G}_2(n)
\end{eqnarray}
with the positive constant 
$C_{\Delta}^{\rm (i)}:=16\tilde{C}_{\tilde{U}}C_{\tilde{D}} C_{t,h}\tilde{C}_{h,V}(\kappa_d)^2 \hat{K}_{t,h,V}$. 
\medskip

Case~(ii) can be treated in the same way.
\medskip 

Let us consider Case~(iii). Since the set $X$ at most consists two sites, if $x\ne y$, then $x$ or $y$ is contained in $Y_\ell$ 
because of the condition $[\tilde{\mathcal{H}}_{0,X},\Delta_{\tilde{D}}^\ell(\tilde{V}_Y,r)]\ne 0$. 
This case is reduced to Case~(i) or (ii). Therefore, we assume $x=y$. For short, we call Case~(iii) with $x=y$ Case~(iii)$'$. 
For $u\in X\cap Y_\ell$, there exists $v\in Y$ such that ${\rm dist}(u,v)={\rm dist}(u,Y)\le \ell$.
Note that 
$$
n<{\rm dist}(z,\zeta)\le {\rm dist}(z,x)+{\rm dist}(x,\zeta)\le 2k.
$$
{From} these observations, in the same way as in Case~(i), we have 
\begin{eqnarray}
\label{DeltakellXYZbound4}
& &\sum_{X,Y} \sum_{\ell,k} \Vert \Delta^{k,\ell}(X,Y,Z)\Vert\cdot
{\rm Ind}[{\rm supp}\; \Delta^{k,\ell}(X,Y,Z)\not\subset{Z}_n]\cdot{\rm Ind}[\mbox{Case~(iii)$'$ occurs}]\nonumber\\
&\le&\sum_{z,x,u,v,\zeta}\sum_{X\ni x,u}\sum_{Y\ni v}
 \sum_{\ell,k} \Vert \Delta^{k,\ell}(X,Y,Z)\Vert\cdot{\rm Ind}[z\in Z]\nonumber\\
&\times&{\rm Ind}[{\rm dist}(x,z)\le k]\cdot{\rm Ind}[{\rm dist}(\zeta,x)\le k]\cdot
{\rm Ind}[{\rm dist}(u,v)\le \ell]\cdot{\rm Ind}[n<2k]\nonumber\\
&\le&8\tilde{C}_{\tilde{U}}C_{\tilde{D}}\Vert\tilde{\mathcal{H}}_{0,Z}\Vert\sum_{z,x,u,v,\zeta}\sum_{X\ni x,u}\sum_{Y\ni v}
 \sum_{\ell,k}\; \ell^d\cdot\Vert \tilde{\mathcal{H}}_{0,X}\Vert\cdot|Y|^2
\Vert \tilde{V}_Y\Vert F_{\mu_1,r_0}(k-1)G_{\tilde{D}}(\ell)\nonumber\\
&\times&{\rm Ind}[z\in Z]\cdot{\rm Ind}[{\rm dist}(x,z)\le k]\cdot{\rm Ind}[{\rm dist}(\zeta,x)\le k]\cdot
{\rm Ind}[{\rm dist}(u,v)\le \ell]\cdot{\rm Ind}[n<2k].\nonumber\\
\end{eqnarray}
Further, by using (\ref{tildecalH0Zexpdecay}) and (\ref{boundedtildeVX2}), one has 
\begin{eqnarray}
\label{DeltakellXYZbound5}
& &\sum_{x,u,v,\zeta}\sum_{X\ni x,u}\sum_{Y\ni v}
 \sum_{\ell,k} \ell^d\cdot\Vert \tilde{\mathcal{H}}_{0,X}\Vert\cdot|Y|^2
\Vert \tilde{V}_Y\Vert F_{\mu_1,r_0}(k-1)G_{\tilde{D}}(\ell)\nonumber\\
&\times&{\rm Ind}[{\rm dist}(x,z)\le k]\cdot{\rm Ind}[{\rm dist}(\zeta,x)\le k]\cdot
{\rm Ind}[{\rm dist}(u,v)\le \ell]\cdot{\rm Ind}[n<2k]\nonumber\\
&\le&C_{t,h}\tilde{C}_{h,V}\sum_{x,u,v,\zeta}\sum_{\ell,k} \ell^d \cdot e^{-m_{t,h}{\rm dist}(x,u)}
\cdot F_{\mu_1,r_0}(k-1)G_{\tilde{D}}(\ell)\nonumber\\
&\times&{\rm Ind}[{\rm dist}(x,z)\le k]\cdot{\rm Ind}[{\rm dist}(\zeta,x)\le k]\cdot
{\rm Ind}[{\rm dist}(u,v)\le \ell]\cdot{\rm Ind}[n<2k].\nonumber\\
\end{eqnarray}
The quantity in the right-hand side can be easily calculated as follows: 
\begin{eqnarray}
& &\sum_{x,u,v,\zeta}\sum_{\ell,k} \ell^d \cdot e^{-m_{t,h}{\rm dist}(x,u)}
\cdot F_{\mu_1,r_0}(k-1)G_{\tilde{D}}(\ell)\nonumber\\
&\times&{\rm Ind}[{\rm dist}(x,z)\le k]\cdot{\rm Ind}[{\rm dist}(\zeta,x)\le k]\cdot
{\rm Ind}[{\rm dist}(u,v)\le \ell]\cdot{\rm Ind}[n<2k].\nonumber\\
&\le&\kappa_d \sum_{x,u,v}\sum_{\ell,k} \ell^d k^d \cdot e^{-m_{t,h}{\rm dist}(x,u)}
\cdot F_{\mu_1,r_0}(k-1)G_{\tilde{D}}(\ell)\nonumber\\
&\times&{\rm Ind}[{\rm dist}(x,z)\le k]\cdot{\rm Ind}[{\rm dist}(u,v)\le \ell]\cdot{\rm Ind}[n<2k].\nonumber\\
&\le& (\kappa_d)^2 \sum_{x,u}\sum_{\ell,k} \ell^{2d} k^d e^{-m_{t,h}{\rm dist}(x,u)}
F_{\mu_1,r_0}(k-1)G_{\tilde{D}}(\ell)\cdot{\rm Ind}[{\rm dist}(x,z)\le k]\cdot{\rm Ind}[n<2k]\nonumber\\
&\le& (\kappa_d)^2 \hat{K}_{t,h}'\sum_{x}\sum_{\ell,k} \ell^{2d} k^d 
F_{\mu_1,r_0}(k-1)G_{\tilde{D}}(\ell)\cdot {\rm Ind}[{\rm dist}(x,z)\le k]\cdot{\rm Ind}[n<2k]\nonumber\\
&\le& (\kappa_d)^3 \hat{K}_{t,h}'\sum_{\ell,k} \ell^{2d} k^{2d} 
F_{\mu_1,r_0}(k-1)G_{\tilde{D}}(\ell)\cdot {\rm Ind}[n<2k]\le \tilde{F}_2(n), 
\end{eqnarray}
where $\tilde{F}_2(n)$ is a sub-exponentially decaying function, and the positive constant $\hat{K}_{t,h}'$ is given by 
$$
\hat{K}_{t,h}':=\sum_{u\in\ze^d}e^{-m_{t,h}{\rm dist}(0,u)}. 
$$
Combining this bound, (\ref{DeltakellXYZbound4}) and (\ref{DeltakellXYZbound5}), we obtain 
\begin{eqnarray}
\label{DeltakellXYZboundiiip}
& &\sum_{X,Y} \sum_{\ell,k} \Vert \Delta^{k,\ell}(X,Y,Z)\Vert\cdot
{\rm Ind}[{\rm supp}\; \Delta^{k,\ell}(X,Y,Z)\not\subset{Z}_n]\cdot{\rm Ind}[\mbox{Case~(iii)$'$ occurs}]\nonumber\\
& &\le C_\Delta^{{\rm (iii)}'}\Vert\tilde{\mathcal{H}}_{0,Z}\Vert \tilde{F}_2(n)
\end{eqnarray}
with a positive constant $C_\Delta^{{\rm (iii)}'}$
\smallskip 

Consider the remaining Case~(iv). Similarly, there exist $v,w\in Y$ such that ${\rm dist}(x,v)={\rm dist}(x,Y)\le \ell$ and 
${\rm dist}(w,y)={\rm dist}(y,Y)\le\ell$. In addition, there exist $u\in X$ and $u'\in Y$ such that 
${\rm dist}(u,u')={\rm dist}(u,Y)\le \ell$. Therefore, 
\begin{eqnarray*}
n&<&{\rm dist}(z,\zeta)\le {\rm dist}(z,x)+{\rm dist}(x,v)+{\rm dist}(v,w)+{\rm dist}(w,y)+{\rm dist}(y,\zeta)\\
&\le& 2k+2\ell+{\rm dist}(v,w).
\end{eqnarray*}
In the same way, we have 
\begin{eqnarray}
\label{DeltakellXYZbound6}
& &\sum_{X,Y} \sum_{\ell,k} \Vert \Delta^{k,\ell}(X,Y,Z)\Vert\cdot
{\rm Ind}[{\rm supp}\; \Delta^{k,\ell}(X,Y,Z)\not\subset{Z}_n]\cdot{\rm Ind}[\mbox{Case~(iv) occurs}]\nonumber\\
&\le&\sum_{z,x,v,w,y,\zeta}\sum_{Y\ni v,w}\sum_{u,u'}\sum_{X\ni u}
\sum_{\ell,k} \Vert \Delta^{k,\ell}(X,Y,Z)\Vert\cdot{\rm Ind}[z\in Z]\nonumber\\
&\times&{\rm Ind}[{\rm dist}(z,x)\le k]\cdot{\rm Ind}[{\rm dist}(x,v)\le \ell]\cdot{\rm Ind}[{\rm dist}(w,y)\le \ell]\cdot
{\rm Ind}[{\rm dist}(y,\zeta)\le k]\nonumber\\
&\times&{\rm Ind}[u'\in Y]\cdot{\rm Ind}[{\rm dist}(u,u')\le \ell]\cdot{\rm Ind}[n<2k+2\ell+{\rm dist}(v,w)]\nonumber\\
&\le&8\tilde{C}_{\tilde{U}}C_{\tilde{D}}\Vert\tilde{\mathcal{H}}_{0,Z}\Vert
\sum_{z,x,v,w,y,\zeta}\sum_{Y\ni v,w}\sum_{u,u'}\sum_{X\ni u}\sum_{\ell,k}\;
\ell^d\Vert \tilde{\mathcal{H}}_{0,X}\Vert\cdot|Y|^2
\Vert \tilde{V}_Y\Vert F_{\mu_1,r_0}(k-1)G_{\tilde{D}}(\ell)\nonumber\\
&\times&{\rm Ind}[z\in Z]\cdot {\rm Ind}[{\rm dist}(z,x)\le k]\cdot{\rm Ind}[{\rm dist}(x,v)\le \ell]\cdot
{\rm Ind}[{\rm dist}(w,y)\le \ell]\nonumber\\
&\times&{\rm Ind}[{\rm dist}(y,\zeta)\le k]\cdot{\rm Ind}[u'\in Y]\cdot{\rm Ind}[{\rm dist}(u,u')\le \ell]\cdot
{\rm Ind}[n<2k+2\ell+{\rm dist}(v,w)]. \nonumber\\
\end{eqnarray}
Note that
\begin{eqnarray*}
& &\sum_{u,u'}\sum_{X\ni u}{\rm Ind}[u'\in Y]\cdot{\rm Ind}[{\rm dist}(u,u')\le \ell]\cdot\Vert \tilde{\mathcal{H}}_{0,X}\Vert\\
&\le&C_{t,h} {K}_{t,h}\sum_{u,u'}{\rm Ind}[u'\in Y]\cdot{\rm Ind}[{\rm dist}(u,u')\le \ell]\\
&\le&C_{t,h}{K}_{t,h}\kappa_d\ell^d\sum_{u'}{\rm Ind}[u'\in Y]=C_{t,h}{K}_{t,h}\kappa_d\ell^d|Y|, 
\end{eqnarray*}
where we have used (\ref{boundednessH0Zx}). 
Using this estimate, the quantity in the right-hand side of (\ref{DeltakellXYZbound6}) can be evaluated as 
\begin{eqnarray}
\label{DeltakellXYZbound7}
& &\sum_{x,v,w,y,\zeta}\sum_{Y\ni v,w}\sum_{u,u'}\sum_{X\ni u}\sum_{\ell,k}
\ell^d\cdot\Vert \tilde{\mathcal{H}}_{0,X}\Vert\cdot|Y|^2
\Vert \tilde{V}_Y\Vert F_{\mu_1,r_0}(k-1)G_{\tilde{D}}(\ell)\nonumber\\
&\times& {\rm Ind}[{\rm dist}(z,x)\le k]\cdot{\rm Ind}[{\rm dist}(x,v)\le \ell]\cdot
{\rm Ind}[{\rm dist}(w,y)\le \ell]\nonumber\\
&\times&{\rm Ind}[{\rm dist}(y,\zeta)\le k]\cdot{\rm Ind}[u'\in Y]\cdot{\rm Ind}[{\rm dist}(u,u')\le \ell]\cdot
{\rm Ind}[n<2k+2\ell+{\rm dist}(v,w)]\nonumber\\
&\le& C_{t,h}{K}_{t,h}\kappa_d \sum_{x,v,w,y,\zeta}\sum_{Y\ni v,w}\sum_{\ell,k}
\ell^{2d}\cdot|Y|^3\Vert \tilde{V}_Y\Vert F_{\mu_1,r_0}(k-1)G_{\tilde{D}}(\ell)\nonumber\\
&\times&{\rm Ind}[{\rm dist}(z,x)\le k]\cdot{\rm Ind}[{\rm dist}(x,v)\le \ell]\cdot
{\rm Ind}[{\rm dist}(w,y)\le \ell]\cdot{\rm Ind}[{\rm dist}(y,\zeta)\le k]\nonumber\\
&\times&{\rm Ind}[n<2k+2\ell+{\rm dist}(v,w)]\nonumber\\
&\le& C_{t,h}{K}_{t,h}\kappa_d \sum_{x,v,w,y,\zeta}\sum_{\ell,k}\ell^{2d}e^{-\tilde{m}_{h,V}{\rm dist}(v,w)}
F_{\mu_1,r_0}(k-1)G_{\tilde{D}}(\ell)\nonumber\\
&\times&{\rm Ind}[{\rm dist}(z,x)\le k]\cdot{\rm Ind}[{\rm dist}(x,v)\le \ell]\cdot
{\rm Ind}[{\rm dist}(w,y)\le \ell]\cdot{\rm Ind}[{\rm dist}(y,\zeta)\le k]\nonumber\\
&\times&{\rm Ind}[n<2k+2\ell+{\rm dist}(v,w)]\nonumber\\
&\le& C_{t,h}{K}_{t,h}(\kappa_d)^3 \sum_{x,v,w}\sum_{\ell,k}\ell^{3d}k^d e^{-\tilde{m}_{h,V}{\rm dist}(v,w)}
F_{\mu_1,r_0}(k-1)G_{\tilde{D}}(\ell)\nonumber\\
&\times&{\rm Ind}[{\rm dist}(z,x)\le k]\cdot{\rm Ind}[{\rm dist}(x,v)\le \ell]\cdot{\rm Ind}[n<2k+2\ell+{\rm dist}(v,w)].
\end{eqnarray}
Similarly to Case~(i), we have 
$$
e^{-\tilde{m}_{h,V}{\rm dist}(v,w)/2}F_{\mu_1,r_0}(k-1)G_{\tilde{D}}(\ell)
\le \tilde{F}_3(k)\tilde{G}_3(\ell)\tilde{G}_4(n)
$$
with sub-exponentially decaying functions, $\tilde{F}_3(k)$, $\tilde{G}_3(\ell)$ and $\tilde{G}_4(n)$, 
under the condition $n<2k+2\ell+{\rm dist}(v,w)$. The quantity in the above right-hand side can be estimated as 
\begin{eqnarray*}
& &\sum_{x,v,w}\sum_{\ell,k}\ell^{3d}k^d e^{-\tilde{m}_{h,V}{\rm dist}(v,w)}
F_{\mu_1,r_0}(k-1)G_{\tilde{D}}(\ell)\nonumber\\
&\times&{\rm Ind}[{\rm dist}(z,x)\le k]\cdot{\rm Ind}[{\rm dist}(x,v)\le \ell]\cdot{\rm Ind}[n<2k+2\ell+{\rm dist}(v,w)]\\
&\le& \sum_{x,v,w}\sum_{\ell,k}\ell^{3d}k^d e^{-\tilde{m}_{h,V}{\rm dist}(v,w)/2}\tilde{F}_3(k)\tilde{G}_3(\ell)\tilde{G}_4(n)
{\rm Ind}[{\rm dist}(z,x)\le k]\cdot{\rm Ind}[{\rm dist}(x,v)\le \ell]\\
&\le& \hat{K}_{h,V} \sum_{x,v}\sum_{\ell,k}\ell^{3d}k^d \tilde{F}_3(k)\tilde{G}_3(\ell)\tilde{G}_4(n)\cdot
{\rm Ind}[{\rm dist}(z,x)\le k]\cdot{\rm Ind}[{\rm dist}(x,v)\le \ell]\\
&\le& \hat{K}_{h,V} \kappa_d \sum_{x}\sum_{\ell,k}\ell^{4d}k^d \tilde{F}_3(k)\tilde{G}_3(\ell)\tilde{G}_4(n)\cdot
{\rm Ind}[{\rm dist}(z,x)\le k]\\
&\le& \hat{K}_{h,V} (\kappa_d)^2 \sum_{\ell,k}\ell^{4d}k^{2d} \tilde{F}_3(k)\tilde{G}_3(\ell)\tilde{G}_4(n)
\le \hat{K}_{h,V}'\tilde{G}_4(n), 
\end{eqnarray*}
where the positive constant $\hat{K}_{h,V}$ is given by  (\ref{hatKhV}), and 
\begin{equation}
\hat{K}_{h,V}'=\hat{K}_{h,V} (\kappa_d)^2 \sum_{\ell,k}\ell^{4d}k^{2d} \tilde{F}_3(k)\tilde{G}_3(\ell).
\end{equation}
Combining this, (\ref{DeltakellXYZbound6}) and (\ref{DeltakellXYZbound7}), we obtain 
\begin{eqnarray}
\label{DeltakellXYZboundiv}
& &\sum_{X,Y} \sum_{\ell,k} \Vert \Delta^{k,\ell}(X,Y,Z)\Vert\cdot
{\rm Ind}[{\rm supp}\; \Delta^{k,\ell}(X,Y,Z)\not\subset{Z}_n]\cdot{\rm Ind}[\mbox{Case~(iv) occurs}]\nonumber\\
& &\le C_\Delta^{\rm (iv)}\Vert\tilde{\mathcal{H}}_{0,Z}\Vert \tilde{G}_4(n)
\end{eqnarray}
with the positive constant $C_\Delta^{\rm (iv)}$. 

Consequently, we obtain 
\begin{equation}
\label{commuUH0DUH0Znorm}
\Vert [\tilde{U}^\ast(r)[\tilde{\mathcal{H}}_0,\tilde{D}]\tilde{U}(r),\tilde{\mathcal{H}}_{0,Z}]
-[\tilde{U}^\ast(r)[\tilde{\mathcal{H}}_0,\tilde{D}]\tilde{U}(r),\tilde{\mathcal{H}}_{0,Z}]_{Z_n}\Vert
\le \Vert\tilde{\mathcal{H}}_{0,Z}\Vert\mathcal{F}_1(n)
\end{equation}
{from} (\ref{UH_0DUH0dif}), (\ref{DeltakellXYZboundi}), (\ref{DeltakellXYZboundiiip}) and (\ref{DeltakellXYZboundiv}), 
where $\mathcal{F}_1$ is the sub-exponentially decaying function.
\medskip 

In the same way, we can estimate the norm of (\ref{UH_0DUH0Zn}) as follows: 
When the summand in the right-hand side of (\ref{UH_0DUH0Zn}) is non-vanishing, there exists $z\in Z\cap (X\cup Y_\ell)_k$. 
Therefore, there exists $x\in X\cup Y_\ell$ such that ${\rm dist}(x,z)={\rm dist}(x,Z)\le k$. 
Clearly, the following two cases are possible: 
\begin{description}
\item{(a)}\quad $x\in X$
\item{(b)}\quad $x\in Y_\ell$
\end{description}
In both of the cases, there exists $u\in X\cap Y_\ell$ when $\Delta^{k,\ell}(X,Y,Z)$ is non-vanishing. 
Therefore, there exists $v\in Y$ such that ${\rm dist}(u,v)={\rm dist}(u,Y)\le \ell$.  

Consider first Case~(a). It is enough to estimate  
\begin{equation}
\sum_{z\in Z}\sum_{x,u,v}
\sum_{X\ni x,u}\sum_{Y\ni v} \sum_{\ell,k}\Vert\Delta^{k,\ell}(X,Y,Z)\Vert\cdot{\rm Ind}[{\rm dist}(z,x)\le k]\cdot
{\rm Ind}[{\rm dist}(u,v)\le \ell].
\end{equation}
For this purpose, we recall the inequality (\ref{normDeltakellXYZ}),   
\begin{equation*}
\Vert\Delta^{k,\ell}(X,Y,Z)\Vert\le 
8\tilde{C}_{\tilde{U}}C_{\tilde{D}}\kappa_d\Vert\tilde{\mathcal{H}}_{0,Z}\Vert\cdot\ell^d\cdot 
\Vert \tilde{\mathcal{H}}_{0,X}\Vert\cdot|Y|^2\Vert \tilde{V}_Y\Vert \cdot F_{\mu_1,r_0}(k-1)G_{\tilde{D}}(\ell). 
\end{equation*}
Therefore, it is sufficient to estimate 
\begin{eqnarray}
J_a&:=&\sum_{x,u,v}
\sum_{X\ni x,u}\sum_{Y\ni v} \sum_{\ell,k}\; \ell^d\cdot 
\Vert \tilde{\mathcal{H}}_{0,X}\Vert\cdot|Y|^2\Vert \tilde{V}_Y\Vert \cdot F_{\mu_1,r_0}(k-1)G_{\tilde{D}}(\ell)\nonumber\\
& &\times{\rm Ind}[{\rm dist}(z,x)\le k]\cdot{\rm Ind}[{\rm dist}(u,v)\le \ell].
\end{eqnarray}
By using (\ref{boundednessH0Zxy}) and (\ref{boundedtildeVX2}), we have 
\begin{eqnarray}
J_a&\le&C_{t,h}K_{t,h}\tilde{C}_{h,V}\sum_{x,u,v}\sum_{\ell,k}\; \ell^d\cdot
e^{-m_{t,h}{\rm dist}(x,u)}\cdot F_{\mu_1,r_0}(k-1)G_{\tilde{D}}(\ell)
\nonumber\\& &
\times{\rm Ind}[{\rm dist}(z,x)\le k]\cdot{\rm Ind}[{\rm dist}(u,v)\le \ell]\nonumber\\
&\le&C_{t,h}K_{t,h}\tilde{C}_{h,V}\kappa_d\sum_{x,u}\sum_{\ell,k}\; \ell^{2d}\cdot
e^{-m_{t,h}{\rm dist}(x,u)}\cdot F_{\mu_1,r_0}(k-1)G_{\tilde{D}}(\ell)\cdot
{\rm Ind}[{\rm dist}(z,x)\le k]\nonumber\\
&\le&C_{t,h}(K_{t,h})^2\tilde{C}_{h,V}\kappa_d \sum_{x}\sum_{\ell,k}\; \ell^{2d}\cdot F_{\mu_1,r_0}(k-1)G_{\tilde{D}}(\ell)\cdot
{\rm Ind}[{\rm dist}(z,x)\le k]\nonumber\\
&\le&C_{t,h}(K_{t,h})^2\tilde{C}_{h,V}(\kappa_d)^2 \sum_{\ell,k}\; \ell^{2d}\cdot k^d\cdot 
F_{\mu_1,r_0}(k-1)G_{\tilde{D}}(\ell)<\infty.
\end{eqnarray}
{From} these observations, we obtain 
\begin{equation}
\label{DeltakellXYZCaseabound}
\sum_{X,Y} \sum_{\ell,k}\Vert\Delta^{k,\ell}(X,Y,Z)\Vert\cdot
{\rm Ind}[\mbox{Case (a) occurs}]\le C_a \Vert \tilde{\mathcal{H}}_{0,Z}\Vert.
\end{equation}
with the positive constant $C_a$. 

Next, consider Case~(b). In this case, since $x\in Y_\ell$, 
there exists $w\in Y$ such that ${\rm dist}(x,w)={\rm dist}(x,Y)\le \ell$. Therefore, it is sufficient to estimate 
\begin{equation}
\sum_{z\in Z}\sum_{x,u,v,w}
\sum_{X\ni u}\sum_{Y\ni v,w} \sum_{\ell,k}\Vert\Delta^{k,\ell}(X,Y,Z)\Vert
\cdot{\rm Ind}[{\rm dist}(z,x)\le k,{\rm dist}(x,w)\le\ell,{\rm dist}(u,v)\le \ell].
\end{equation}
Similarly to Case~(a), we have 
\begin{eqnarray}
& &\sum_{x,u,v,w}
\sum_{X\ni u}\sum_{Y\ni v,w} \sum_{\ell,k}\; \ell^d\cdot 
\Vert \tilde{\mathcal{H}}_{0,X}\Vert\cdot|Y|^2\Vert \tilde{V}_Y\Vert \cdot F_{\mu_1,r_0}(k-1)G_{\tilde{D}}(\ell)\nonumber\\
& &\times{\rm Ind}[{\rm dist}(z,x)\le k,{\rm dist}(x,w)\le\ell]\cdot{\rm Ind}[{\rm dist}(u,v)\le \ell]\nonumber\\
&\le& C_{t,h}K_{t,h}\tilde{C}_{h,V}\sum_{x,u,v,w}\sum_{\ell,k}\; \ell^d\cdot e^{-\tilde{m}_{h,V}{\rm dist}(v,w)}\cdot
F_{\mu_1,r_0}(k-1)G_{\tilde{D}}(\ell)\nonumber\\
& &\times{\rm Ind}[{\rm dist}(z,x)\le k,{\rm dist}(x,w)\le\ell]\cdot{\rm Ind}[{\rm dist}(u,v)\le \ell]\nonumber\\
&\le&C_{t,h}K_{t,h}\tilde{C}_{h,V}\kappa_d \sum_{x,v,w}\sum_{\ell,k}\; \ell^{2d}\cdot e^{-\tilde{m}_{h,V}{\rm dist}(v,w)}\cdot
F_{\mu_1,r_0}(k-1)G_{\tilde{D}}(\ell)\nonumber\\
& &\times{\rm Ind}[{\rm dist}(z,x)\le k,{\rm dist}(x,w)\le\ell]\nonumber\\
&\le&C_{t,h}K_{t,h}\tilde{C}_{h,V}\kappa_d \hat{K}_{h,V}'' \sum_{x,w}\sum_{\ell,k}\; \ell^{2d}\cdot 
F_{\mu_1,r_0}(k-1)G_{\tilde{D}}(\ell)
\nonumber\\& &\times
{\rm Ind}[{\rm dist}(z,x)\le k,{\rm dist}(x,w)\le\ell]\nonumber\\
&\le&C_{t,h}K_{t,h}\tilde{C}_{h,V}(\kappa_d)^3 \hat{K}_{h,V}'' \sum_{\ell,k}\;
 \ell^{3d}\cdot k^d \cdot F_{\mu_1,r_0}(k-1)G_{\tilde{D}}(\ell)<\infty, 
\end{eqnarray}
where 
$$
\hat{K}_{h,V}'':=\sum_{v\in\ze^d}e^{-\tilde{m}_{h,V}{\rm dist}(v,0)}. 
$$
As a result, we obtain 
$$ 
\sum_{X,Y} \sum_{\ell,k}\Vert\Delta^{k,\ell}(X,Y,Z)\Vert\cdot{\rm Ind}[\mbox{Case~(b) occurs}]
\le C_b\Vert\tilde{\mathcal{H}}_{0,Z}\Vert
$$
with the positive constant $C_b$. Combining this, (\ref{UH_0DUH0Zn}) and (\ref{DeltakellXYZCaseabound}), we obtain 
\begin{equation}
\label{normcommuUH0DUH0ZZn}
\Vert[\tilde{U}^\ast(r)[\tilde{\mathcal{H}}_0,\tilde{D}]\tilde{U}(r),\tilde{\mathcal{H}}_{0,Z}]_{Z_n}\Vert
\le (C_a+C_b)\Vert\tilde{\mathcal{H}}_{0,Z}\Vert. 
\end{equation}

\Section{Local approximation of $[\tilde{U}^\ast(s)\tilde{V}\tilde{U}(s),\tilde{\mathcal{H}}_{0,Z}]$}

For the purpose of dealing with $\tilde{\mathcal{W}}_Z^{(3)}(s)$, in this appendix, we consider local approximation of 
the other operator, $[\tilde{U}^\ast(s)\tilde{V}\tilde{U}(s),\tilde{\mathcal{H}}_{0,Z}]$,   
in the expression (\ref{tildecalWZ3s}) of $\tilde{\mathcal{W}}_Z^{(3)}(s)$ with $\tilde{\mathcal{M}}_{t,Z}(s)$ 
of (\ref{tildecalMtZs}).

To begin with, we recall  
$$
\tilde{U}^\ast(s)\tilde{V}\tilde{U}(s)=\sum_X \sum_m \Delta_{\tilde{U}}^m(\tilde{V}_X,s).
$$
We define 
$$
[\tilde{U}^\ast(s)\tilde{V}\tilde{U}(s),\tilde{\mathcal{H}}_{0,Z}]_{Z_n}
:=\sum_X \sum_m \;[\Delta_{\tilde{U}}^m(\tilde{V}_X,s),\tilde{\mathcal{H}}_{0,Z}]\cdot{\rm Ind}[X_m\subset Z_n].
$$
Then, one has 
\begin{equation}
\Vert [\tilde{U}^\ast(s)\tilde{V}\tilde{U}(s),\tilde{\mathcal{H}}_{0,Z}]
- [\tilde{U}^\ast(s)\tilde{V}\tilde{U}(s),\tilde{\mathcal{H}}_{0,Z}]_{Z_n}\Vert
\le \sum_X \sum_m \Vert[\Delta_{\tilde{U}}^m(\tilde{V}_X,s),\tilde{\mathcal{H}}_{0,Z}]\Vert{\rm Ind}[X_m\not\subset Z_n].  
\end{equation}
When the summand in the right-hand side is non-vanishing, there exists $z\in Z\cap X_m$. 
In addition, for such a site $z$, there exists $x\in X$ such that ${\rm dist}(z,x)={\rm dist}(z,X)\le m$. 
{From} the condition $X_m\not\subset Z_n$, there exists $\zeta\in X_n\backslash Z_n$. 
Clearly, ${\rm dist}(z,\zeta)>n$. For this site $\zeta$, 
there exists $y\in X$ such that ${\rm dist}(\zeta,y)={\rm dist}(\zeta,X)\le m$. From these observations, one has 
$$ 
n<{\rm dist}(z,\zeta)\le {\rm dist}(z,x)+{\rm dist}(x,y)+{\rm dist}(y,\zeta)
\le 2m+{\rm dist}(x,y).
$$
By using these observations, (\ref{Deltanbound2}) and (\ref{boundedtildeVX}), we have 
\begin{eqnarray}
\label{commuUVUH0Zbound}
& &\Vert [\tilde{U}^\ast(s)\tilde{V}\tilde{U}(s),\tilde{\mathcal{H}}_{0,Z}]
- [\tilde{U}^\ast(s)\tilde{V}\tilde{U}(s),\tilde{\mathcal{H}}_{0,Z}]_{Z_n}\Vert\nonumber\\
&\le& 2\tilde{C}_{\tilde{U}}\Vert \tilde{\mathcal{H}}_{0,Z}\Vert \sum_{z\in Z}\sum_{x,y,\zeta} \sum_{X\ni x,y} \sum_m
|X|\Vert\tilde{V}_X\Vert\cdot F_{\mu_1,r_0}(m-1)\nonumber\\
& &\times{\rm Ind}[{\rm dist}(z,x)\le m,{\rm dist}(\zeta,y)\le m, n<2m+{\rm dist}(x,y)]\nonumber\\
&\le& 2\tilde{C}_{\tilde{U}}\tilde{C}_{h,V}\Vert \tilde{\mathcal{H}}_{0,Z}\Vert\sum_{z\in Z}\sum_{x,y,\zeta} \sum_m
e^{-\tilde{m}_{h,V}{\rm dist}(x,y)}\cdot F_{\mu_1,r_0}(m-1)\nonumber\\
& &\times{\rm Ind}[{\rm dist}(z,x)\le m,{\rm dist}(\zeta,y)\le m, n<2m+{\rm dist}(x,y)].
\end{eqnarray}
In order to estimate the right-hand side, we use the inequality, 
\begin{equation}
e^{-\tilde{m}_{h,V}{\rm dist}(x,y)/2}F_{\mu_1,r_0}(m-1)\le \tilde{F}_4(m)\tilde{G}_5(n)
\end{equation}
under the condition $n<2m+{\rm dist}(x,y)$. Here, $\tilde{F}_4(m)$ and $\tilde{G}_5(n)$ are sub-exponentially 
decaying functions. This type of inequality was already used as (\ref{additiveinequ}). As a result, we have 
\begin{eqnarray}
& &\sum_{x,y,\zeta} \sum_m
e^{-\tilde{m}_{h,V}{\rm dist}(x,y)}\cdot F_{\mu_1,r_0}(m-1)
\nonumber\\& &\times
{\rm Ind}[{\rm dist}(z,x)\le m,{\rm dist}(\zeta,y)\le m, n<2m+{\rm dist}(x,y)]\nonumber\\
&\le&\sum_{x,y,\zeta} \sum_m
e^{-\tilde{m}_{h,V}{\rm dist}(x,y)/2}\cdot \tilde{F}_4(m)\tilde{G}_5(n)\cdot
{\rm Ind}[{\rm dist}(z,x)\le m,{\rm dist}(\zeta,y)\le m]\nonumber\\
&\le& (\kappa_d)^2\hat{K}_{h,V}\left[\sum_m m^{2d}\cdot \tilde{F}_4(m)\right]\tilde{G}_5(n), 
\end{eqnarray}
where the positive constant $\hat{K}_{h,V}$ is given by  (\ref{hatKhV}).  
Substituting this into the right-hand side of (\ref{commuUVUH0Zbound}), we obtain 
\begin{equation}
\label{commuUVUH0Znorm}
\Vert [\tilde{U}^\ast(s)\tilde{V}\tilde{U}(s),\tilde{\mathcal{H}}_{0,Z}]
- [\tilde{U}^\ast(s)\tilde{V}\tilde{U}(s),\tilde{\mathcal{H}}_{0,Z}]_{Z_n}\Vert\le \Vert\tilde{\mathcal{H}}_{0,Z}\Vert\mathcal{F}_2(n)
\end{equation}
with the sub-exponentially decaying function $\mathcal{F}_2(n)$. 

In the same way, we have 
\begin{eqnarray}
\label{normcommuUVUH0ZZn}
& &\Vert[\tilde{U}^\ast(s)\tilde{V}\tilde{U}(s),\tilde{\mathcal{H}}_{0,Z}]_{Z_n}\Vert\nonumber\\
&\le& 2\sum_{z\in Z}\sum_x \sum_{X\ni x}\sum_m \Vert \tilde{\mathcal{H}}_{0,Z}\Vert\cdot\Vert \Delta_{\tilde{U}}^m(\tilde{V}_X,s)\Vert
\cdot{\rm Ind}[{\rm dist}(z,x)\le m] \nonumber\\
&\le& 2\tilde{C}_{\tilde{U}}\Vert \tilde{\mathcal{H}}_{0,Z}\Vert\sum_{z\in Z}\sum_x \sum_{X\ni x}\sum_m|X|
\Vert\tilde{V}_X\Vert \cdot F_{\mu_1,r_0}(m-1)\cdot{\rm Ind}[{\rm dist}(z,x)\le m] \nonumber\\
&\le&2\tilde{C}_{\tilde{U}}\tilde{C}_{h,V}\Vert \tilde{\mathcal{H}}_{0,Z}\Vert\sum_{z\in Z}\sum_x\sum_m
F_{\mu_1,r_0}(m-1)\cdot{\rm Ind}[{\rm dist}(z,x)\le m] \nonumber\\
&\le& 4\tilde{C}_{\tilde{U}}\tilde{C}_{h,V}\kappa_d \Vert \tilde{\mathcal{H}}_{0,Z}\Vert
\sum_m m^d F_{\mu_1,r_0}(m-1)\le C_{[\tilde{U}\tilde{V}\tilde{U},\tilde{\mathcal{H}}_{0,Z}]}\Vert \tilde{\mathcal{H}}_{0,Z}\Vert 
\end{eqnarray}
with the positive constant $C_{[\tilde{U}\tilde{V}\tilde{U},\tilde{\mathcal{H}}_{0,Z}]}$. 

\Section{Local boundedness of the interaction $\tilde{\mathcal{W}}^{(3)}(s)$}
\label{Sec:BoundtildecalWZ3s}

Similarly to the previous two cases of $\tilde{\mathcal{W}}_Z^{(1)}(s)$ and $\tilde{\mathcal{W}}_Z^{(2)}(s)$, 
it is enough to prove the inequality (\ref{diftildecalWZ3sPi}) below. 

We write 
\begin{equation}
\tilde{\mathcal{L}}_Z(s):=is[\tilde{U}^\ast(s)\tilde{V}\tilde{U}(s),\tilde{\mathcal{H}}_{0,Z}]
-\int_0^s dr [\tilde{U}^\ast(r)[\tilde{\mathcal{H}}_0,\tilde{D}(r)]\tilde{U}(r),\tilde{\mathcal{H}}_{0,Z}]
\end{equation}
for short. Then, the interaction $\tilde{\mathcal{W}}_Z^{(3)}(s)$ is written 
\begin{equation}
\tilde{\mathcal{W}}_Z^{(3)}(s)=\int_{-\infty}^{+\infty}dt\; w_\gamma(t)\int_0^t dr\; \tilde{\tau}_{s,r}(\tilde{\mathcal{L}}_Z(s))
\end{equation}
{from} (\ref{tildecalMtZs}) and (\ref{tildecalWZ3s}). Therefore, one has 
\begin{equation}
\tilde{\mathcal{W}}_Z^{(3)}(s)-\Pi_{\mathcal{Z}_n}(\tilde{\mathcal{W}}_Z^{(3)}(s))
=\int_{-\infty}^{+\infty}dt\; w_\gamma(t)\int_0^t dr\; \left[\tilde{\tau}_{s,r}(\tilde{\mathcal{L}}_Z(s))
-\Pi_{\mathcal{Z}_n}(\tilde{\tau}_{s,r}(\tilde{\mathcal{L}}_Z(s)))\right].
\end{equation}
The norm is estimated by 
\begin{equation}
\label{diftildecalWZ3s}
\Vert\tilde{\mathcal{W}}_Z^{(3)}(s)-\Pi_{\mathcal{Z}_n}(\tilde{\mathcal{W}}_Z^{(3)}(s))\Vert
\le \int_{-\infty}^{+\infty}dt\; w_\gamma(t)\int_0^{|t|} dr\; \left\Vert\tilde{\tau}_{s,r}(\tilde{\mathcal{L}}_Z(s))
-\Pi_{\mathcal{Z}_n}(\tilde{\tau}_{s,r}(\tilde{\mathcal{L}}_Z(s)))\right\Vert.
\end{equation}
In order to estimate the right-hand side, we introduce   
\begin{equation}
\tilde{\mathcal{L}}_{Z,Z_n}(s):=is[\tilde{U}^\ast(s)\tilde{V}\tilde{U}(s),\tilde{\mathcal{H}}_{0,Z}]_{Z_n}
-\int_0^s dr [\tilde{U}^\ast(r)[\tilde{\mathcal{H}}_0,\tilde{D}(r)]\tilde{U}(r),\tilde{\mathcal{H}}_{0,Z}]_{Z_n}. 
\end{equation}
By using (\ref{commuUH0DUH0Znorm}) and (\ref{commuUVUH0Znorm}), we have 
\begin{eqnarray}
\label{tildecalLZsdifnorm}
& &\left\Vert\tilde{\mathcal{L}}_Z(s)- \tilde{\mathcal{L}}_{Z,Z_n}(s)\right\Vert\nonumber\\
&\le& |s|\left\Vert [\tilde{U}^\ast(s)\tilde{V}\tilde{U}(s),\tilde{\mathcal{H}}_{0,Z}]
- [\tilde{U}^\ast(s)\tilde{V}\tilde{U}(s),\tilde{\mathcal{H}}_{0,Z}]_{Z_n}\right\Vert\nonumber\\
&+&\int_0^{|s|}dr \left\Vert[\tilde{U}^\ast(r)[\tilde{\mathcal{H}}_0,\tilde{D}(r)]\tilde{U}(r),\tilde{\mathcal{H}}_{0,Z}]
- [\tilde{U}^\ast(r)[\tilde{\mathcal{H}}_0,\tilde{D}(r)]\tilde{U}(r),\tilde{\mathcal{H}}_{0,Z}]_{Z_n} \right\Vert\nonumber\\
&\le&|s|\Vert\tilde{\mathcal{H}}_{0,Z}\Vert[\mathcal{F}_1(n)+\mathcal{F}_2(n)].
\end{eqnarray}
Further, by using (\ref{normcommuUH0DUH0ZZn}) and (\ref{normcommuUVUH0ZZn}), we obtain 
\begin{eqnarray}
\label{normtildecalLZZns}
\Vert\tilde{\mathcal{L}}_{Z,Z_n}(s)\Vert&\le& |s|\Vert[\tilde{U}^\ast(s)\tilde{V}\tilde{U}(s),\tilde{\mathcal{H}}_{0,Z}]_{Z_n}\Vert
+\int_0^{|s|}dr \Vert [\tilde{U}^\ast(r)[\tilde{\mathcal{H}}_0,\tilde{D}(r)]\tilde{U}(r),\tilde{\mathcal{H}}_{0,Z}]_{Z_n}\Vert
\nonumber\\
&\le&(C_a+C_b+C_{[\tilde{U}\tilde{V}\tilde{U},\tilde{\mathcal{H}}_{0,Z}]})|s|\cdot\Vert \tilde{\mathcal{H}}_{0,Z}\Vert. 
\end{eqnarray}

Note that 
\begin{eqnarray}
\label{tildetautildecalLZsPiZndif}
\left\Vert \tilde{\tau}_{s,t'}(\tilde{\mathcal{L}}_Z(s))-\Pi_{\mathcal{Z}_n}(\tilde{\tau}_{s,t'}(\tilde{\mathcal{L}}_Z(s)))
\right\Vert&\le&\left\Vert \tilde{\tau}_{s,t'}(\tilde{\mathcal{L}}_Z(s))-\tilde{\tau}_{s,t'}(\tilde{\mathcal{L}}_{Z,Z_n}(s))\right\Vert
\nonumber\\
&+&\left\Vert \tilde{\tau}_{s,t'}(\tilde{\mathcal{L}}_{Z,Z_n}(s))
- \Pi_{\mathcal{Z}_n}(\tilde{\tau}_{s,t'}(\tilde{\mathcal{L}}_{Z,Z_n}(s)))\right\Vert\nonumber\\
&+&\left\Vert \Pi_{\mathcal{Z}_n}(\tilde{\tau}_{s,t'}(\tilde{\mathcal{L}}_{Z,Z_n}(s)))-
\Pi_{\mathcal{Z}_n}(\tilde{\tau}_{s,t'}(\tilde{\mathcal{L}}_{Z}(s)))\right\Vert.\nonumber\\
\end{eqnarray}
As to the first and third terms in the right-hand side, we have 
\begin{eqnarray}
\left\Vert \tilde{\tau}_{s,t'}(\tilde{\mathcal{L}}_Z(s))-\tilde{\tau}_{s,t'}(\tilde{\mathcal{L}}_{Z,Z_n}(s))\right\Vert
&\le& \left\Vert \tilde{\mathcal{L}}_Z(s)- \tilde{\mathcal{L}}_{Z,Z_n}(s)\right\Vert\nonumber\\
&\le& |s|\Vert\tilde{\mathcal{H}}_{0,Z}\Vert[\mathcal{F}_1(n)+\mathcal{F}_2(n)] 
\end{eqnarray}
and 
\begin{eqnarray}
\left\Vert \Pi_{\mathcal{Z}_n}(\tilde{\tau}_{s,t'}(\tilde{\mathcal{L}}_{Z,Z_n}(s)))-
\Pi_{\mathcal{Z}_n}(\tilde{\tau}_{s,t'}(\tilde{\mathcal{L}}_{Z}(s)))\right\Vert
&\le& \left\Vert \tilde{\mathcal{L}}_{Z,Z_n}(s)-\tilde{\mathcal{L}}_{Z}(s) \right\Vert\nonumber\\ 
&\le& |s|\Vert\tilde{\mathcal{H}}_{0,Z}\Vert[\mathcal{F}_1(n)+\mathcal{F}_2(n)], 
\end{eqnarray}
where we have used the inequality (\ref{tildecalLZsdifnorm}).
Substituting these into the right-hand side of (\ref{diftildecalWZ3s}), we have 
\begin{eqnarray}
\label{diftildecalWZ3Pi}
& &\Vert\tilde{\mathcal{W}}_Z^{(3)}(s)-\Pi_{\mathcal{Z}_n}(\tilde{\mathcal{W}}_Z^{(3)}(s))\Vert\nonumber\\
&\le& 2|s|\Vert\tilde{\mathcal{H}}_{0,Z}\Vert[\mathcal{F}_1(n)+\mathcal{F}_2(n)]
\int_{-\infty}^{+\infty}dt\; w_\gamma(t)|t|\nonumber\\ 
&+&\int_{-\infty}^{+\infty}dt\; w_\gamma(t)\int_0^{|t|} dr\left\Vert \tilde{\tau}_{s,r}(\tilde{\mathcal{L}}_{Z,Z_n}(s))
- \Pi_{\mathcal{Z}_n}(\tilde{\tau}_{s,r}(\tilde{\mathcal{L}}_{Z,Z_n}(s)))\right\Vert.\nonumber\\
\end{eqnarray}
The second term in the right-hand side is estimated as 
\begin{eqnarray}
\label{diftildecalWZ3Pi2}
& &\int_{-\infty}^{+\infty}dt\; w_\gamma(t)\int_0^{|t|} dr\;\left\Vert \tilde{\tau}_{s,r}(\tilde{\mathcal{L}}_{Z,Z_n}(s))
- \Pi_{\mathcal{Z}_n}(\tilde{\tau}_{s,r}(\tilde{\mathcal{L}}_{Z,Z_n}(s)))\right\Vert\nonumber\\
&\le& 4\Vert \tilde{\mathcal{L}}_{Z,Z_n}(s)\Vert \tilde{W}_\gamma(T)
+\int_{-T}^{+T}dt\; w_\gamma(t)\int_0^{|t|} dr\left\Vert \tilde{\tau}_{s,r}(\tilde{\mathcal{L}}_{Z,Z_n}(s))
- \Pi_{\mathcal{Z}_n}(\tilde{\tau}_{s,r}(\tilde{\mathcal{L}}_{Z,Z_n}(s)))\right\Vert\nonumber\\
\end{eqnarray}
for any positive $T$, where we have written 
\begin{equation}
\tilde{W}_\gamma(T):=\int_T^\infty dt\;w_\gamma(t)|t|.
\end{equation}
Further, in order to estimate the integrand in the second term in the right-hand side, 
we use the inequality (\ref{taudifPi}) that is obtained from 
the Lieb-Robinson bound. As result, we have 
\begin{eqnarray}
& &\left\Vert \tilde{\tau}_{s,r}(\tilde{\mathcal{L}}_{Z,Z_n}(s))
- \Pi_{\mathcal{Z}_n}(\tilde{\tau}_{s,r}(\tilde{\mathcal{L}}_{Z,Z_n}(s)))\right\Vert
\nonumber\\&\le& 
2\kappa_d C_{\tilde{U},\tilde{\mathcal{H}}}\mathcal{C}_3\Vert \tilde{\mathcal{L}}_{Z,Z_n}(s)\Vert \exp[v_{\rm LR}|r|]\cdot
n^d F_{\mu_3,2r_0}(n),
\end{eqnarray}
where we have also used the inequality (\ref{sumFmu2bound}) and $|Z_n|\le \kappa_d|Z|n^d=2\kappa_d n^d$. 
Note that
\begin{eqnarray}
\int_{-T}^Tdt\;w_\gamma(t)\int_0^{|t|}dr\; e^{v_{\rm LR}r}&\le&
\frac{2}{v_{\rm LR}}\int_{0}^Tdt\;w_\gamma(t)e^{v_{\rm LR}t}\nonumber\\
&\le&\frac{2c_\gamma}{v_{\rm LR}}\int_{0}^Tdt\;e^{v_{\rm LR}t}\le \frac{2c_\gamma}{(v_{\rm LR})^2}e^{v_{\rm LR}T},
\end{eqnarray}
where we have used the bound $w_\gamma(t)\le c_\gamma$ with a positive constant $c_\gamma$ for all $t$ \cite{BMNS}. 
These two inequalities yield 
\begin{eqnarray}
\label{diftildecalWZ3Pi22}
& &\int_{-T}^{+T}dt\; w_\gamma(t)\int_0^{|t|} dr\left\Vert \tilde{\tau}_{s,r}(\tilde{\mathcal{L}}_{Z,Z_n}(s))
- \Pi_{\mathcal{Z}_n}(\tilde{\tau}_{s,r}(\tilde{\mathcal{L}}_{Z,Z_n}(s)))\right\Vert\nonumber\\
&\le& \frac{4\kappa_d C_{\tilde{U},\tilde{\mathcal{H}}}\mathcal{C}_3c_\gamma}{(v_{\rm LR})^2}
\Vert \tilde{\mathcal{L}}_{Z,Z_n}(s)\Vert \cdot e^{v_{\rm LR}T}n^d F_{\mu_3,2r_0}(n). 
\end{eqnarray}
We take $T=T(n)$ which is the function of $n$ given by (\ref{T(n)}). Then, 
$$
\exp[v_{\rm LR}T(n)]\cdot F_{\mu_3,2r_0}(n)=F_{\mu_3/2,2r_0}(n)
$$
as in (\ref{FexpTnF}). By combining this, (\ref{diftildecalWZ3Pi}), (\ref{diftildecalWZ3Pi2}) and (\ref{diftildecalWZ3Pi22}), 
we obtain 
\begin{eqnarray}
\Vert\tilde{\mathcal{W}}_Z^{(3)}(s)-\Pi_{\mathcal{Z}_n}(\tilde{\mathcal{W}}_Z^{(3)}(s))\Vert
&\le& 4\tilde{W}_\gamma(0)|s|\Vert\tilde{\mathcal{H}}_{0,Z}\Vert[\mathcal{F}_1(n)+\mathcal{F}_2(n)]
+4\Vert \tilde{\mathcal{L}}_{Z,Z_n}(s)\Vert \tilde{W}_\gamma(T(n))\nonumber\\
&+&\frac{4\kappa_d C_{\tilde{U},\tilde{\mathcal{H}}}\mathcal{C}_3c_\gamma}{(v_{\rm LR})^2}
\Vert \tilde{\mathcal{L}}_{Z,Z_n}(s)\Vert \cdot n^d F_{\mu_3/2,2r_0}(n).
\end{eqnarray}
Since the quantity $\Vert \tilde{\mathcal{L}}_{Z,Z_n}(s)\Vert$ satisfies the inequality (\ref{normtildecalLZZns}), 
we obtain the desired bound, 
\begin{equation} 
\label{diftildecalWZ3sPi}
\Vert\tilde{\mathcal{W}}_Z^{(3)}(s)-\Pi_{\mathcal{Z}_n}(\tilde{\mathcal{W}}_Z^{(3)}(s))\Vert
\le |s|\Vert\tilde{\mathcal{H}}_{0,Z}\Vert\mathcal{F}_3(n),
\end{equation}
with the sub-exponentially decaying function $\mathcal{F}_3(n)$. 

\bigskip\bigskip\bigskip

\noindent
{\bf Acknowledgements:} I would like to thank Shu Nakamura for helpful discussions. 

\end{document}